\documentclass[submit]{smj}

\usepackage{amssymb,amsthm,algorithm,algorithmicx,longtable,xr}
\usepackage{booktabs}
\newtheorem{assump}{Assumption}

\newtheorem{theorem}{Theorem}
\newtheorem{prop}{Proposition}
\makeatletter
\newcommand*{\indep}{%
	\mathbin{%
		\mathpalette{\@indep}{}%
	}%
}
\newcommand*{\nindep}{%
	\mathbin{
		\mathpalette{\@indep}{\not}
	}%
}
\newcommand*{\@indep}[2]{%
	\sbox0{$#1\perp\m@th$}
	\sbox2{$#1=$}
	\sbox4{$#1\vcenter{}$}
	\rlap{\copy0}
	\dimen@=\dimexpr\ht2-\ht4-.2pt\relax
	\kern\dimen@
	{#2}%
	\kern\dimen@
	\copy0 
} 

\makeatletter
\newcommand*{\addFileDependency}[1]{
  \typeout{(#1)}
  \@addtofilelist{#1}
  \IfFileExists{#1}{}{\typeout{No file #1.}}
}
\makeatother

\newcommand*{\myexternaldocument}[1]{%
    \externaldocument{#1}%
    \addFileDependency{#1.tex}%
    \addFileDependency{#1.aux}%
}

\myexternaldocument{Supplementary}

\Author{ Thi Kim Hue Nguyen \Affil{1}, 
        Monica Chiogna\Affil{2}, Davide Risso\Affil{1}
        and Erika Banzato\Affil{1}
}
\AuthorRunning{N.T.K. Hue \textrm{et al.}}

\Affiliations{

\item Department of Statistical Sciences, 
      University of Padova,
      Padova,
      Italy

\item Department of Statistical Sciences ``Paolo Fortunati'', 
University of Bologna,
      Bologna,
      Italy

}   

\CorrAddress{Nguyen Thi Kim Hue, 
             Department of Statistical Sciences, University of Padova,
             Via Cesare Battisti, 241, 
             35121 Padova,
             Italy}
\CorrEmail{nguyen@stat.unipd.it}
\CorrPhone{(+39)\;0498274124}
\CorrFax{(+39)\;0498274170}

\Title{Guided structure learning of DAGs for count data}
\TitleRunning{Guided structure learning of DAGs for count data}

\Abstract{
In this paper, we tackle structure learning of Directed Acyclic Graphs (DAGs), with the idea of exploiting available prior knowledge of the domain at hand to guide the search of the best structure. In particular, we assume to know the topological ordering of variables in addition to the given data. We study a new algorithm for learning the structure of DAGs, proving its theoretical consistency in the limit of infinite observations. Furthermore, we experimentally compare the proposed algorithm to several popular competitors, to study its behaviour in finite samples. Biological validation of the algorithm is presented through the analysis of non-small cell lung cancer data.

}


\begin{document}

\maketitle


\section{Introduction}
\label{sec:intro}
Current demand for modelling complex interactions between variables, combined with the greater availability of high-dimensional discrete data, possibly showing a large number of zeros and measured on a small number of units, has led to an increased focus on structure learning for discrete data in high dimensional settings. {In some applications, directed graphs are preferable, as they translate naturally into domain-specific concepts. For instance, when modelling interactions among genes, an arrow between two nodes describes the direction of the information flow between the corresponding two genes. Hence, the problem often specializes in structure learning of DAGs.} 

Some solutions are nowadays available in the literature for learning (sparse) DAGs for discrete data. \cite{hadiji2015poisson} introduce a novel family of non-parametric Poisson graphical models, called Poisson Dependency Networks (PDN), trained using functional gradient ascent; \cite{park2015learning} define general Poisson DAG models which are identifiable from observational data, and present a polynomial-time algorithm that learns the Poisson DAG model under suitable regularity conditions.

The previously cited approaches, and, more broadly, typical approaches to structure learning of DAGs, usually assume no knowledge of the graph structure to be learned other than sparseness. However, in the context of learning gene networks, a wealth of information is available, usually stored in repositories such as KEGG \citep{kanehisa:2000},  about a myriad of interactions, reactions, and regulations. Such information is often identified piecemeal over extended periods and by a variety of researchers, and can therefore be not fully precise. Nevertheless, it allows ordering variables following directional relationships. 

{%
When the ordering of variables is known, then the strategy of neighbourhood recovery turns the problem of learning the structure of a DAG into a straightforward task. The graph selection problem is split into a sequence of feature selection problems by assuming that the conditional distribution of each variable given its precedents in the topological ordering follows the chosen distribution. To learn the structure of a DAG, it is sufficient to perform a (sparse) regression for each variable, treating all preceding variables as covariates. It is known that simply performing lasso-type $\ell_1$-penalized regressions yields consistency for both the coefficients and the sparsity pattern (the set of nonzero coefficients) in regression \citep{li2015sparsistency}, and thus,  yields consistency for the DAG structure. 

The idea of assuming a known ordering of the variables is not novel and several authors have considered decoupling the search over orderings from the graph estimation given the ordering \citep{friedman2003being}. However,  coming up with a good ordering of variables usually requires a significant amount of domain knowledge, which is not commonly available in many practical applications. As a consequence, various approaches exploiting the topological ordering of the variables implement, in different ways,  a search over the space of topological orderings. In the Gaussian setting, \cite{buhlmann2014cam} estimate a superset of the skeleton of the underlying DAG, then search a topological ordering using (restricted) maximum likelihood estimation based on an additive structural equation model with Gaussian errors, and finally, exploiting the estimated order of the variables, use sparse additive regression to estimate the functions in an additive structural equation model. \cite{10.5555/3020336.3020407} propose to learn a DAG by a search, not over the space of structures, but over the space of orderings, selecting for each ordering {  the best consistent network} (OS algorithm); \cite{schmidt2007learning} couple the OS algorithm with a sparsity-promoting $\ell_1$-regularization.


In this work, we propose one algorithm for structure learning of DAGs for count data which is not affected by the non-uniqueness of the topological ordering and overcomes some of the shortcomings induced by the use of penalized procedures. It is the case that penalization is scale-variant, a condition that often interferes with some of the filtering steps that are commonly performed in {the analysis of complex datasets, such as those arising in genomics}. Moreover,  it could suffer from the over-shrinking of small but significant covariate effects.  Our proposal, named { Or-PPGM (PC-based learning of Oriented Poisson Graphical Models)}, is based on a modification of the PC algorithm \citep{kalisch2007estimating, spirtes1993causation}.  In Or-PPGM, we assume to know whether a variable, $i$ say, comes before or after $j$ ($i\ne j$) in an ordering of the available variables that describes the fundamental mechanisms operative in the physical situation. {Furthermore,} we substitute penalized estimation of the local regressions with a testing procedure on the regression parameters, following the lines of the  PC algorithm. Provided that the assumed topological ordering belongs to the space of true topological orderings, we give a theoretical proof of convergence of Or-PPGM that shows that the proposed algorithm consistently estimates the edges of the underlying DAG, as the sample size $n\rightarrow \infty,$ irrespective of the choice of the topological ordering.  The iterative testing procedure performed within the PC algorithm allows to guarantee scale-invariance of the procedure and avoids over-shrinking of small effects. 
}

The paper is organized as follows. Some essential concepts on DAG models and Poisson DAG models are given in Section \ref{background}. Section \ref{proposedmethod} is devoted to the illustration of the proposed algorithm. We then provide statistical guarantees in Section \ref{statistical}, and, in Section \ref{Empiricalstudy-chap2},  experimental results that illustrate the performance of our methods in finite samples. {Section \ref{application} provides an application to gene expression data.} Some conclusions and remarks are provided in Section \ref{guidremarks}. Results needed to prove the main theorem in the paper, another structure learning algorithm and additional simulation results can be found in  Supplementary Material.

\section{Background on Poisson DAG models}\label{background}
In this section, we review, setting up the required notation, some essential concepts on DAG models and present  Poisson DAG models according to the model specification introduced by \cite{park2015learning}.  



Consider a $p$-dimensional random vector $\bold{X}=(X_1,\ldots,X_p)$ such that each random variable $X_s$ corresponds to a node of a directed graph $G=(V,E)$ with index set $V=\{1,2,\ldots,p\}$. 
A directed edge from node $k$ to node $j$ is denoted by $k\rightarrow j$, $k$ is called a parent of node $j$, and $j$ is a child of $k$. The set of parents of a vertex $j$, denoted $pa(j)$, consists of all parents of node $j;$ 
its descendants, i.e., nodes that can be reached from $j$ by repeatedly moving from parent to child, are denoted $de(j).$ Non-descendants of  $j$ are $nd(j) = V\backslash(\{j\}\cup de(j))$. 

A DAG is a directed graph that does not have any directed cycles. In other words, there is no pair $(j,k)$ such that there are directed paths from $j$ to $k$ and from $k$ to $j$.  A topological ordering $j_1,\ldots,j_p$ is an order of $p$ nodes such that there are no directed paths from $j_k$ to $j_t$ if $k>t$.

In a DAG, independence is encoded by the relation of d-separation, defined as in \cite{lauritzengraphical}.
A random vector $\bold{X}$ satisfies the local Markov property with respect to (w.r.t.) a DAG $G$ if 
$X_v \indep \bold{X}_{nd(v)\backslash pa(v)}|\bold{X}_{pa(v)}$   for every $v\in V,$ { where $\bold{X}_U=\{X_t, \, t\in U\}.$}  Similarly, $\bold{X}$ satisfies the global Markov property w.r.t. $G$ if 
$\bold{X}_A\indep \bold{X}_B|\bold{X}_C$
for all triples of pairwise disjoint subsets $A, B, C \subset V $ such that $C$ d-separates $A$ and $B$ in $G$, which we denote by $A \indep_G B | C.$
In this work, we make the assumption that the DAG $G$ is a perfect map, i.e., it satisfies the global Markov property and its reverse implication, known as faithfulness. A distribution  $P_\bold{X}$ is said to be faithful to graph $G$ if $\bold{X}_A\indep \bold{X}_B|\bold{X}_C\Rightarrow A\indep_G B|C,$ for all disjoint vertex sets $A,B,C.$

In the Poisson case, the distribution of $\mathbf{X}$ has the form 
\begin{eqnarray}\label{jointdist}
\mathbb{P}_{\boldsymbol{\theta}}(\bold{x})&=&\prod_{s=1}^{p}\mathbb{P}_{\boldsymbol{\theta}_s}(x_s|\bold{x}_{pa(s)})\\
&=&\exp\big\{\sum_{s=1}^p\sum_{t\in pa(s)}\theta_{st}x_sx_t-\sum_{s=1}^p\log(x_s!)-\sum_{s=1}^pe^{\sum_{t\in pa(s)}\theta_{st}x_t}\big\}.\nonumber
\end{eqnarray}
where { $\bold{x}$ is a realization of the random variable $\bold{X}$}, $\boldsymbol{\theta}_{s}=\{\theta_{st}|~ t\in pa(s)\}$, and $\boldsymbol{\theta}=\{\boldsymbol{\theta}_s,~ s\in V\}$ denotes the set of conditional dependence parameters of the local Poisson regression models characterizing the conditional densities $\mathbb{P}_{\boldsymbol{\theta}_s}(x_s|\bold{x}_{pa(s)})$, 

{\color{black}
$$\mathbb{P}_{\boldsymbol{\theta}_s}(x_s|\bold{x}_{pa(s)})\! = \! \exp\!\big\{\!\!\!\sum_{t\in {pa(s)}}\!\!\!\theta_{st}x_sx_t-\log(x_s!)-e^{\!\sum_{t\in {pa(s)}}\!\!\!\theta_{st}x_t}\big\}.$$
If we zero-pad the parameter $\boldsymbol{\theta}_s \in \mathbb{R}^{|pa(s)|}$ to include zero weights over $V\backslash\{\{s\}\cup pa(s)\}$, then the resulting parameter would lie in $\mathbb{R}^{p-1}$. Therefore, Poisson conditional densities can be written as,}
\begin{eqnarray}\label{dijoinprob}
\mathbb{P}_{\boldsymbol{\theta}_s}(x_s|\bold{x}_{pa(s)})\!& =& \! \exp\!\big\{\!\!\!\sum_{t\in {pa(s)}}\!\!\!\theta_{st}x_sx_t-\log(x_s!)-e^{\!\sum_{t\in {pa(s)}}\!\!\!\theta_{st}x_t}\big\}\\
&=&\exp\big\{x_{s}\langle\boldsymbol{\theta}_{s},\bold{x}_{V\backslash\{s\}}\rangle-\log(x_s!)-D(\langle\boldsymbol{\theta}_{s},\bold{x}_{V\backslash\{s\}}\rangle)\big\},\nonumber
\end{eqnarray}
where $\langle .,.\rangle$  denotes the dot product, and $D(\langle\boldsymbol{\theta}_{s},\bold{x}_{V\backslash\{s\}}\rangle)=e^{\sum_{t\in V\backslash\{s\}}\theta_{st}x_t}$.
This specification puts an edge from node $t$ to node $s$ if $\theta_{st}\ne 0$. A missing edge $t\rightarrow s$ corresponds to the condition $\theta_{st}= 0$, implying conditional independence of $X_s$ and $X_t$ given the parents of $s$, i.e., $X_s\indep X_t|\bold{x}_{pa(s)}$.  As we are only interested in the structure of the graph $G$, without loss of generality we have assumed that the local Poisson regression models characterizing the conditional densities $\mathbb{P}_{\boldsymbol{\theta}_s}(x_s|\bold{x}_{pa(s)})$ have zero intercept.  Specification \eqref{dijoinprob} is similar to that used in \cite{allen2013local} for the undirected version of Poisson graphical models. The only difference lies in the identification of the parameter space for $\boldsymbol{\theta}$  that guarantees the existence of the joint distribution.  While the distribution represented in \eqref{jointdist} is always a valid distribution, in the undirected case a joint distribution compatible with the local specifications exists only if all parameters assume non-positive values.

It is worth noting that, to have a perfect map, it is enough to assume faithfulness of the Poisson node conditional distributions to the graph $G$, as this guarantees faithfulness of the joint distributions thanks to the equivalence between local and global Markov property. 

\section{The Or-PPGM algorithm}\label{proposedmethod}
In this section, we tackle structure learning of DAGs, with the idea of exploiting available prior knowledge of the domain at hand to guide the search for the best structure. In particular, we will assume to know the topological ordering of variables. 
%
%
%
%
In what follows, we adopt the convention of using superscripts, e.g., $\bold{X}^{(1)},\ldots,\bold{X}^{(n)},$ to denote independent copies of the $p$-random vector  $\bold{X}$, where $\bold{X}^{(i)}= (X_{i1},\ldots,X_{ip})$. We denote with $\mathbb{X}=\{\bold{x}^{(1)},\ldots,\bold{x}^{(n)}\}$  the collection of $n$ observed samples drawn from the random vectors $\bold{X}^{(1)},\ldots,\bold{X}^{(n)}$, with $\bold{x}^{(i)}=(x_{i1},\ldots,x_{ip}),$ $ i=1,\ldots,n$. 

Let $i_1,i_2,\ldots,i_p$ indicate one of the possible topological orderings of the variables. The conditional distribution of each variable $X_{i_s}$  given its precedents, denoted $pre(i_s),$ in the topological ordering $i_1,i_2,\ldots,i_p$ can be written as
\begin{eqnarray}\label{orconddist}
\mathbb{P}_{\boldsymbol{\theta}_{i_s|pre(i_s)}}(x_{i_s}|\bold{x}_{pre(i_s)})&=& \exp\big\{x_{i_s} \langle\boldsymbol{\theta}_{i_s|pre(i_s)},\bold{x}_{pre(i_s)}\rangle -\log (x_{i_s}!)\nonumber\\
&&- D(\langle\boldsymbol{\theta}_{i_s|pre(i_s)},\bold{x}_{pre(i_s)}\rangle )\big\},
\end{eqnarray}
where $\boldsymbol{\theta}_{s|\bold{K}}=\{\theta_{st|\bold{K}}:~t\in \bold{K}\}$ 
denote the set of conditional parameters on conditional set $\bold{K}$.
Then, a rescaled negative node conditional log-likelihood formed by products of all the conditional distributions is as follows
\begin{eqnarray}\label{orloglikelihood}
l(\boldsymbol{\theta}_{i_s|pre(i_s)},\mathbb{X}_{i_s};\mathbb{X}_{pre(i_s)})
&&= -\frac{1}{n}\log \prod_{i=1}^{n}\mathbb{P}_{\boldsymbol{\theta}_{i_s|pre(i_s)}}(x_{ii_s}|\bold{x}_{pre(i_s)}^{(i)})\\
&&=\frac{1}{n}\sum_{i=1}^{n}\big[-x_{ii_s}\langle\boldsymbol{\theta}_{i_s|pre(i_s)},\bold{x}_{pre(i_s)}^{(i)}\rangle-\log (x_{ii_s}!)\nonumber\\ &&\quad -D(\langle\boldsymbol{\theta}_{i_s|pre(i_s)},\bold{x}_{pre(i_s)}^{(i)}\rangle)\big].\nonumber
\end{eqnarray}
The absence of an edge $t\rightarrow i_s$ implies $\theta_{i_st|pre(i_s)}$ to be equal to zero. 

{\color{black}  When the topological ordering is known, the expression of the joint distribution in \eqref{jointdist} suggests that  the structure of the network might be recovered from observed data  by disjointly maximizing the single factors in the log-likelihood $\ell(\boldsymbol{\theta},\mathbb{X})$ since the log-likelihood is decomposable as the sum of partial log-likelihoods
over all nodes. Therefore, structure learning could be based on solving local convex optimization problems. Each local estimated conditional dependence parameter  ${\boldsymbol{\hat\theta}_s}$ is then combined to form the global estimate.}

To tackle this problem, we propose a new algorithm, called Or-PPGM, based on a modification of the well-known PC algorithm \citep{kalisch2007estimating}.  
Here, we exploit the idea that the consistency of the PC algorithm ultimately depends upon the consistency of the tests of conditional independence employed in the learning process. In our case, consistent tests can be constructed from Wald-type tests on the parameters ${\theta}_{st|\bold{K}}$ (see also \cite{JMLR:v22:18-401}). We combine this idea with that of making use of topological ordering to determine the sequence of tests to be performed. Assuming that the order of variables is specified beforehand considerably reduces the number of conditional independence tests to be performed. Indeed, for each $s\in V$, it is sufficient to test if the data support the existence of the conditional independence relation $X_s\indep X_t|\bold{X}_{\bold{S}}$ only for  $t\in pre(s)$ and for any $\bold{S}\subseteq  pre(s)\backslash \{t\}$. {   In detail, we assume  that the  distribution of each variable $X_{s},$  conditional to all possible subsets  of variables $\mathbf{X}_\bold{K}, \, \bold{K}\subseteq pre(s)$ is a Poisson distribution:
\begin{equation*}
    X_s| {\bold{x}_{\bold{K}}}\sim \text{Pois}\big(\exp\big\{\sum_{t\in \bold{K}}\theta_{st|\bold{K}}x_t\big\}\big).
\end{equation*}

Then, the algorithm starts from the complete DAG obtained by directing all edges of a complete undirected graph as suggested by the topological ordering. At each level of the cardinality of the conditioning variable set $\bold{S}$, we test, at some pre-specified significance level, the null hypothesis $H_0: {\theta}_{st|\bold{K}}=0$, where $\bold{S}=\bold{K}\backslash\{s\}$. If the null hypothesis is not rejected, the edge $t\rightarrow s$ is considered to be absent from the graph. We note that the cardinality of the set $\bold{S}$  increases from 0 to $\min\{ord(s)-1,m\}$, where $ord(s)$ is the position of node $s$ in the topological ordering and $m$ an upper bound on the cardinality of conditional sets. 
{\color{black} It is worth noting that the value of $m$ is chosen based on prior knowledge about the sparsity of the graph. In the case of no prior knowledge, it will be set to $p-2$.}
For a description of the conditional independence test, as well as the definition of an appropriate test statistic, we refer readers to \cite{JMLR:v22:18-401}. The pseudo-code of the Or-PPGM algorithm is given in Algorithm \ref{Or-PPGMpseudocode}.
\begin{algorithm}
{	\fontsize{9.5}{10}\selectfont
		\caption{\label{Or-PPGMpseudocode} The Or-PPGM algorithm. }
		\begin{algorithmic}[1]
			\hrule
			\vskip 2pt
			\State{\textbf{Input}:} $n$ independent realizations of the $p$-random vector $\bold{X}$,  $\bold{x}^{(1)},\bold{x}^{(2)},\ldots,\bold{x}^{(n)}$; a topological ordering $Ord$, (and a stopping level $m$).
			\State{\textbf{Output}:} An estimated DAG $\hat{G}$.
			\State{} Form the complete undirected graph $\tilde{G}$ on the vertex set $V$.
			\State{} Orient edges on $\tilde{G}$ respecting the topological ordering to form DAG $G'$. 
			\State{} $l=-1$;  $\quad \hat{G}=G'$
			\State{} \textbf{repeat}
			\State{} \quad $l=l+1$
			\State{} \quad \textbf{for} all vertices $s\in V$, \textbf{do}
			\State{} \quad\quad let $\bold{K}_s = pa(s)$
			\State{} \quad \textbf{end for}
			\State{} \quad \textbf{repeat}
			\State{} \quad\quad
			Select a (new) edge $t\rightarrow s$ in $\hat{G}$ such that 
			\State{} \quad\quad$|\bold{K}_s\backslash\{t\}|\ge l$.
			\State{} \quad\quad\textbf{repeat}
			\State{} \quad\quad\quad choose a (new) set $\bold{S}\subset \bold{K}_s\backslash\{t\}$ with $|\bold{S}|=l$.
			\State{} \quad\quad\quad\textbf{if} $H_0: {\theta}_{st|\bold{K}}=0$ not rejected
			\State{} \quad\quad\quad\quad delete edge $t\rightarrow s$ from $\hat{G}$
			\State{} \quad\quad\quad \textbf{end if}
			\State{} \quad\quad\textbf{until} edge $t\rightarrow s$ is deleted or all $\bold{S}\subset \bold{K}_s\backslash\{t\}$ with $|\bold{S}|=l$ have been considered.
			
			\State{} \quad\textbf{until} all edge $t\rightarrow s$ in $\hat{G}$ such that $|\bold{K}_s\backslash\{t\}|\ge l$ and 
			$\bold{S}\subset \bold{K}_s\backslash\{t\}$ with $|\bold{S}|=l$ have been tested for conditional independence.
			\State{} \textbf{until} $l=m$ or for each edge $t\rightarrow s$ in $\hat{G}$: $|\bold{K}_s\backslash\{t\}|< l$.
		\end{algorithmic}
		\hrule
	}
\end{algorithm}

{

}
%

\section{Statistical Guarantees}\label{statistical}

In this section, we address the property of statistical consistency of Or-PPGM. {In detail, we study}  the limiting behaviour of our estimation procedure as the sample size $n$, and the model size $p$ go to infinity. 
In what follows, we derive uniform consistency of our estimators explicitly as a function of the sample size, $n$, the number of nodes, $p$, (and of  $m$) by assuming that the true distribution is faithful to the graph. 
We acknowledge that our results are based on the work of \cite{yang2012graphical}  for exponential family models, and leverage the proof of Lemma 4 in \cite{kalisch2007estimating} in the proof of our main theorem.  

For the readers' convenience, before stating the main result, we summarize some notation that will be used throughout this proof. Given a vector $u\in \mathbb{R}^p$, and a parameter $q\in[0,\infty]$, we write $\|u\|_q$ to denote the usual $\ell_q$ norm. Given a matrix $A\in \mathbb{R}^{p\times p}$, denote the largest and smallest eigenvalues as $\Lambda_{\max}(A)$, $\Lambda_{\min}(A)$, respectively. We use $|||A|||_2= \sqrt{\Lambda_{\max}(A^TA)}$ to denote the spectral norm, corresponding to the largest singular value of $A$,
and the $\ell_\infty$ matrix norm is defined as
$|||A|||_\infty=\max_{i=1,\ldots,a}\sum_{j=1}^{a}|A_{i,j}|.$

\subsection{Assumptions}
We will begin by stating the assumptions that underlie our analysis, and then give a precise statement of the main results. 

Denote the population Fisher information matrix and the sample Fisher information matrix corresponding to the  covariates in model { \eqref{dijoinprob} with $\bold{K}= V\backslash\{s\}$ } as follows
$I_s(\boldsymbol{\theta}_s)=- \mathbb{E}_{\boldsymbol{\theta}}\left(\nabla^2 \log\left(\mathbb{P}_{\boldsymbol{\theta}_s}(X_s|\bold{X}_{V\backslash\{s\}})\right)\right),$
and
$Q_s(\boldsymbol{\theta}_s)= \nabla^2 l(\boldsymbol{\theta}_s,\bold{X}_s;\bold{X}_{V\backslash \{s\}}).$
We note that we will consider the problem of maximum likelihood on a closed and bounded dish $\boldsymbol{\Theta}\subset \mathbb{R}^{(p-1)}$. {\color{black}  For $\boldsymbol{\theta}_{s|\bold{K}}\in \mathbb{R}^{|\bold{K}|}$, we can immerse $\boldsymbol{\theta}_{s|\bold{K}}$ into $\boldsymbol{\Theta}\subset \mathbb{R}^{(p-1)}$ by zero-pad $\boldsymbol{\theta}_{s|\bold{K}}$ to include zero weights over $V\backslash\{\bold{K}\cup\{s\}\}$. }
\begin{assump}\label{assum1}
	The coefficients $\boldsymbol{\theta}_{s|\bold{K}}\in \boldsymbol{\Theta}$ for  all sets $\bold{K}\subseteq V\backslash\{s\}$ and all $s\in V$ have an upper bound norm, 
	$\max_{s,t,\bold{K}}|\theta_{st|\bold{K}}|\le M,$ and  a lower bound norm, 
	$\min_{s,t,\bold{K}}|\theta_{st|\bold{K}}|\ge c,~\forall t\in \bold{K}.$
 \end{assump}
\begin{assump}\label{assum2} The Fisher information  matrix corresponding to the  covariates in model  \eqref{dijoinprob} with $\bold{K}= V\backslash\{s\}$ has bounded eigenvalues, i.e., there exists a constant $\lambda_{\min}>0$ such that
	$\Lambda_{\min}(I_s(\boldsymbol{\theta}_s))\ge \lambda_{\min}, ~\forall~\boldsymbol{\theta}_s\in \boldsymbol{\Theta}.$
	Moreover, we require that
	$\Lambda_{\max}\bigg(\mathbb{E}_{\boldsymbol{\theta}}\left( \bold{X}_{V\backslash \{s\}}^T\bold{X}_{V\backslash \{s\}}\right)\bigg)\le \lambda_{\max}, \forall 	s\in V,\forall~\boldsymbol{\theta}\in \boldsymbol{\Theta},$
	where  $\lambda_{\max}$ is some constant such that $\lambda_{\max} <\infty$.
\end{assump}

	Assumption~\ref{assum1} simply bounds the effects of covariates in all local models. In other words, we consider that the parameters $\theta_{st|\bold{K}}$ belong to a compact set bounded by $M$. {\color{black} Note that the upper bound value $M$ can be arbitrarily large. Hence, this assumption does not limit the general applicability of the method.} Being the expected value of the rescaled negative log-likelihood twice differentiable, the lower
	bound on the eigenvalues of the Fisher information matrix in  Assumption~\ref{assum2} guarantees strong convexity in all partial models. 	Condition on the upper eigenvalue of the covariance matrix guarantees that the relevant covariates do not become overly dependent, a requirement which is commonly adopted in these settings.
 
 It is worth noting that for a given topological ordering,  consistency of the proposed algorithm requires that condition in Assumption \ref{assum1}  is satisfied only on all subsets  $\bold{K}\subseteq pre(i_s)$. However, the topological ordering may not be unique and, as a consequence, different topological orderings may lead to different results.  To prove consistency uniformly over all topological orderings, a stronger assumption is needed, that requires the condition in Assumption \ref{assum1} to be satisfied on all subsets $\bold{K}\subseteq V\backslash \{s\}$. This is the solution adopted here.

\begin{assump}\label{tail}
	Suppose $\bold{X}$ is a $p$-random vector with node conditional distribution specified in \eqref{dijoinprob}. Then, for any   positive constant $\delta$, there exists some constant $c_1>0$, such that
	$\mathbb{P}_{\boldsymbol{\theta}_s}({X_s}\ge \delta\log n)\le c_1 n^{-\delta}, ~\forall~ s\in V, ~\forall~\boldsymbol{\theta}_s\in {\boldsymbol{\Theta}}.$
	
\end{assump}

\begin{assump}\label{innerprod}
	Suppose $\bold{X}$ is a $p$-random vector with node conditional distribution specified in \eqref{dijoinprob}. Then, for any $\boldsymbol{\theta}\in {\boldsymbol{\Theta}},$ there exists some positive constants $\nu,c_2,$ and $\gamma<1/3$, such that
	$\mathbb{P}_{\boldsymbol{\theta}}(\nu+\langle\boldsymbol{\theta},\bold{X}\rangle \ge \gamma\log n)\le c_2 \kappa(n,\gamma),$
	where $\kappa(n,\gamma)=o_p(n^{a})$ for some $a < -1$.
\end{assump}
	The condition on the marginal distribution in Assumption \ref{tail} guarantees that the considered variables do not have heavy tails,  a common condition permitting to achieve consistency. Assumption \ref{innerprod} specifies the parameter space on which we can prove the consistency of local estimators. 
	Compared to  Assumption 5 in \cite{yang2012graphical} and Condition 4 in \cite{yang2015graphical}, Assumption~\ref{innerprod} appears to be much weaker. Indeed,  \cite{yang2012graphical} require $\gamma<\dfrac{1}{4}$ and $\|\boldsymbol{\theta\|_2}\le \dfrac{\log n}{18\log(\max\{n,p\})}$, whereas we only require $\gamma<\dfrac{1}{3}$ and no specified bound is put on $\|\boldsymbol{\theta}\|_2$ (since the negative elements of $\boldsymbol{\theta}$ can be arbitrarily small). Moreover, Condition 4 in \cite{yang2012graphical} is written in analytical form, i.e., a form more restrictive than the probability form here employed.

	When conditional dependencies are all positive, a condition also known as ``additive relationship'' among variables, Assumption 	\ref{innerprod} also { implies} the sparsity of the graphs.



\subsection{Consistency of the Or-PPGM algorithm}
{  
DAG models can be defined only up to their Markov equivalence class, a set consisting of all DAGs encoding the same set of conditional independence.  However,  Poisson DAG models in \eqref{dijoinprob}  are identifiable, as shown in Appendix, Theorem \ref{identify}, where we provide an alternative proof of identifiability benefiting from the ideas developed in the work of \cite{peters2013identifiability}, and avoiding a redundant condition in \cite{park2015learning} (see Theorem 3.1). 
Identifiability has important consequences in our setting. Indeed, it guarantees that the true graph is unique, and, consequently, that Or-PPGM converges to the true unique graph irrespective of which ordering among the true existing ones is chosen to inform the algorithm.
}

\begin{theorem} \label{mainresult}
	Assume \ref{assum1}- \ref{innerprod}. Denote by $\hat{G}(\alpha_n)$ the estimator resulting from Algorithm \ref{Or-PPGMpseudocode}, and by $G$ the true graph. Then, there exists a numerical sequence $\alpha_n\longrightarrow 0$, such that 
	$\mathbb{P}_{\boldsymbol{\theta}}(\hat{G}(\alpha_n)=G)=1, ~\forall~\boldsymbol{\theta}\in {\Omega(\boldsymbol{\Theta})},$
	when $n\longrightarrow \infty$, where $\Omega(\boldsymbol{\Theta})$ is the space such that the faithfulness assumption is satisfied.
\end{theorem} 
{   {\bf Proof.} See Appendix, Section \ref{proof_theorem1}.}

{  
The proof of the above-given Theorem \ref{mainresult} does not depend on which topological order is considered. This implies that, even if for different topological orderings $T_1,T_2,\ldots,T_k$, Algorithm \ref{Or-PPGMpseudocode} performs different sequences of tests $S_1,S_2,\ldots,S_k$, resulting respectively in estimated graphs $\hat{G}^{T_1}(\alpha_n), \hat{G}^{T_2}(\alpha_n), \ldots, \hat{G}^{T_k}(\alpha_n),$   there exists a numerical sequence $\alpha_n\rightarrow 0$, such that the estimators $\hat{G}^{T}(\alpha_n),~ T=T_1,T_2,\ldots, T_k$ {converge} to the true unique graph.	
}

{\color{black}It is worth noting that for structure learning of undirected graphs, \cite{JMLR:v22:18-401} derived statistical guarantees  based on the assumption that the node-wise data generating process belongs to the truncated Poisson distribution. In the case of Poisson node conditional distributions, a proof of consistency of PC-LPGM with proper test statistic can be provided in the situation of “competitive relationships” between variables, and it is still an unsolved question in the case of unrestricted conditional interaction parameters. Here, we consider DAGs, a situation that guarantees the existence of a joint distribution without the need of restricting conditional interaction parameters, i.e., considering both positive and negative parameters. Moreover, in \cite{JMLR:v22:18-401}, for each pair of nodes $s$ and $t$, we test $\theta_{st|\bold{K}} = 0$, where $\bold{K}$ could be all possible subsets of $V\backslash\{s\}$. Here, for each ordered pair of nodes $s$ and $t$, we test $\theta_{st|\bold{K}} = 0$, with $\bold{K}$ as a subset of $pre(s)$. Therefore, 
the number of conditional independent tests performing during the run of PC procedure reduces. This difference ensures the validation of the proof of Theorem \ref{mainresult} when moving from undirected graphs to DAGs. }

%
\section{Empirical study}\label{Empiricalstudy-chap2}
\noindent
Here, we empirically evaluate the ability of our proposal to retrieve the true DAG. As a measure of the ability to recover the true structure of the graphs,  we  adopt three criteria including Precision $P$;  Recall $R$; and their harmonic mean, known as $F_1$-score, respectively defined as 
$$ P=\frac{TP}{TP+FP},\, R=\frac{TP}{TP+FN},\, F_1=2  \frac{P .  R}{P+R},$$
where TP (true positive), FP (false positive), and FN (false negative) refer to the { number of} inferred edges \citep{liu2010stability}. 

 {   We also aim to compare Or-PPGM to possible contestants. To evaluate the effect of limiting the cardinality of the conditional set, we  consider a variant of our proposal, that we call Oriented-Local Poisson Graphical Models (Or-LPGM), that  for each $i_s \in V$ fixes the set of parents of node $i_s$ to be  
$$\hat{pa}(i_s)=\{t\in pre(i_s) \text{ such that } H_0: ~ {\theta}_{i_st|pre(i_s)}=0 \text{ is rejected}\}.$$
Moreover, we compare Or-PPGM to several popular competitors. }As competitors, we consider structure learning algorithms for both Poisson and non-Poisson variables. {   Some of the considered competitors are adaptations to our specific setting of established algorithms and are, therefore, firstly scrutinised in this simulation exercise.} In detail, as representatives of algorithms for Poisson data, we consider: i) one variant of the K2 algorithm \citep{cooper1992bayesian}, PKBIC, able to deal with Poisson data and based on a scoring criterion frequently used in model selection (see Supplementary Material, Section B for details); ii) the PDN (Poisson Dependency Networks) algorithm in \cite{hadiji2015poisson}; iii) the overdispersion scoring (ODS) algorithm in \cite{park2015learning}. It is worth noting that  PKBIC is indeed a new structure learning algorithm for Poisson data, whose consistency is proved in Supplementary Material, Section B. Moreover, we consider a structure learning method dealing with the class of categorical data, namely the Max Min Hill Climbing (MMHC) algorithm \citep{tsamardinos2006max}. To apply such algorithms, we categorize our data using Gaussian mixture models on log-transformed data shifted by 1 \citep{fraley2002model}. 
{ Finally, taking into account that structure learning for discrete data is usually performed by employing methods for continuous data after suitable data transformation, we consider two representatives of  approaches based on the Gaussian assumption,  that are, the PC algorithm \citep{kalisch2007estimating}, and the Bayesian network structure learning \citep{kuipers2022efficient} using BGe score, applied to  log-transformed  data shifted by 1. }



\subsection{Data generation}
\noindent
For two different cardinalities,  $p=10$ and $p=100$,  we consider three graphs of different structure:
{  (i) a scale-free graph, in which the node degree distribution follows a power law; (ii) a hub graph, where each node is connected to one of the hub nodes; (iii) {\color{black} an Erdos-Renyi graph, where the presence of the edges 
is drawn from   {independent and identically distributed} Bernoulli random variables. 
To construct the scale-free and Erdos-Renyi graphs, we employed the R package \textit{igraph} \citep{csardi2006igraph}. For the scale-free graphs, we followed the Barabasi-Albert model with parameter $power=0.01, zero.appeal=p$. For the Erdos-Renyi graphs, we followed the Erdos-Renyi model with probability to draw one edge between two vertices $\gamma=0.2$ for $p=10$ and $\gamma=0.02$ for $p=100$. }To construct the hub graphs, we assumed 2 hub nodes for $p= 10$, and 5 hub nodes for $p=100$.  To convert them into DAGs, we fixed a topological ordering for each graph by taking a permutation of considered variables. Once the order was defined, undirected edges were oriented to form a DAG. }
See Figure~\ref{DAGtypes10} and~\ref{DAGtypes100} for plots of the three chosen DAGs for $p=10$ and $p=100$, respectively.

\begin{figure}[htbp]
	\begin{center}
		\includegraphics[width = 1\linewidth, height=0.3\textheight]{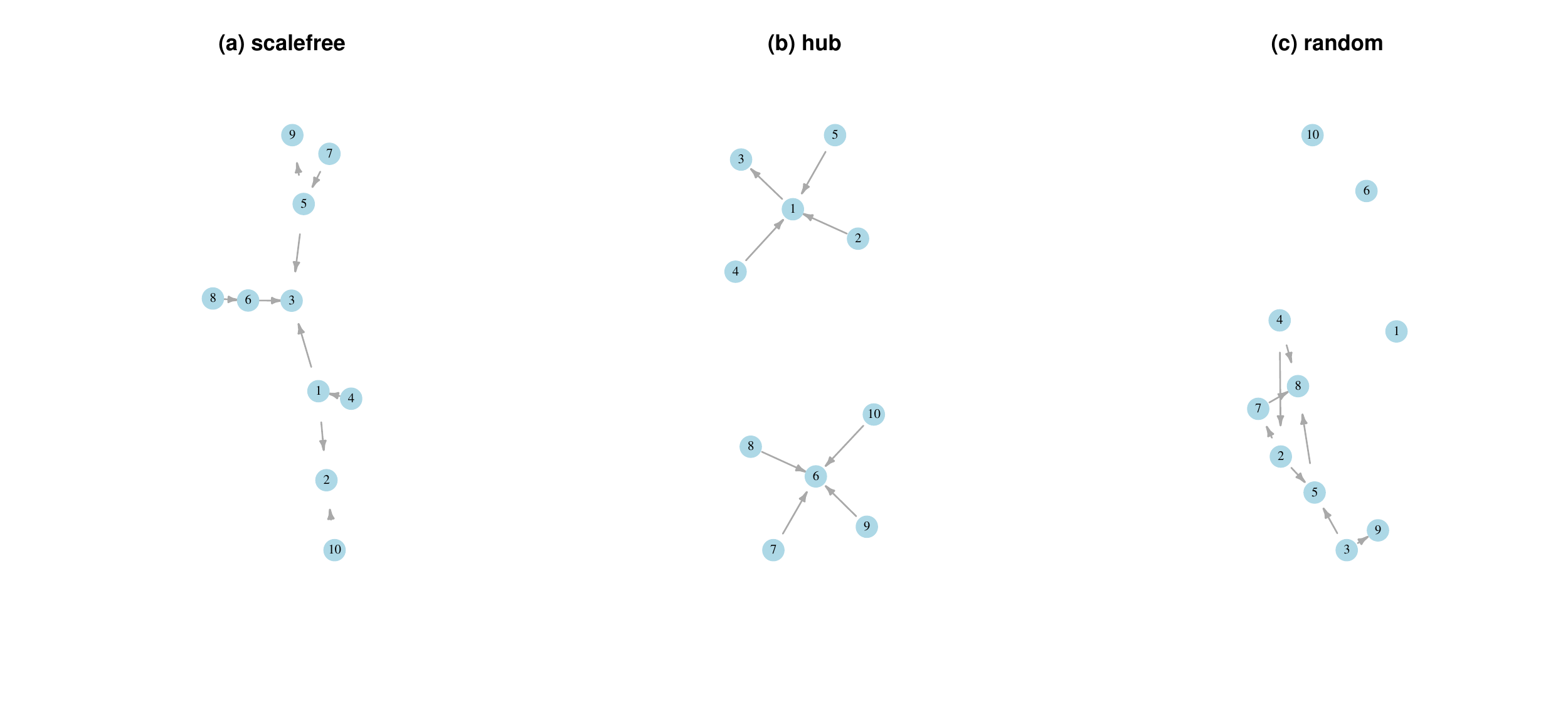}
		\vspace{-2cm}
		\caption{\scriptsize The graph structures for  $p=10$ employed in the simulation studies: (a) scale-free; (b)  hub; (c)  Erdos-Renyi graph.}
		\label{DAGtypes10}
	\end{center}
\end{figure} 
\begin{figure}[htbp]
	\begin{center}
		\includegraphics[width = 1\linewidth, height=0.28\textheight]{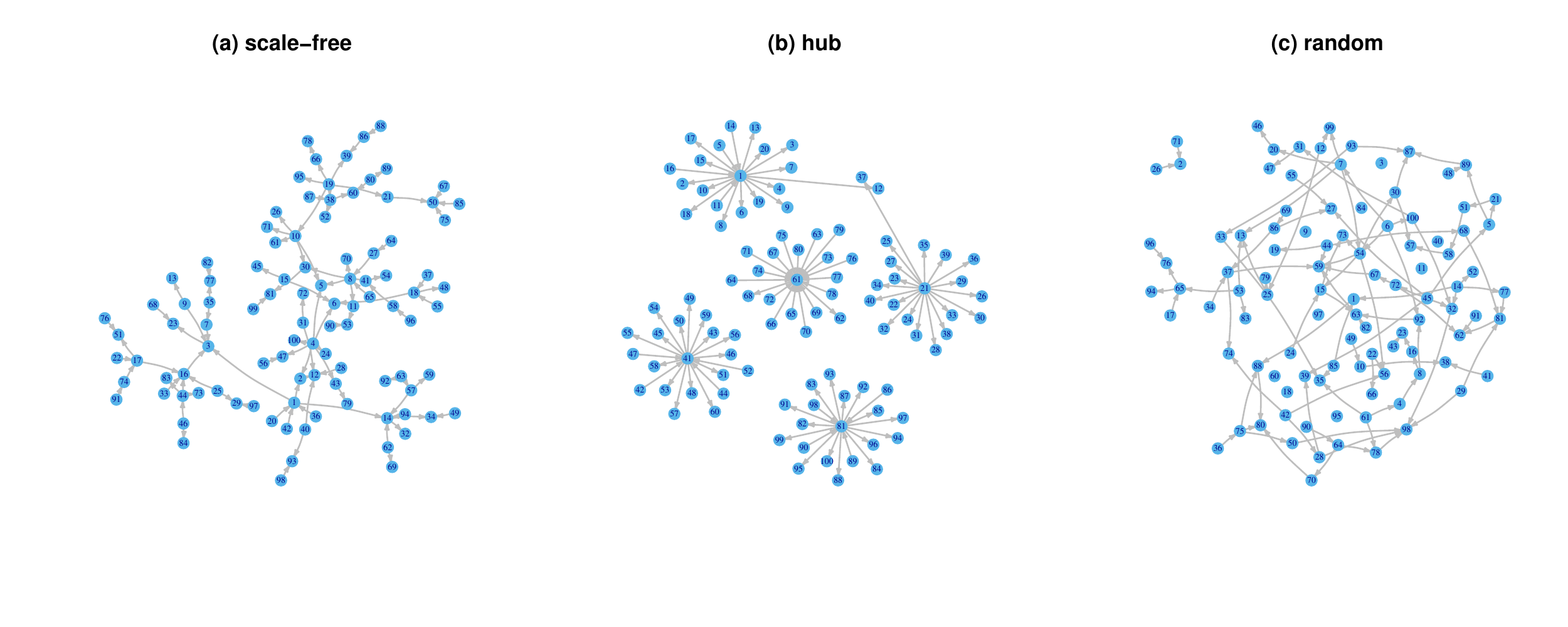}
		\vspace{-2cm}
		\caption{\scriptsize The graph structures for  $p=100$ employed in the simulation studies: (a) scale-free; (b)  hub; (c)  random graph.}
		\label{DAGtypes100}
	\end{center}
\end{figure}
To simulate data, we first construct an adjacency matrix $Adj=(\theta_{ij})$ as follows:
\begin{itemize}
	\item [1.] fill in the adjacency matrix $Adj$ with zeros;
	\item [2.] replace every entry corresponding to a directed edge by one;
	\item[3.] replace each entry equal to 1 with an independent realization from a Uniform random variable $U([-0.5, 0.5])$, representing the true values of parameter $\theta_{st}$. 
\end{itemize}
This yields a matrix $Adj$ whose entries are either zeros or in the range $[-0.5,0.5]$, representing positive and negative relations among variables. 
For each DAG corresponding to an adjacency matrix $Adj$, 50 datasets are sampled for four sample sizes, $n=100,200,500,1000$ with $p=10$, and $n=200,500,1000, 2000$ with $p=100$  as follows. The realization of the first random variable $X_{i_1}$ in the topological ordering $i_1,i_2,\ldots,i_p$ is sampled from a
$ \text{Pois}(\exp\{\theta_1\}),$ where
the default value of $\theta_1$ is 0. Realizations of the following random variables are recursively sampled from
$$X^{(t)}_{i_j}\sim \text{Pois}(\exp\{\sum_{k=i_1}^{i_{(j-1)}}\theta_{i_jk}x_{tk}\}).$$

\subsection{Learning algorithms}\label{guidcompals}

\noindent
Acronyms of the considered algorithms are listed below, along with specifications, if needed, of tuning parameters. In this study, besides the topological ordering, we also specify an additional input,  the upper limit for the cardinality of conditional sets, $m$, which in this study was set to $m=8$ for $p=10$ and $m=3$ for $p=100$, respectively. 

\begin{itemize}
\item[-] {\bf Or-PPGM}:  PC-based learning of Oriented Poisson Graphical Models (Section \ref{proposedmethod});
	\item[-] {\bf PKBIC}: variant of K2 tailored on Poisson data based on the use of BIC (Supplementary Material, Section B); 
	
	\item[-] {\bf Or-LPGM}: Oriented Local Poisson Graphical Model, variant of Or-PPGM with no restriction on cardinality of the conditioning set (Section \ref{Empiricalstudy-chap2});
	\item[-]{\bf PDN}: Poisson Dependency Networks algorithm  \citep{hadiji2015poisson} with n.trees = 20;
	\item[-]{\bf ODS}: Overdispersion Scoring (ODS) algorithm \citep{park2015learning} with $k$-fold cross validation ($k=10$);
	\item[-] {\bf MMHC}: Max Min Hill Climbing algorithm \citep{tsamardinos2006max} applied to data categorized by mixture models, using $\chi^2$ tests of independence.
	\item[-] {\bf PC}: PC algorithm \citep{kalisch2007estimating}  applied to log-transformed data, using Gaussian conditional independent tests.
\item[-] {\bf GBiDAG}: {\color{black} Bayesian network structure learning \citep{kuipers2022efficient} with an iterative order MCMC algorithm on an expanded search space
 using BGe score, using the order as an input, and applied to log-transformed data.}
\end{itemize}

We note that  ODS, PDN and MMHC employ a preliminary step aimed to estimate the topological ordering. This makes the comparison with our algorithms not completely fair. Nevertheless, we decided to consider these algorithms in our numerical studies to get a measure of the impact of the knowledge of the true topological ordering. 

{\color{black} It is also worth noting that the PC algorithm returns PDAGs that consist of both directed and undirected edges. In this case, we borrow the idea of  \cite{dor1992simple} to extend a PDAG to DAG. 
This procedure is guaranteed to find a solution for CPDAGs but not for more general PDAGs, as a directed extension of the PDAG may not exist, and the procedure will pick only one DAG from the equivalence class. However, we can prove that the Poisson DAG  is identifiable (see Appendix, Theorem \ref{identify}), i.e., there is a unique DAG equivalent to the set of conditional independence relations between considered variables. Hence, there is no problem if the procedure picks only one DAG from the equivalence class because the equivalence class has only one element.
For details of the algorithm, we refer the interested reader to the paper by \cite{dor1992simple}.}

\subsection{Results}\label{guidresults}
\noindent
For the two considered vertex cardinalities,  $p=\{10, 100\}$, and the chosen sample sizes,  $n=\{100,200,500,1000,2000\}$, Table \ref{ttable10-chap2} and Table \ref{ttable100-chap2} report, respectively,  Monte Carlo means of TP, FP, FN, P, R and $F_1$ score for each considered method. Each value is computed as an average of the 150 values obtained by simulating 50 samples for each of the three networks. Results disaggregated by network types are given in Supplementary Material, Section D, Tables D.1, and Tables D.2. These results indicate that the proposed algorithm (Or-PPGM), along with Or-LPGM and the modification of K2 (PKBIC) described in Supplementary Material, Section B, outperforms, on average, Gaussian-based competitors (GBiDAG, PC), category-based competitors (MMHC), as well as the state-of-the-art algorithms that are specifically designed for Poisson graphical models (ODS, PDN).

{When $p=10,$ the algorithms PKBIC, Or-PPGM and Or-LPGM reach the highest $F_1$ score, followed by the ODS, GBiDAG, and the PC algorithms. When $n \ge 1000$, the three first algorithms recover almost all edges, see Figure \ref{DAG10F1}.} A closer look at the Precision $P$ and Recall $R$ plot (see Figure D.1 in Supplementary, Section D.2) provides further insight into the behaviour of considered methods. The PKBIC, Or-LPGM and Or-PPGM algorithms always reach the highest Precision $P$ and Recall $R$. 

{\color{black} It is interesting to note that the performance of  PKBIC, Or-PPGM and Or-LPGM appears to be far better than that of the competing algorithms employing the Poisson assumption (PDN and ODS). The use of topological ordering overcomes the inaccuracies of the first step of the ODS algorithm, i.e., the identification of the order of variables, as well as the uncertainties in recovering the direction of interactions in PDN.} 

\begin{figure}[http]
	\begin{center}
		\includegraphics[width = 0.95\linewidth, height=0.43\textheight]{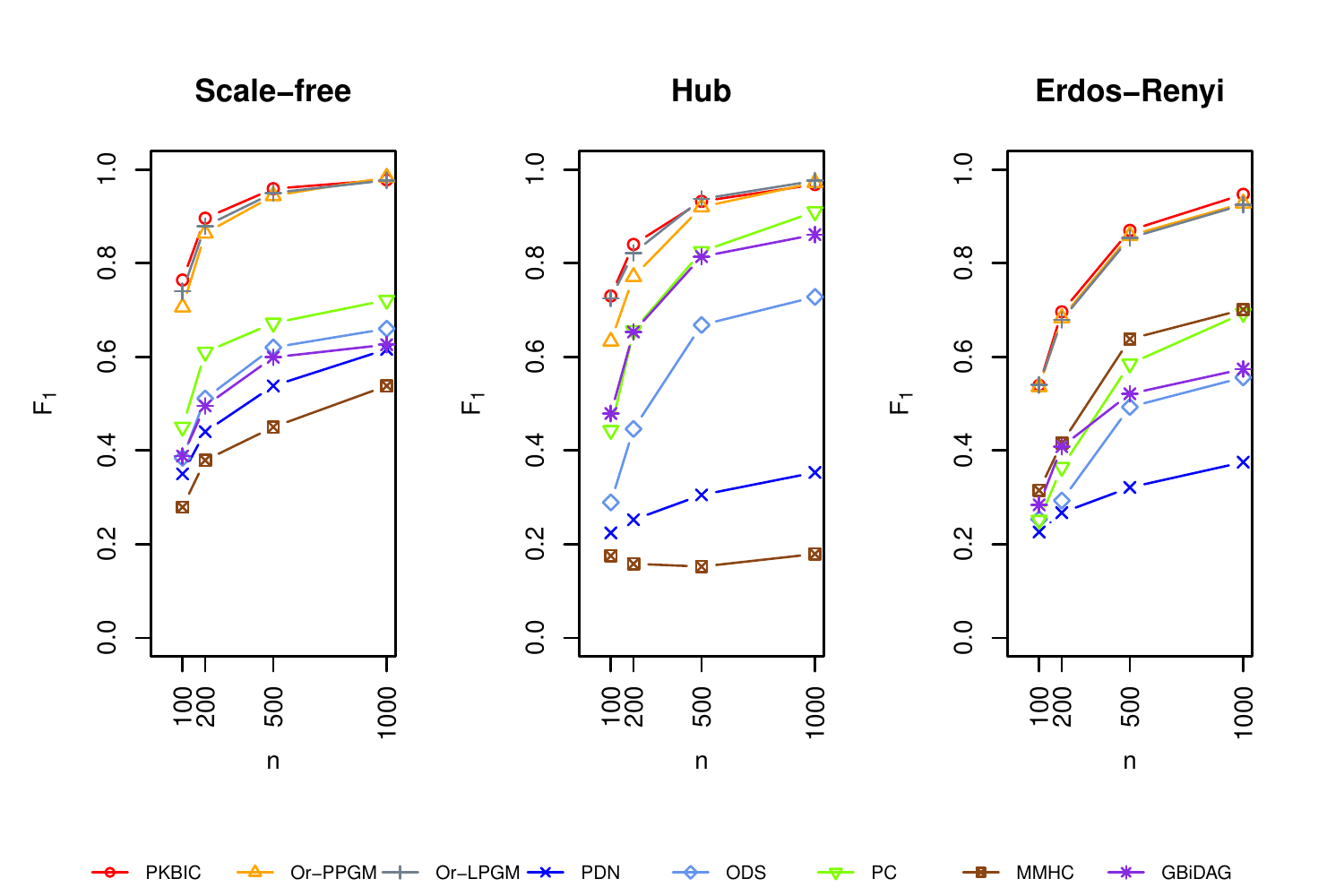}
		\caption{\scriptsize $F_1$-score of the considered algorithms: PKBIC; Or-PPGM; Or-LPGM; PDN; ODS; MMHC;  PC; and BiDAG for the three types of graphs in Figure \ref{DAGtypes10} with $p=10$ and sample sizes $n=100,200,500,1000.$}
		\label{DAG10F1}
	\end{center}
\end{figure}


{\color{black}When considering other methods,  category-based methods (MMHC), and Gaussian-based methods (GBiDAG, PC), both perform less accurately than the three leading methods, i.e., PKBIC, Or-PPGM and Or-LPGM}. {\color{black} Moreover, the GBiDAG is the closest method to our proposal, i.e., using the topological ordering as an input to search for the underlying structure. However, this algorithm works well only for the hub graph with $p=10$. This result can be explained by {the loss of information due to the data transformation, an approach can be ill-suited, possibly leading to wrong inferences in some circumstances \citep{gallopin2013hierarchical}. Another variant of BiDAG that uses BDe score with categorical representation proved to be sensitive to the categorisation of the data, and in particular, not effective when the employed categorisation was that of MMHC. 

}} 

{Results for the high dimensional setting ($p=100$) are somehow comparable to the ones of the previous setting, as it can be seen in Figure~\ref{DAG100F1}, and Figure D.2 in Supplementary, Section D.2. The performance of the considered algorithms are clustered into two different groups. In detail, PKBIC, Or-PPGM, and  Or-LPGM still rank as the top three best algorithms, with Or-LPGM scoring as the best-performing one for the highest sample size. Overall,  their $F_1$ scores become already reasonable when $n$ approaches 1000 observations.}
\begin{figure}[http]
	\begin{center}
		\includegraphics[width = 0.95\linewidth, height=0.43\textheight]{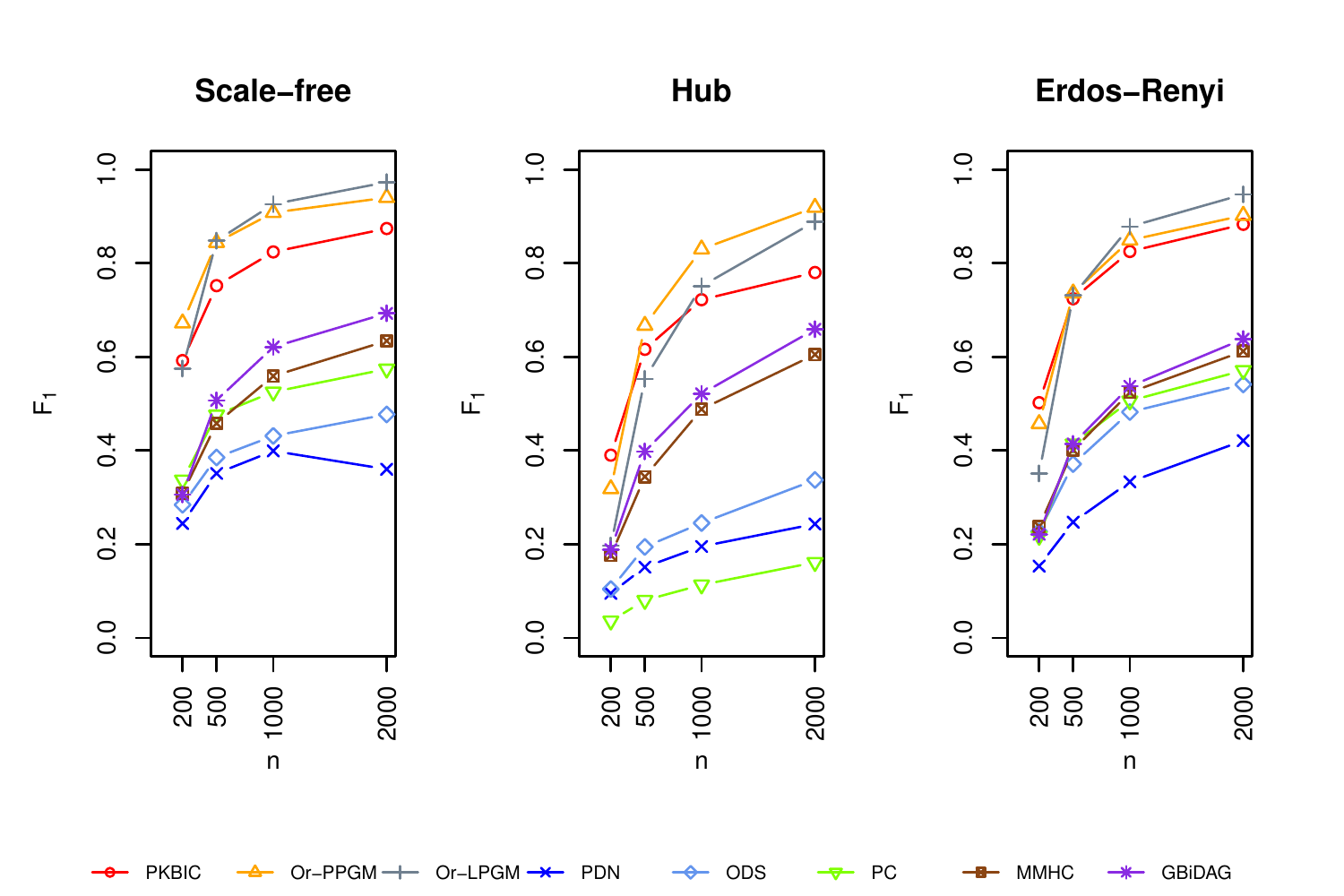}
		\caption{\scriptsize $F_1$-score of the considered algorithms: PKBIC; Or-PPGM; Or-LPGM; PDN; ODS; MMHC;  PC; and BiDAG for the three types of graphs in Figure \ref{DAGtypes100} with $p=100$ and sample sizes $n=200,500,1000,2000.$}
		\label{DAG100F1}
	\end{center}
\end{figure}


{\color{black} We also need to stress the good performances of  Or-PPGM related to the difference between penalization and restriction of the conditional sets. In the PDN  algorithm, as well as in the ODS algorithm, a prediction model is fitted locally on all other variables, using a series of independent penalized regressions. In contrast, Or-PPGM controls the number of variables in the conditional sets for node $s$, which is progressively increased from 0 to $\min\{m,ord(s)-1\}$. 

As a final remark, we note that the performances of  ODS are overall less accurate than expected. A reason for it is that ODS uses the LPGM model \citep{allen2013local} to search the candidate parent sets for each node.  As a consequence, the performance of ODS is highly dependent on the result obtained by the LPGM algorithm. However, this result depends on the tuning of its parameters ($\beta$, $\gamma$, $sth$, etc). Here, we used the best combination of parameters that we managed to find in \cite{JMLR:v22:18-401}, i.e., $B=50,$ $nlambda=20$, $\frac{\lambda_{min}}{\lambda_{max}}=0.01$, $\gamma = 10^{-6}$, $sth=0.6$, $\beta=0.1$ for $p=10$ and $\beta=0.05$ for $p=100$. }

{
\section{Results on Non-small cell lung cancer data}\label{application}

Here, we show an application of our proposed algorithm to the problem of learning gene interactions starting from gene expression measurements on a set of lung cancer cells.

{\color{black} Specifically, we aim at reconstructing the connected part of the manually curated network in Figure 2c of \cite{xue2020rapid} from gene expression data, exploiting the topological ordering deriving from the non-small cell lung cancer \citep{kanehisa:2000}, which is directly connected to the scope of the original analysis.} 

Briefly, the data consists of gene expression measurements of individual cells by \textit{RNA sequencing}, which yields discrete counts as a measure of the activity of each gene. We followed the filtering procedure described in the original publication \citep[see][for details]{xue2020rapid}.

\cite{xue2020rapid} identified 10 different clusters of cells based on their sensitivity to treatment. We selected only the cells in clusters 1, 3, 4, 5, and 10 as described in Figure 1 of \cite{xue2020rapid}, which leads to a total of $n=5505$ cells. These clusters correspond to the cells that showed resistance to the treatment and are of particular biological interest. 

{\color{black} The network in Figure 2c of \cite{xue2020rapid} presents a total of 11 genes, of which, only 8 belong to a connected component of the non-small cell lung cancer pathways, and hence further considered for the analysis, see Figure \ref{KRAS}a. The topological ordering of the considered genes is based on the KEGG pathway database. 
It is worth noting that these 8 genes belong to the initial (upstream) part of the signaling pathway, i.e., the genes that come earlier in the directional flow of information that does not depend on downstream signaling, which alleviates concerns with confounding from other genes in the pathway.
}
}

\begin{figure}[htbp]
	\begin{center}
		\includegraphics[width = 0.95\linewidth, height=0.45\textheight]{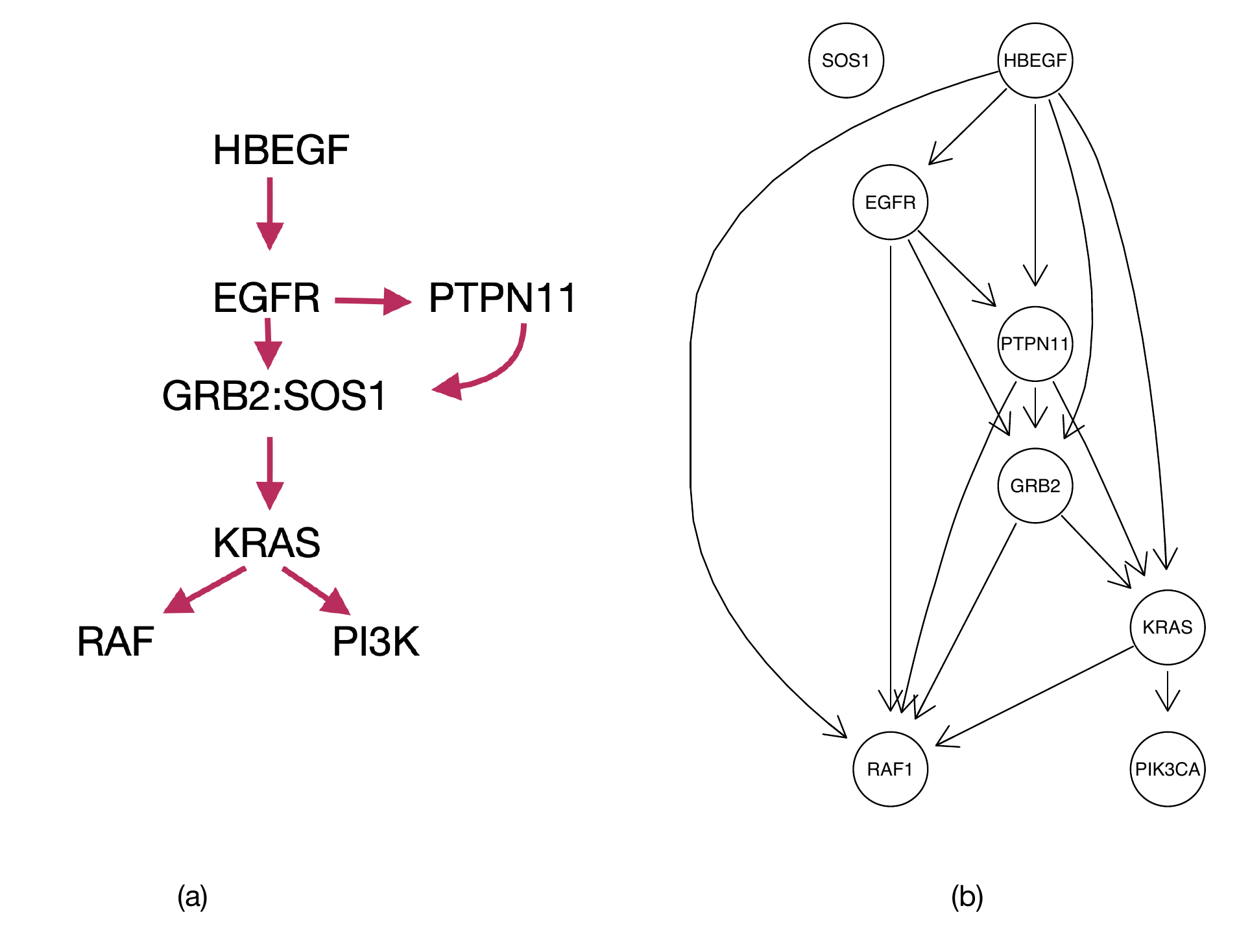}
		\caption{\scriptsize (a) Non-small cell lung cancer network manually curated by \cite{xue2020rapid}; (b) Non-small cell lung cancer network estimated by Or-PPGM algorithm. }
		\label{KRAS}
	\end{center}
\end{figure}

 After removing values that are more than three standard deviations away from the mean, we applied the Or-PPGM algorithm, with a significance level of $5\%$, and the results are shown in Figure \ref{KRAS}b.
By visually comparing our estimated DAG with the manually curated network of \cite{xue2020rapid}, we confirm that our algorithm can reconstruct, from gene expression data, a biologically meaningful structure that confirms several known biological processes. {\color{black} For instance,  the expression of heparin-binding epidermal growth factor (HBEGF) mRNA encodes a ligand of the epidermal growth factor receptor (EGFR) \citep{lemmon2010cell}. In detail, the ligand drives changes in EGFR depending on the expression of secreted HBEGF, and EGFR plays a potential role in mediating adaptation \citep{xue2020rapid}. This is consistent with EGFR being a descendant of HBEGF. Moreover, by activating EGFR, the secretion of HBEGF affects a  population of cells in an autocrine and/or paracrine fashion, which drives nucleotide exchange to activate RAS. Indeed, \cite{xue2020rapid} showed that stimulation with recombinant EGF induced KRAS activation in sorted quiescent cells and enhanced signalling in an EGFR- and PTPN11-dependent manner. This is coherent with the path from HBEGF through EGFR and PTPN11 to KRAS. Aside from this, it is well known that GRB2 mediates the EGF-dependent activation of guanine nucleotide exchange on RAS \citep{gale1993grb2}}.
The fact that our algorithm reconstructs this known signalling pathway holds promise for novel biological insight that could be provided by inspecting other, lesser-known, gene interactions.


\section{Conclusions and remarks}\label{guidremarks}
\noindent
We have considered structure learning of DAGs for count data in a scenario where we know one possible topological ordering of the variables. We have proposed and compared various guided structure learning algorithms that owe their attractiveness to the improvement in accuracy and the reduction of computational costs due to exploitation of the topological ordering, an ingredient that considerably reduces the search space.     For the new proposals, estimators enjoy strong statistical guarantees under assumptions considerably weaker than those employed in related works. Following the empirical comparison with several different approaches, our proposals appear to be promising algorithms as far as prediction accuracy is concerned.

{\color{black} Here, we consider the probabilistic interpretation of DAGs, which differs from the causal interpretation (see \cite{dawid2010beware} for more details). Indeed, in the applications that motivate this work, we cannot assume that the set of observed variables is causally sufficient, that is every common direct cause of the observed variables is also observed. Taking this under consideration, we can not guarantee full appropriateness of a  causal interpretation (see also \cite{pearl2000models}). Nonetheless, as we show in Section 6, learning DAGs is informative and aids the interpretation of the results compared to undirected graphs, ultimately generating causal hypotheses that can be further explored with subsequent experiments.}

{   It is worth remarking that when $p$ is small (such as $p=10$), Or-LPGM performs similarly to Or-PPGM while the average of the runtime of Or-PPGM is around 12 times that of Or-LPGM (see Supplementary, Table D.1). Hence, in low-dimensional regimes, Or-LPGM is the worth-considering variant. However, when the number of variables is large (for example $p=100$), and the sample size is not large enough (for example $n=200$), the performance of Or-LPGM is overall less accurate than Or-PPGM (see Supplementary, Table D.2), an effect due to  inclusion of covariates which are spuriously related to the outcome reduces the estimated residual variance when the variance is unknown. 
Here, Or-PPGM is the preferable solution. }

{\color{black} Or-PPGM makes the stronger assumption that $X_s$ conditional on all possible subsets of its precedents follows a Poisson distribution, while Or-LPGM relaxes this assumption, requiring that $X_s$ conditional on its precedents is Poisson. However, the stronger assumption of Or-PPGM did not negatively affect the performances of the algorithm on simulations built on Or-LPGM assumptions.}

An important side effect of our empirical study was shedding some light on the effect of data transformation finalized to the use of structure learning under model specifications irrespective of the discrete nature of the data, such as those for continuous or categorical data.  We have noticed that making the data continuous by log transformation is better than categorizing them when the PC algorithm is used and that mixture-based categorization is better than cut points-based categorization with K2.  This is an important empirical conclusion that we draw from this study.

{\color{black} Results from Empirical study showed that the two considered variants of BiDAG \citep{kuipers2022efficient} that use BDe or BGe scores perform less accurately than the proposed method because of the loss of information due to the data transformation. In future work, it will be worth considering the BiDAG with the score defined in the PKBIC algorithm, an approach does not suffer from the loss of information due to the data transformation.}

Overall, our exploration has consolidated avenues for learning DAGs, the properties and applications of which leave much room for future research. For example,  the topological ordering may be misspecified, or only a partial order on the set of nodes might be specified due to several reasons.  How to tackle these and other extensions of our setting is the core of our current research.

\begin{center}
	\fontsize{9}{9}\selectfont
		\begin{longtable}{l| l |  r r r r r r r}
			\caption{\label{ttable10-chap2} Monte Carlo marginal means of $TP, FP, FN, P, R$, and $F_1$ score obtained by simulating 50 samples from each of the three networks shown in Figure \ref{DAGtypes10} ($p=10$). The levels of significance of tests $\alpha=2(1-\Phi(n^{0.15}))$ .}
			\\
			\toprule
			$n$	& Algorithm & $TP$ & $FP$ & $FN$ & $P$ & $R$ &$F_1$  \\
			\midrule
			\endfirsthead
			\multicolumn{8}{c}%
			{{\bfseries \tablename\ \thetable{} -- continued from previous page}} \\
			\toprule
			$n$	& Algorithm & $TP$ & $FP$ & $FN$ & $P$ & $R$ &$F_1$  \\
			\midrule
			\endhead
			
100&PKBIC & 5.133 & 1.393 & 3.200 & 0.791 & 0.613 & 0.678 \\ 
  &Or-PPGM & 4.573 & 1.380 & 3.760 & 0.774 & 0.546 & 0.625 \\ 
  &Or-LPGM & 5.200 & 1.740 & 3.133 & 0.758 & 0.621 & 0.669 \\ 
  &PDN & 6.187 & 32.613 & 2.147 & 0.164 & 0.738 & 0.267 \\ 
  &ODS & 1.791 & 0.721 & 6.581 & 0.786 & 0.211 & 0.315 \\ 
  &PC & 2.734 & 2.604 & 5.626 & 0.511 & 0.325 & 0.392 \\ 
  &MMHC & 1.723 & 3.088 & 6.635 & 0.381 & 0.205 & 0.260 \\
  &GBiDAG & 3.080 & 4.094 & 5.283 & 0.434 & 0.368 & 0.392 \\ 
   
  & & & & & & \\
200  &PKBIC & 6.293 & 0.693 & 2.040 & 0.907 & 0.750 & 0.811 \\ 
  &Or-PPGM & 5.853 & 0.680 & 2.480 & 0.904 & 0.698 & 0.773 \\ 
  &Or-LPGM & 6.173 & 0.867 & 2.160 & 0.887 & 0.737 & 0.794 \\ 
  &PDN & 6.820 & 28.820 & 1.513 & 0.201 & 0.814 & 0.320 \\ 
  &ODS & 2.667 & 1.236 & 5.681 & 0.714 & 0.315 & 0.422 \\ 
  &PC & 4.062 & 2.000 & 4.283 & 0.654 & 0.484 & 0.550 \\ 
  &MMHC &  2.329 & 3.859 & 6.007 & 0.392 & 0.276 & 0.319 \\
  &GBiDAG & 4.139 & 3.368 & 4.194 & 0.559 & 0.497 & 0.521 \\
  & & & & & & \\
500  &PKBIC & 7.593 & 0.520 & 0.740 & 0.940 & 0.909 & 0.920 \\ 
  &Or-PPGM & 7.340 & 0.433 & 0.993 & 0.947 & 0.879 & 0.907 \\ 
  &Or-LPGM & 7.387 & 0.387 & 0.947 & 0.955 & 0.884 & 0.914 \\ 
  &PDN & 7.067 & 22.180 & 1.267 & 0.258 & 0.844 & 0.388 \\ 
  &ODS & 4.347 & 1.813 & 3.987 & 0.714 & 0.520 & 0.593 \\ 
  &PC & 5.450 & 1.839 & 2.886 & 0.746 & 0.654 & 0.694  \\ 
  &MMHC & 3.338 & 4.365 & 5.000 & 0.444 & 0.398 & 0.417 \\
  &GBiDAG & 5.351 & 2.932 & 2.986 & 0.656 & 0.643 & 0.646 \\ 
  & & & & & & \\
1000  &PKBIC & 8.093 & 0.353 & 0.240 & 0.962 & 0.970 & 0.964 \\ 
  &Or-PPGM & 7.907 & 0.180 & 0.427 & 0.979 & 0.948 & 0.961 \\ 
  &Or-LPGM & 7.880 & 0.180 & 0.453 & 0.980 & 0.944 & 0.959 \\ 
  &PDN & 7.307 & 17.927 & 1.027 & 0.309 & 0.873 & 0.448 \\ 
  &ODS & 5.213 & 2.527 & 3.120 & 0.681 & 0.625 & 0.648 \\ 
  &PC & 6.233 & 1.513 & 2.100 & 0.805 & 0.749 & 0.775 \\ 
  &MMHC & 4.000 & 4.327 & 4.340 & 0.484 & 0.477 & 0.478 \\
  &GBiDAG & 5.799 & 2.772 & 2.537 & 0.681 & 0.698 & 0.688 \\
			
			\hline
		\end{longtable}
\end{center}

\begin{center}
	\fontsize{9}{9}\selectfont
		\begin{longtable}{l| l |  r r r r r r r}
			\caption{\label{ttable100-chap2} Monte Carlo marginal means of $TP, FP, FN, P, R$, and $F_1$ score obtained by simulating 50 samples from each of the three networks shown in Figure \ref{DAGtypes100} ($p=100$). The levels of significance of tests $\alpha=2(1-\Phi(n^{0.2}))$ for $n=500,1000, 2000$, and $\alpha=2(1-\Phi(n^{0.225}))$ for $n=200$.}
			\\
			\toprule
			$n$	& Algorithm & $TP$ & $FP$ & $FN$ & $P$ & $R$ &$F_1$  \\
			\midrule
			\endfirsthead
			\multicolumn{8}{c}%
			{{\bfseries \tablename\ \thetable{} -- continued from previous page}} \\
			\toprule
			$n$	& Algorithm & $TP$ & $FP$ & $FN$ & $P$ & $R$ &$F_1$  \\
			\midrule
			\endhead
			
			200  &PKBIC & 60.220 & 81.113 & 40.780 & 0.424 & 0.595 & 0.495 \\ 
  &Or-PPGM & 35.307 & 5.320 & 65.693 & 0.854 & 0.349 & 0.482 \\ 
  &Or-LPGM & 26.107 & 6.340 & 74.893 & 0.780 & 0.258 & 0.374 \\ 
  &PDN & 50.580 & 547.487 & 50.420 & 0.128 & 0.498 & 0.164 \\ 
  &ODS & 23.413 & 91.947 & 77.587 & 0.201 & 0.229 & 0.205 \\ 
  &PC & 14.399 & 16.601 & 86.895 & 0.403 & 0.141 & 0.205  \\ 
  &MMHC & 27.527 & 98.873 & 73.473 & 0.216 & 0.272 & 0.241 \\ 
  &GBiDAG& 36.553 & 167.740 & 64.447 & 0.178 & 0.362 & 0.238\\

    & & & & & & \\
  500&PKBIC & 80.933 & 49.520 & 20.067 & 0.619 & 0.800 & 0.697 \\ 
  &Or-PPGM & 63.180 & 3.333 & 37.820 & 0.950 & 0.625 & 0.749 \\ 
  &Or-LPGM & 58.813 & 2.493 & 42.187 & 0.956 & 0.580 & 0.711 \\ 
  &PDN & 64.773 & 419.773 & 36.227 & 0.216 & 0.637 & 0.250 \\ 
  &ODS & 36.040 & 84.007 & 64.960 & 0.298 & 0.353 & 0.317 \\ 
  &PC & 26.693 & 26.127 & 74.307 & 0.444 & 0.259 & 0.323 \\ 
  &MMHC & 47.993 & 89.340 & 53.007 & 0.348 & 0.474 & 0.401 \\ 
  &GBiDAG & 57.773 & 103.533 & 43.227 & 0.358 & 0.573 & 0.440 \\
    & & & & & & \\
  1000&PKBIC & 89.000 & 34.685 & 12.013 & 0.719 & 0.879 & 0.790 \\ 
  &Or-PPGM & 77.658 & 1.208 & 23.356 & 0.985 & 0.769 & 0.862 \\ 
  &Or-LPGM & 76.153 & 0.200 & 24.847 & 0.997 & 0.751 & 0.852 \\ 
  &PDN & 69.713 & 323.793 & 31.287 & 0.286 & 0.685 & 0.309 \\ 
  &ODS & 45.413 & 83.947 & 55.587 & 0.355 & 0.444 & 0.386 \\ 
  &PC & 34.087 & 32.827 & 66.913 & 0.459 & 0.331 & 0.381 \\ 
  &MMHC & 62.447 & 74.787 & 38.553 & 0.455 & 0.618 & 0.524 \\ 
  &GBiDAG & 69.173 & 77.180 & 31.827 & 0.473 & 0.686 & 0.559 \\ 
    & & & & & & \\
  2000&PKBIC & 91.833 & 23.713 & 9.167 & 0.794 & 0.906 & 0.846 \\ 
  &Or-PPGM & 87.273 & 1.493 & 13.727 & 0.984 & 0.865 & 0.920 \\ 
  &Or-LPGM & 89.267 & 0.020 & 11.733 & 1.000 & 0.882 & 0.936 \\ 
  &PDN & 67.733 & 237.620 & 33.267 & 0.338 & 0.664 & 0.341 \\ 
  &ODS & 54.767 & 82.687 & 46.233 & 0.398 & 0.538 & 0.452 \\ 
  &PC & 41.400 & 39.180 & 59.600 & 0.478 & 0.402 & 0.435 \\ 
  &MMHC & 72.100 & 60.813 & 28.900 & 0.544 & 0.714 & 0.617 \\ 
  &GBiDAG & 78.420 & 57.313 & 22.580 & 0.579 & 0.778 & 0.663  \\ 
			\hline
			
		\end{longtable}
\end{center}

\section{Software}
\label{sec5}

The methods presented in this article are available in the \textit{learn2count} R package, available at \url{https://github.com/drisso/learn2count}. The code to reproduce the analyses of this paper is available at \url{https://github.com/kimhuenguyen/guided_structure_learning}.

\appendix
\setcounter{prop}{0}
    \renewcommand{\theprop}{\Alph{section}\arabic{prop}}
    
    \setcounter{theorem}{0}
    \renewcommand{\thetheorem}{\Alph{section}\arabic{theorem}}

\section*{Appendix}
\section{Identifiability}\label{identifiabe-chap3}
\noindent

In what follows, we provide a proof of identifiability of models specified in Section \ref{background} of the main paper.
\begin{prop}\label{setcond}
	Let $\mathbf{X}$ be a $p$-random vector defined as in \eqref{jointdist} and $G=(V,E)$ be a DAG. Consider a variable $X_{j},~ j\in V$, and  {one of its parents} $k\in pa_G(j)$. For all set $S$ with $pa_G(j)\backslash \{k\}\subseteq  S\subseteq nd_G(j)\backslash \{k\},$ we have $X_j\nindep X_k|\mathbf{X}_{S}.$
\end{prop}
\begin{proof}
	{This proposition can be proved easily by using the definition of d-connection and the faithfulness assumption. Indeed, for a fixed node $j\in V$, for all $k\in pa_G(j)$ and for all set $S$ satisfies $pa_G(j)\backslash \{k\}\subseteq  S\subseteq nd_G(j)\backslash \{k\},$ there always exists the path $k\rightarrow j$ satisfies the definition of d-connection. Hence,} $X_j\nindep X_k|\mathbf{X}_{S}.$
\end{proof}
\begin{theorem}\label{identify}
	The Poisson DAG model defined as in Equation \ref{dijoinprob} is identifiable.
\end{theorem}
\begin{proof}
	Assume there are two structure models as in \ref{dijoinprob} which both encode the same set of conditional independences, 
	one with graph $G$, and the other with graph $G'$. We will show that $G\equiv G'$. 
	
	Since DAGs do not contain any cycles, we can always find one node without any child. Indeed, assume to start at some node, and follow a directed path that contains the chosen node. After at most $ |V-1|$ steps, a node without any child is reached. Eliminating such a node from the graph leads to a new DAG. 
	
	We repeat this process  {on $G$ and $G'$} for all nodes that have: (i)  no children, (ii) the same parents in $G$ and $G'$. This process terminates with one of two possible outputs: (a) no nodes left; (b) a  {subset} of variables, which we call again $\mathbf{X}$,  two  {sub-graphs}, which we call again $G$ and $G'$, and a node $j$ that has no children in $G$ such that either $pa_G(j)\ne pa_{G'}(j)$ or $ch_{G'}(j)\ne \emptyset$.  If (a) occurs, the two graphs are identical and the result is proved.  In what follows, we consider the case that (b) occurs.
	
	For such a $j$ node, we have 
	\begin{equation}\label{indepp}
	X_j\indep X_{V\backslash (pa_G(j)\cup \{j\})}|\mathbf{X}_{pa_G(j)},
	\end{equation}
	thanks to the Markov properties with respect to $G$. To make our argument clear, we divide the set of parents $pa_G(j)$ into three disjoint partitions $W,Y,Z$ representing, respectively,  the set of common parents in both graphs; the set of parents in $G$ being a subset of children in $G'$; the set of parents in $G$ which are not parents in $G'$. Formalizing,
	\begin{itemize}
		\item $Z=pa_G(j)\cap pa_{G'}(j)$;
		\item $Y\subset pa_G(j)$ such that $ch_{G'}(j)=Y\cup T$;
		\item $W\subset pa_G(j)$ such that $W$ are not adjacent to $j$ in $G'$.
	\end{itemize}
	Thus, 
	\begin{eqnarray*}
		pa_G(j)=W\cup Y\cup Z, && ch_G(j)=\emptyset,\\
		pa_{G'}(j)= D\cup Z, && ch_{G'}(j)=T\cup Y,
	\end{eqnarray*}
	where $D$ is not adjacent to $j$ in $G$. Let $U=W\cup Y$ and consider the following two cases:
	
	\begin{itemize}
		\item $U=\emptyset$. Then, there exists a node $d\in D$ or a node $t\in T$, otherwise $j$ would have been discarded.
		\begin{itemize}
			\item If there exists a node $d\in D$,  \eqref{indepp} implies $X_j\indep X_d|\mathbf{X}_{Q}$, for $Q=Z\cup D\backslash\{d\}$, which contradicts Proposition \eqref{setcond} applied to $G'$.
			\item If $D=\{\emptyset\}$, and there exists a node $t\in T$, then  \eqref{indepp} implies $X_j\indep X_t|\mathbf{X}_{Q}$, for $Q=Z\cup pa_{G'}(t)\backslash\{j\}$, which contradicts Proposition \eqref{setcond} applied to $G'$.
		\end{itemize}
		\item $U\ne \emptyset$. We note that, within the structure of the graph $G'$, the Poisson assumption implies
		\begin{equation}\label{varmean1}
		\text{Var}\big(X_j|\mathbf{X}_{pa_{G'}(j)}\big)=\mathbb{E}\big(X_j|\mathbf{X}_{pa_{G'}(j)}\big).
		\end{equation}
		However, by applying the law of total variance we get
		\begin{eqnarray*}
			\text{Var}\big(X_j|\mathbf{X}_{pa_{G'}(j)}\big)&=& \text{Var}\big(\mathbb{E}(X_j|\mathbf{X}_{pa_{G'}(j)}\cup \mathbf{X}_{pa_G(j)})|\mathbf{X}_{pa_{G'}(j)}\big)\\
			&&+\mathbb{E}\big(\text{Var}(X_j|\mathbf{X}_{pa_{G'}(j)}\cup \mathbf{X}_{pa_G(j)})|\mathbf{X}_{pa_{G'}(j)}\big).
		\end{eqnarray*}
		By applying Property \eqref{indepp}  we can rewrite
		\begin{equation}\label{var}
		\text{Var}\big(X_j|\mathbf{X}_{pa_{G'}(j)}\big)= \text{Var}\big(\mathbb{E}(X_j| \mathbf{X}_{pa_G(j)})|\mathbf{X}_{pa_{G'}(j)}\big)+\mathbb{E}\big(\text{Var}(X_j| \mathbf{X}_{pa_G(j)})|\mathbf{X}_{pa_{G'}(j)}\big).
		\end{equation}
		Let $f_s(\mathbf{X}_{pa(s)})=\exp\{\sum_{t\in pa(s)}\theta_{st}X_t\}, \forall~ s\in V$. In graph $G$, we have $X_j|\mathbf{X}_{pa_G(j)}\sim \text{Pois}(f_j(\mathbf{X}_{pa_G(j)}))$, so that 
		$$\mathbb{E}(X_j| \mathbf{X}_{pa_G(j)})=\text{Var}(X_j| \mathbf{X}_{pa_G(j)}) =f_j(\mathbf{X}_{pa_G(j)}).$$
		Hence, from Equation \eqref{var}, we get
		\begin{eqnarray}\label{var1}
		\text{Var}\big(X_j|\mathbf{X}_{pa_{G'}(j)}\big)&=& \text{Var}\big(f_j(\mathbf{X}_{pa_G(j)})|\mathbf{X}_{pa_{G'}(j)}\big)+\mathbb{E}\big(\mathbb{E}(X_j| \mathbf{X}_{pa_G(j)})|\mathbf{X}_{pa_{G'}(j)}\big)\\
		&=& \text{Var}\big(f_j(\mathbf{X}_{pa_G(j)})|\mathbf{X}_{pa_{G'}(j)}\big)+\mathbb{E}\big(\mathbb{E}(X_j| \mathbf{X}_{pa_G(j)}\cup \mathbf{X}_{pa_{G'}(j)})|\mathbf{X}_{pa_{G'}(j)}\big)\nonumber\\
		&=& \text{Var}\big(f_j(\mathbf{X}_{pa_G(j)})|\mathbf{X}_{pa_{G'}(j)}\big)+\mathbb{E}\big(X_j|\mathbf{X}_{pa_{G'}(j)}\big),\nonumber
		\end{eqnarray}
		by applying \eqref{indepp}. Equation \eqref{var1} implies
		
		$$\text{Var}\big(X_j|\mathbf{X}_{pa_{G'}(j)}\big)>\mathbb{E}\big(X_j|\mathbf{X}_{pa_{G'}(j)}\big),$$
		since $\text{Var}\big(f_j(\mathbf{X}_{pa_G(j)})|\mathbf{X}_{pa_{G'}(j)}\big)>0$ in general, except at the root node. 
	\end{itemize}
\end{proof}
\section{Proof of Theorem 1}\label{proof_theorem1}

{ 
\begin{proof}
Given a topological ordering, let $\hat{\theta}_{st|\bold{K}}$, and $\theta^*_{st|\bold{K}}$ denote the estimated and true partial weights between $X_s$ and $X_t$ given $X_r, r\in \bold{S}$, where $\bold{S}=\bold{K}\backslash\{t\}\subseteq pre(s)$.  For a fixed-ordered pair of nodes $s,t$, the conditioning sets are elements of 
		$$K_{st}^m=\left\{\bold{S}\subseteq pre(s)\backslash \{t\}: |\bold{S}|\le \min \{ord(s)-1,m\}\right\}.$$
		The cardinality is bounded by 
		$$|K_{st}^m|\le C p^{\min \{ord(s)-1,m\}}\le C p^m,\quad\text{ for some } 0<C<\infty.$$
		Let $E_{st|\bold{K}}$ denote type I or type II errors occurring when testing $H_0:~ \theta_{st|\bold{K}}=0$. Thus
		\begin{equation}\label{error}
		E_{st|\bold{K}} =E_{st|\bold{K}}^I\cup E_{st|\bold{K}}^{II},
		\end{equation}
		in which, for $n$ large enough 
		\begin{itemize}
			\item type I error $E_{st|\bold{K}}^I$: $|Z_{st|\bold{K}}|> \Phi^{-1}(1-\alpha/2)$ and ${\theta^*_{st|\bold{K}}}=0$;
			\item type II error $E_{st|\bold{K}}^{II}$: $|Z_{st|\bold{K}}|\le \Phi^{-1}(1-\alpha/2)$ and ${\theta^*_{st|\bold{K}}}\ne 0$;
		\end{itemize}
		where $Z_{st|\bold{K}}=\displaystyle\dfrac{\sqrt{n}\hat{\theta}_{st|\bold{K}}}{\sqrt{ \left[ J(\hat{\boldsymbol{\theta}}_{s|\bold{K}})^{-1}\right]_{tt}}}$ was defined in \cite{JMLR:v22:18-401}, and $\alpha$ is a chosen significance level. 
		Consider  an arbitrary matrix $\boldsymbol{\theta}_{|\bold{K}}=\{\boldsymbol{\theta}_{s|\bold{K}}\}^T_{s\in V}\in \Omega(\boldsymbol{\Theta})$, such that $|\theta_{st|\bold{K}}|\ge \delta$, for some $\delta>0$. Let $\boldsymbol{\theta}^0_{|\bold{K}}$ be the matrix that has the same elements as $\boldsymbol{\theta}_{|\bold{K}}$ except $\theta_{st|\bold{K}}=\theta^0_{st|\bold{K}}=0$.
		Choose $\alpha_n=2(1-\Phi(n^b))$, { where $0<b<1/2$ will be chosen later}, then
		\begin{eqnarray}\label{error1}
		\sup_{s,t,\bold{K}\in K_{st}^m}\mathbb{P}_{\boldsymbol{\theta}^0_{|\bold{K}}}(E^I_{st|\bold{K}}) &=& \sup_{s,t,\bold{K}\in K_{st}^m}\mathbb{P}_{\boldsymbol{\theta}^0_{|\bold{K}}}\bigg(|\hat{\theta}_{st|\bold{K}}|>n^{b-1/2}\sqrt{ \left[ J(\hat{\boldsymbol{\theta}}_{s|\bold{K}})^{-1}\right]_{tt}}\bigg)\nonumber\\
		&=& \sup_{s,t,\bold{K}\in K_{st}^m}\mathbb{P}_{\boldsymbol{\theta}^0_{|\bold{K}}}\bigg(|\hat{\theta}_{st|\bold{K}}-\theta^0_{st|\bold{K}}|>n^{b-1/2}\sqrt{ \left[ J(\hat{\boldsymbol{\theta}}_{s|\bold{K}})^{-1}\right]_{tt}}\bigg)\nonumber\\
		&\le& \exp\{ -cn\}+c_2n\kappa(n,\gamma)+c_1n^{-2},
		\end{eqnarray}
		using Theorem C.6, Supplementary, and the fact that 
		$n^{b-1/2}\sqrt{ \left[ J(\hat{\boldsymbol{\theta}}_{s|\bold{K}})^{-1}\right]_{tt}}\longrightarrow 0$
		as $n\longrightarrow \infty.$ Furthermore, with the choice of $\alpha_n$ above, and $\delta\ge 2n^{b-1/2}\sqrt{ \left[ J(\hat{\boldsymbol{\theta}}_{s|\bold{K}})^{-1}\right]_{tt}}$,
		\begin{eqnarray*}
			\sup_{s,t,\bold{K}\in K_{st}^m}\mathbb{P}_{\boldsymbol{\theta}_{|\bold{K}}}(E^{II}_{st|\bold{K}}) &=& \sup_{s,t,\bold{K}\in K_{st}^m}\mathbb{P}_{\boldsymbol{\theta}_{|\bold{K}}}\bigg(|\hat{\theta}_{st|\bold{K}}|\le n^{b-1/2}\sqrt{ \left[ J(\hat{\boldsymbol{\theta}}_{s|\bold{K}})^{-1}\right]_{tt}}\bigg)\\
			&=& \sup_{s,t,\bold{K}\in K_{st}^m}\mathbb{P}_{\boldsymbol{\theta}_{|\bold{K}}}\bigg(|\theta_{st|\bold{K}}|-|\hat{\theta}_{st|\bold{K}}|\ge |\theta_{st|\bold{K}}|-n^{b-1/2}\sqrt{ \left[ J(\hat{\boldsymbol{\theta}}_{s|\bold{K}})^{-1}\right]_{tt}}\bigg)\\
			&\le& \sup_{s,t,\bold{K}\in K_{st}^m}\mathbb{P}_{\boldsymbol{\theta}_{|\bold{K}}}\bigg(|\theta_{st|\bold{K}}-\hat{\theta}_{st|\bold{K}}|\ge |\theta_{st|\bold{K}}|-n^{b-1/2}\sqrt{ \left[ J(\hat{\boldsymbol{\theta}}_{s|\bold{K}})^{-1}\right]_{tt}}\bigg)\\
			&\le& \sup_{s,t,\bold{K}\in K_{st}^m}\mathbb{P}_{\boldsymbol{\theta}_{|\bold{K}}}\bigg(|\hat{\theta}_{st|\bold{K}}-\theta_{st|\bold{K}}|\ge n^{b-1/2}\sqrt{ \left[ J(\hat{\boldsymbol{\theta}}_{s|\bold{K}})^{-1}\right]_{tt}}\bigg),
		\end{eqnarray*}
		Finally, by Theorem C.6, Supplementary, we then obtain
		\begin{equation}\label{error2}
		\sup_{s,t,\bold{K}\in K_{st}^m}\mathbb{P}_{\boldsymbol{\theta}_{|\bold{K}}}(E^{II}_{st|\bold{K}})
		\le \exp\{ -cn\}+c_2n\kappa(n,\gamma)+c_1n^{-2},
		\end{equation}
		as $n\longrightarrow\infty$. 	Now, by \eqref{error}-\eqref{error2}, we get
		\begin{eqnarray}\label{totalerror}
		&&\mathbb{P}_{\boldsymbol{\theta}}(\text{ a type I or II error occurs in testing procedure})\nonumber\\
		&&~\le \mathbb{P}_{\boldsymbol{\theta}_{|\bold{K}}}(\cup_{s,t,\bold{K}\in K_{st}^m}E_{st|\bold{K}})\nonumber\\
		&&~\le O_p(p^{m+2})\sup_{s,t,\bold{K}\in K_{st}^m}\mathbb{P}_{\boldsymbol{\theta}_{|\bold{K}}}(E_{st|\bold{K}})\nonumber\\
		&&~\le O_p(p^{m+2})\bigg[\exp\{ -cn\}+c_2n\kappa(n,\gamma)+c_1n^{-2}\bigg]\nonumber\\
		&&~\rightarrow 0\nonumber, 
		\end{eqnarray}
		as $n\longrightarrow\infty$.

\end{proof}}

\section*{Acknowledgements}
\noindent
D.R. was supported by the National Cancer Institute of the National Institutes of Health [2U24CA180996]. This work was supported in part by CZF2019-002443 (D.R. and T.K.H.N.) from the Chan Zuckerberg Initiative DAF, an advised fund of the Silicon Valley Community Foundation.

\bibliography{smj-template}

\begin{thebibliography}{29}
\providecommand{\natexlab}[1]{#1}
\providecommand{\url}[1]{\texttt{#1}}
\expandafter\ifx\csname urlstyle\endcsname\relax
  \providecommand{\doi}[1]{doi: #1}\else
  \providecommand{\doi}{doi: \begingroup \urlstyle{rm}\Url}\fi

\bibitem[Allen and Liu(2013)]{allen2013local}
{\rm Allen, G. {\rm and} Liu, Z.} (2013).
\newblock A local {P}oisson graphical model for inferring networks from
  sequencing data.
\newblock \emph{NanoBioscience, IEEE Transactions on}, {\bf 12}\penalty0 (3),
  \penalty0 189--198.

\bibitem[B{\"u}hlmann et~al.(2014)B{\"u}hlmann, Peters, and
  Ernest]{buhlmann2014cam}
{\rm B{\"u}hlmann, P., Peters, J., {\rm and} Ernest, J.} (2014).
\newblock {CAM:} causal additive models, high-dimensional order search and
  penalized regression.
\newblock \emph{The Annals of Statistics}, {\bf 42}\penalty0 (6), \penalty0
  2526--2556.

\bibitem[Cooper and Herskovits(1992)]{cooper1992bayesian}
{\rm Cooper, G.~F. {\rm and} Herskovits, E.} (1992).
\newblock A bayesian method for the induction of probabilistic networks from
  data.
\newblock \emph{Machine learning}, {\bf 9}\penalty0 (4), \penalty0 309--347.

\bibitem[Csardi et~al.(2006)Csardi, Nepusz, et~al.]{csardi2006igraph}
{\rm Csardi, G., Nepusz, T., et~al.} (2006).
\newblock The igraph software package for complex network research.
\newblock \emph{InterJournal, complex systems}, {\bf 1695}\penalty0 (5),
  \penalty0 1--9.

\bibitem[Dawid(2010)]{dawid2010beware}
{\rm Dawid, A.~P.} (2010).
\newblock Beware of the dag!
\newblock In \emph{Causality: objectives and assessment}, pages 59--86. PMLR.

\bibitem[Dor and Tarsi(1992)]{dor1992simple}
{\rm Dor, D. {\rm and} Tarsi, M.} (1992).
\newblock A simple algorithm to construct a consistent extension of a partially
  oriented graph.
\newblock \emph{Technicial Report R-185, Cognitive Systems Laboratory, UCLA}.

\bibitem[Fraley and Raftery(2002)]{fraley2002model}
{\rm Fraley, C. {\rm and} Raftery, A.~E.} (2002).
\newblock Model-based clustering, discriminant analysis, and density
  estimation.
\newblock \emph{Journal of the American Statistical Association}, {\bf
  97}\penalty0 (458), \penalty0 611--631.

\bibitem[Friedman and Koller(2003)]{friedman2003being}
{\rm Friedman, N. {\rm and} Koller, D.} (2003).
\newblock Being bayesian about network structure. a bayesian approach to
  structure discovery in bayesian networks.
\newblock \emph{Machine learning}, {\bf 50}, \penalty0 95--125.

\bibitem[Gale et~al.(1993)Gale, Kaplan, Lowenstein, Schlessinger, and
  Bar-Sagi]{gale1993grb2}
{\rm Gale, N.~W., Kaplan, S., Lowenstein, E.~J., Schlessinger, J., {\rm and}
  Bar-Sagi, D.} (1993).
\newblock Grb2 mediates the egf-dependent activation of guanine nucleotide
  exchange on ras.
\newblock \emph{Nature}, {\bf 363}\penalty0 (6424), \penalty0 88--92.

\bibitem[Gallopin et~al.(2013)Gallopin, Rau, and
  Jaffr{\'e}zic]{gallopin2013hierarchical}
{\rm Gallopin, M., Rau, A., {\rm and} Jaffr{\'e}zic, F.} (2013).
\newblock A hierarchical poisson log-normal model for network inference from
  rna sequencing data.
\newblock \emph{PloS one}, {\bf 8}\penalty0 (10), \penalty0 e77503.

\bibitem[Hadiji et~al.(2015)Hadiji, Molina, Natarajan, and
  Kersting]{hadiji2015poisson}
{\rm Hadiji, F., Molina, A., Natarajan, S., {\rm and} Kersting, K.} (2015).
\newblock Poisson dependency networks: {G}radient boosted models for
  multivariate count data.
\newblock \emph{Machine Learning}, {\bf 100}\penalty0 (2-3), \penalty0
  477--507.

\bibitem[Kalisch and B{\"u}hlmann(2007)]{kalisch2007estimating}
{\rm Kalisch, M. {\rm and} B{\"u}hlmann, P.} (2007).
\newblock Estimating high-dimensional directed acyclic graphs with the
  {PC}-algorithm.
\newblock \emph{Journal of Machine Learning Research}, {\bf 8}\penalty0 (Mar),
  \penalty0 613--636.

\bibitem[Kanehisa and Goto(2000)]{kanehisa:2000}
{\rm Kanehisa, M. {\rm and} Goto, S.} (2000).
\newblock {KEGG:} {K}yoto {E}ncyclopedia of {G}enes and {G}enomes.
\newblock \emph{Nucleic Acids Research}, {\bf 28}\penalty0 (1), \penalty0
  27--30.
\newblock ISSN 0305-1048.
\newblock URL \url{http://www.ncbi.nlm.nih.gov/pubmed/10592173}.

\bibitem[Kuipers et~al.(2022)Kuipers, Suter, and Moffa]{kuipers2022efficient}
{\rm Kuipers, J., Suter, P., {\rm and} Moffa, G.} (2022).
\newblock Efficient sampling and structure learning of bayesian networks.
\newblock \emph{Journal of Computational and Graphical Statistics}, {\bf
  31}\penalty0 (3), \penalty0 639--650.

\bibitem[Lauritzen(1996)]{lauritzengraphical}
{\rm Lauritzen, S.~L.} (1996).
\newblock \emph{Graphical {M}odels}, volume~17.
\newblock Clarendon Press, Oxford.

\bibitem[Lemmon and Schlessinger(2010)]{lemmon2010cell}
{\rm Lemmon, M.~A. {\rm and} Schlessinger, J.} (2010).
\newblock Cell signaling by receptor tyrosine kinases.
\newblock \emph{Cell}, {\bf 141}\penalty0 (7), \penalty0 1117--1134.

\bibitem[Li et~al.(2015)Li, Scarlett, Ravikumar, and
  Cevher]{li2015sparsistency}
{\rm Li, Y.-H., Scarlett, J., Ravikumar, P., {\rm and} Cevher, V.} (2015).
\newblock Sparsistency of $\ell_1$-regularized m-estimators.
\newblock In \emph{Artificial Intelligence and Statistics}, pages 644--652.
  PMLR.

\bibitem[Liu et~al.(2010)Liu, Roeder, and Wasserman]{liu2010stability}
{\rm Liu, H., Roeder, K., {\rm and} Wasserman, L.} (2010).
\newblock Stability approach to regularization selection (stars) for high
  dimensional graphical models.
\newblock In \emph{Advances in neural information processing systems}, pages
  1432--1440.

\bibitem[Nguyen and Chiogna(2021)]{JMLR:v22:18-401}
{\rm Nguyen, T. K.~H. {\rm and} Chiogna, M.} (2021).
\newblock Structure learning of undirected graphical models for count data.
\newblock \emph{Journal of Machine Learning Research}, {\bf 22}\penalty0 (50),
  \penalty0 1--53.
\newblock URL \url{http://jmlr.org/papers/v22/18-401.html}.

\bibitem[Park and Raskutti(2015)]{park2015learning}
{\rm Park, G. {\rm and} Raskutti, G.} (2015).
\newblock Learning large-scale {P}oisson {DAG} models based on overdispersion
  scoring.
\newblock In \emph{Advances in Neural Information Processing Systems}, pages
  631--639.

\bibitem[Pearl et~al.(2000)]{pearl2000models}
{\rm Pearl, J. et~al.} (2000).
\newblock Models, reasoning and inference.
\newblock \emph{Cambridge, UK: CambridgeUniversityPress}, {\bf 19}\penalty0
  (2), \penalty0 3.

\bibitem[Peters and B{\"u}hlmann(2013)]{peters2013identifiability}
{\rm Peters, J. {\rm and} B{\"u}hlmann, P.} (2013).
\newblock Identifiability of {G}aussian structural equation models with equal
  error variances.
\newblock \emph{Biometrika}, {\bf 101}\penalty0 (1), \penalty0 219--228.

\bibitem[Schmidt et~al.(2007)Schmidt, Niculescu-Mizil, and
  Murphy]{schmidt2007learning}
{\rm Schmidt, M., Niculescu-Mizil, A., {\rm and} Murphy, K.} (2007).
\newblock Learning graphical model structure using $\ell_1$-regularization
  paths.
\newblock In \emph{AAAI}, volume~7, pages 1278--1283.

\bibitem[Spirtes et~al.(1993)Spirtes, Glymour, and
  Scheines]{spirtes1993causation}
{\rm Spirtes, P., Glymour, C., {\rm and} Scheines, R.} (1993).
\newblock Causation, prediction, and search. 1993.
\newblock \emph{Lecture Notes in Statistics}, {\bf 81}.

\bibitem[Teyssier and Koller(2005)]{10.5555/3020336.3020407}
{\rm Teyssier, M. {\rm and} Koller, D.} (2005).
\newblock Ordering-based search: A simple and effective algorithm for learning
  bayesian networks.
\newblock In \emph{Proceedings of the Twenty-First Conference on Uncertainty in
  Artificial Intelligence}, UAI'05, page 584–590, Arlington, Virginia, USA,
  2005. AUAI Press.
\newblock ISBN 0974903914.

\bibitem[Tsamardinos et~al.(2006)Tsamardinos, Brown, and
  Aliferis]{tsamardinos2006max}
{\rm Tsamardinos, I., Brown, L.~E., {\rm and} Aliferis, C.~F.} (2006).
\newblock The max-min hill-climbing {B}ayesian network structure learning
  algorithm.
\newblock \emph{Machine learning}, {\bf 65}\penalty0 (1), \penalty0 31--78.

\bibitem[Xue et~al.(2020)Xue, Zhao, Aronowitz, Mai, Vides, Qeriqi, Kim, Li,
  de~Stanchina, Mazutis, Risso, and Lito]{xue2020rapid}
{\rm Xue, J.~Y., Zhao, Y., Aronowitz, J., Mai, T.~T., Vides, A., Qeriqi, B.,
  Kim, D., Li, C., de~Stanchina, E., Mazutis, L., Risso, D., {\rm and} Lito,
  P.} (2020).
\newblock Rapid non-uniform adaptation to conformation-specific kras (g12c)
  inhibition.
\newblock \emph{Nature}, {\bf 577}\penalty0 (7790), \penalty0 421--425.

\bibitem[Yang et~al.(2012)Yang, Allen, Liu, and Ravikumar]{yang2012graphical}
{\rm Yang, E., Allen, G., Liu, Z., {\rm and} Ravikumar, P.~K.} (2012).
\newblock Graphical models via generalized linear models.
\newblock In \emph{Advances in Neural Information Processing Systems}, pages
  1358--1366.

\bibitem[Yang et~al.(2015)Yang, Ravikumar, Allen, and Liu]{yang2015graphical}
{\rm Yang, E., Ravikumar, P., Allen, G.~I., {\rm and} Liu, Z.} (2015).
\newblock Graphical models via univariate exponential family distributions.
\newblock \emph{Journal of Machine Learning Research}, {\bf 16}\penalty0 (1),
  \penalty0 3813--3847.

\end{thebibliography}


\begin{thebibliography}{9}
\providecommand{\natexlab}[1]{#1}
\providecommand{\url}[1]{\texttt{#1}}
\expandafter\ifx\csname urlstyle\endcsname\relax
  \providecommand{\doi}[1]{doi: #1}\else
  \providecommand{\doi}{doi: \begingroup \urlstyle{rm}\Url}\fi

\bibitem[Akaike(1974)]{akaike1974new}
H.~Akaike.
\newblock A new look at the statistical model identification.
\newblock \emph{IEEE transactions on automatic control}, 19\penalty0
  (6):\penalty0 716--723, 1974.

\bibitem[Chickering(2002)]{chickering2002optimal}
D.~M. Chickering.
\newblock Optimal structure identification with greedy search.
\newblock \emph{Journal of machine learning research}, 3\penalty0
  (Nov):\penalty0 507--554, 2002.

\bibitem[Cooper and Herskovits(1992)]{cooper1992bayesian}
G.~F. Cooper and E.~Herskovits.
\newblock A bayesian method for the induction of probabilistic networks from
  data.
\newblock \emph{Machine learning}, 9\penalty0 (4):\penalty0 309--347, 1992.

\bibitem[Haughton et~al.(1988)]{haughton1988choice}
D.~M. Haughton et~al.
\newblock On the choice of a model to fit data from an exponential family.
\newblock \emph{The Annals of Statistics}, 16\penalty0 (1):\penalty0 342--355,
  1988.

\bibitem[Hoeffding(1963)]{hoeffding1963probability}
W.~Hoeffding.
\newblock Probability inequalities for sums of bounded random variables.
\newblock \emph{Journal of the American statistical association}, 58\penalty0
  (301):\penalty0 13--30, 1963.

\bibitem[Johnson and Robinson(1981)]{johnson1981eigenvalue}
C.~R. Johnson and H.~A. Robinson.
\newblock Eigenvalue inequalities for principal submatrices.
\newblock \emph{Linear Algebra and its Applications}, 37:\penalty0 11--22,
  1981.

\bibitem[Nguyen and Chiogna(2021)]{JMLR:v22:18-401}
T~K~H Nguyen and M.~Chiogna.
\newblock Structure learning of undirected graphical models for count data.
\newblock \emph{Journal of Machine Learning Research}, 22\penalty0
  (50):\penalty0 1--53, 2021.
\newblock URL \url{http://jmlr.org/papers/v22/18-401.html}.

\bibitem[Schwarz(1978)]{schwarz1978estimating}
G.~Schwarz.
\newblock Estimating the dimension of a model.
\newblock \emph{The Annals of Statistics}, 6\penalty0 (2):\penalty0 461--464,
  1978.

\bibitem[Wagenmakers and Farrell(2004)]{wagenmakers2004aic}
E.~J. Wagenmakers and S.~Farrell.
\newblock {AIC} model selection using {A}kaike weights.
\newblock \emph{Psychonomic Bulletin \& Review}, 11\penalty0 (1):\penalty0
  192--196, 2004.

\end{thebibliography}

\end{document}


\def\spacingset#1{\renewcommand{\baselinestretch}%
	{#1}\small\normalsize} \spacingset{1}

\setcounter{nx}{1}


\if0\blind
{
	\title{\large\bf  Guided structure learning of DAGs for count data \\
        SUPPLEMENTARY MATERIAL}
	\author{Thi Kim Hue Nguyen \\
		Department of Statistical Sciences, University of Padova\\
		and \\
		Monica Chiogna \\
		Department of Statistical Sciences ``Paolo Fortunati'', University of Bologna\\
		and \\
		Davide Risso\\
		Department of Statistical Sciences, University of Padova\\
		and \\
        Erika Banzato\\
        Department of Statistical Sciences, University of Padova\\
		}
	\maketitle
} \fi

\if1\blind
{
	\bigskip
	\bigskip
	\bigskip
	\begin{center}
		{\LARGE\bf Guided structure learning of DAGs for count data}
	\end{center}
	\medskip
} \fi

\bigskip

\setcounter{section}{0}
\renewcommand{\thesection}{\Alph{section}}


%
%
%
\section{Equivalence of the conditional dependence tests}
Here, we prove the test whether $\theta_{st|\bold{K}} \neq 0$ is equivalent to the conditional dependence test. 
\begin{dl}
Assume the order of variables is specified beforehand, and for each $s\in V$, assume
$X_s| \bold{x}_\bold{K}\sim \text{Pois}\big(\exp\big\{\sum_{t\in \bold{K}}\theta_{st|\bold{K}}x_t\big\}\big),~ \forall~ \bold{K} \subseteq  pre(s)$,  then  the test whether $\theta_{st|\bold{K}} \neq 0$ is equivalent to the conditional dependence test  $X_s \indep X_t|\bold{X}_\bold{S},$ where $ \bold{S}=\bold{K}\backslash\{t\}$.  
\end{dl}
\begin{proof}
Indeed, we combine the idea of Wald-type tests on the parameters ${\theta}_{st|\bold{K}}$  (see \cite{JMLR:v22:18-401}) with that of making use of the topological ordering to determine the sequence of tests to be performed. In detail, for each $s\in V$, we assume
$X_s| {\bold{x}_{\bold{K}}}\sim \text{Pois}\big(\exp\big\{\sum_{t\in \bold{K}}\theta_{st|\bold{K}}x_t\big\}\big), ~\forall ~\bold{K} \subseteq  pre(s)$,  
i.e., 
$$
\mathbb{P}_{\boldsymbol{\theta}_{s|\bold{K}}}(x_{s}|\bold{x}_{\bold{K}})= \exp\big\{\sum_{t\in \bold{K}}\theta_{st|\bold{K}}x_sx_t-\log(x_s!)-e^{\sum_{t\in \bold{K}}\theta_{st|\bold{K}}x_t}\big\}.
$$
 It is easy to see that $\theta_{st|\bold{K}}=0$ is equivalent to the conditional independence relation. Indeed, 
 \begin{itemize}
     \item if $\theta_{st|\bold{K}}=0$  we have
$$
\mathbb{P}_{\boldsymbol{\theta}}(x_{s},x_t|\bold{x}_{\bold{S}})=\mathbb{P}_{\boldsymbol{\theta}}(x_{s}|x_t,\bold{x}_{\bold{S}})\mathbb{P}_{\boldsymbol{\theta}}(x_{t}|\bold{x}_{\bold{S}})=\mathbb{P}_{\boldsymbol{\theta}}(x_{s}|\bold{x}_{\bold{S}})\mathbb{P}_{\boldsymbol{\theta}}(x_{t}|\bold{x}_{\bold{S}}), $$
$\forall~ \bold{x}_{\bold{K}\cup\{s\}}\in \{0,1,2,\ldots\}^{|\bold{K}|+1}.$
Therefore, $X_s\perp X_t|\bold{X}_{\bold{S}}$.
\item if there exists $t\in \bold{K}$ such that  $X_s\perp X_t|\bold{X}_{\bold{S}}$ but $\theta_{st|\bold{K}}\ne 0$, we have 
$$\mathbb{P}_{\boldsymbol{\theta}}(x_{s}x_t|\bold{x}_{\bold{S}})=\mathbb{P}_{\boldsymbol{\theta}}(x_{s}|\bold{x}_{\bold{S}})\mathbb{P}_{\boldsymbol{\theta}}(x_{t}|\bold{x}_{\bold{S}}), \quad \forall~ \bold{x}_{\bold{K}\cup\{s\}}\in \{0,1,2,\ldots\}^{|\bold{K}|+1}.$$ 
Consider $x_j=0, \; \forall j\in \bold{K}\backslash \{{t}\}$, and $x_t=1$, then
$$\mathbb{P}_{\boldsymbol{\theta}}(x_{s},x_t|\bold{x}_{\bold{S}})=\mathbb{P}_{\boldsymbol{\theta}}(x_{s}|x_t,\bold{x}_{\bold{S}})\mathbb{P}_{\boldsymbol{\theta}}(x_{t}|\bold{x}_{\bold{S}})= \exp\{-e^{\theta_{st|\bold{K}}}\} \mathbb{P}_{\boldsymbol{\theta}}(x_{t}|\bold{x}_{\bold{S}})\ne \mathbb{P}_{\boldsymbol{\theta}}(x_{s}|\bold{x}_{\bold{S}})\mathbb{P}_{\boldsymbol{\theta}}(x_{t}|\bold{x}_{\bold{S}}),
$$
contradictory with $X_s\perp X_t|\bold{X}_{\bold{S}}$. Hence, $\theta_{st|\bold{K}}= 0.$

 \end{itemize}
    
\end{proof}

\section{The PK2 algorithm}\label{PK2}
%
%
\noindent
K2 \citep{cooper1992bayesian} is a structure learning algorithm  which owes its popularity to its computational efficiency  and simplicity. Technically speaking, it is defined only for  categorical variables; a natural choice for using K2 on non-categorical data is categorization, a choice that highly depends on the algorithm chosen to categorize. In our empirical study, we develop a modification of K2, named PK2 (Poisson based on K2 algorithm), that substitutes the K2 score with an alternative scoring criterion respectful of the nature of the data.

As alternative scoring criteria, we consider the Bayesian information criterion (BIC) \citep{schwarz1978estimating} and the Akaike information criterion (AIC) \citep{akaike1974new}, that both balance  goodness of fit and model parsimony. 
PK2 works in two steps: (i) forward phase, and (ii) backward phase. In the forward phase,
we employ the same search strategy as in the K2 algorithm. Then, the result from the first phase is taken as the input of the second phase. Here, we use the score to prune the estimated graph by deleting edges that gain larger scores. 

The asymptotic property of PK2 is proved by 
results of \cite{chickering2002optimal}, that states that under the assumption of a generative distribution  perfect w.r.t. the structure from which the data was generated, the greedy search  identifies the true structure up to an equivalent class as the number of observations goes to infinity when a consistent scoring criterion is used. 

As both BIC and AIC are consistent scoring criteria \citep{haughton1988choice}, and the perfect property is guaranteed by Markov and Faithfulness assumption,  Lemma 10 of \cite{chickering2002optimal} holds for  PK2. Moreover, the Poisson model is identifiable (see Appendix A). Therefore, the estimator in PK2 algorithm  convergences asymptotically to the true graph.

{We note that PKAIC (PK2 using AIC) usually performs more poorly than the PKBIC (PK2 using BIC). This result is not surprising since the PKAIC using the AIC criterion  penalizes  less strongly than  the BIC one. As a consequence, the PKAIC results in graphs with more edges than PKBIC does,  see  \cite{wagenmakers2004aic}. For this reason, we consider only PKBIC in the main paper.}

			

\section{Proofs}\label{proofs} 
Here, we provide proofs of results needed to show consistency of our algorithm. 
We begin by introducing  results for the case $\mathbf{K}=V\backslash\{s\}$. Then, the same results for general case  {$\mathbf{K}\subset V\backslash\{s\}$} are deduced.

Before going into details, we first prove the following Lemma, used in the proof of Theorem \ref{dl1}.
{
\begin{bd}\label{Fisher}
	Assume 2 - 5. Then, for any $\delta>0$, we have
	\begin{eqnarray*}
		\mathbb{P}_{\boldsymbol{\theta}}\left(\Lambda_{\max}\big(\frac{1}{n}\sum_{i=1}^{n} (\mathbf{X}_{V\backslash \{s\}}^{(i)})^T\mathbf{X}_{V\backslash \{s\}}^{(i)}\big)\le \lambda_{\max}+\delta\right)&\ge&1- \exp\left\{-c_1'\frac{\delta^2n}{p^2n^{2\gamma}\log^4 n}+c_2'\log p\right\}\\
		&&~-c_2n\kappa(n,\gamma)-c_1n^{-2}\\
		\mathbb{P}_{\boldsymbol{\theta}}\left(\Lambda_{\min}\left(Q_s(\boldsymbol{\theta_s})\right)\ge \lambda_{\min}-\delta\right)&\ge& 1-\exp\left\{-c_1'\frac{\delta^2n}{p^2n^{2\gamma}\log^4 n}+c_2'\log p\right\}\\
		&&~-c_2n\kappa(n,\gamma)-c_1n^{-2}.
	\end{eqnarray*}
\end{bd}
}
\begin{proof}
	We have
	\begin{eqnarray*}
		\Lambda_{\min}(I_s(\boldsymbol{\theta}_s))&=&\min_{\|\mathbf{y}\|_2=1}\mathbf{y}I_s(\boldsymbol{\theta}_s)\mathbf{y}^T\\
		&=&\min_{\|\mathbf{y}\|_2=1}\left\{\mathbf{y}Q_s(\boldsymbol{\theta}_s)\mathbf{y}^T+\mathbf{y}(I_s(\boldsymbol{\theta}_s)-Q_s(\boldsymbol{\theta}_s))\mathbf{y}^T\right\}\\
		&\le& \mathbf{y}Q_s(\boldsymbol{\theta}_s)\mathbf{y}^T+\mathbf{y}(I_s(\boldsymbol{\theta}_s)-Q_s(\boldsymbol{\theta}_s))\mathbf{y}^T,
	\end{eqnarray*}
	where $\mathbf{y}\in \mathbb{R}^{p-1}$ is an arbitrary vector with unit norm. Hence,
	\begin{equation}\label{Hessian}
	\Lambda_{\min}(Q_s(\boldsymbol{\theta}_s))\ge \Lambda_{\min}(I_s(\boldsymbol{\theta}_s))-\max_{\|\mathbf{y}\|_2=1}\mathbf{y}\big(I_s(\boldsymbol{\theta}_s)-Q_s(\boldsymbol{\theta}_s)\big)\mathbf{y}^T\ge \lambda_{\min}-|||I_s(\boldsymbol{\theta}_s)-Q_s(\boldsymbol{\theta}_s)|||_2.
	\end{equation}
	It remains to control the spectral norm $|||I_s(\boldsymbol{\theta}_s)-Q_s(\boldsymbol{\theta}_s)|||_2$.
	The $(j,k)$ element of the matrix $Z^n=Q_s(\boldsymbol{\theta}_s)-I_s(\boldsymbol{\theta}_s)$ can be written as
	\begin{eqnarray*}
		Z_{jk}^n(\boldsymbol{\theta}_s)&=&\frac{1}{n}\sum_{i=1}^nD(\langle\boldsymbol{\theta}_s,{X}^{(i)}_{V\backslash \{s\}}\rangle ){X}_{ij}X_{ik}-\mathbb{E}_{\boldsymbol{\theta}}\left(D\big(\langle\boldsymbol{\theta}_s,{X}_{V\backslash \{s\}}\rangle \big){X}_jX_k\right)\\
		&=&\frac{1}{n}\sum_{i=1}^n Y_i-\mathbb{E}_{\boldsymbol{\theta}}\left(\frac{1}{n}\sum_{i=1}^n Y_i\right),
	\end{eqnarray*}
	where $Y_i=D(\langle\boldsymbol{\theta}_s,{X}^{(i)}_{V\backslash \{s\}}\rangle ){X}_{ij}X_{ik},~i=1,\ldots,n$ are independent. 
	
	Let $\zeta_1=\left\{\max_iX_{it}<3\log n\right\},$ and
	$\zeta_4=\big\{\max_i\langle\boldsymbol{\theta}_s,\mathbf{X}_{V\backslash \{s\}}^{(i)}\rangle < \gamma\log n \big\}$. 
	From Assumption 4, we get
	$$\mathbb{P}_{\boldsymbol{\theta}}(\zeta^c_1)\le c_1nn^{-3}=c_1n^{-2}.$$ Moreover, by Assumption 5, we have
	\begin{eqnarray*}
		\mathbb{P}_{\boldsymbol{\theta}}\big(\zeta_4^c\big)
		&=& \mathbb{P}_{\boldsymbol{\theta}}\big(\max_i\langle\boldsymbol{\theta}_s,\mathbf{X}_{V\backslash \{s\}}^{(i)}\rangle \ge \gamma\log n\big)\\
		&=&\mathbb{P}_{\boldsymbol{\theta}}\big(\max_i\langle\boldsymbol{\theta}_s,\mathbf{X}_{V\backslash \{s\}}^{(i)}\rangle \ge \gamma\log n,\zeta_1\big)+\mathbb{P}_{\boldsymbol{\theta}}\big(\max_i\langle\boldsymbol{\theta}_s,\mathbf{X}_{V\backslash \{s\}}^{(i)}\rangle \ge \gamma\log n,\zeta_1^c\big)\\
		&\le&n\mathbb{P}_{\boldsymbol{\theta}}\big(\langle\boldsymbol{\theta}_s,\mathbf{X}_{V\backslash \{s\}}^{(i)}\rangle \ge \gamma\log n,\zeta_1\big)+\mathbb{P}_{\boldsymbol{\theta}}(\zeta_1^c)\\
		&\le&n\mathbb{P}_{\boldsymbol{\theta}}\big(\nu+\langle\boldsymbol{\theta}_s,\mathbf{X}_{V\backslash \{s\}}^{(i)}\rangle\ge \gamma\log n\big)+\mathbb{P}_{\boldsymbol{\theta}}(\zeta_1^c)\\
		&\le&c_2n\kappa(n,\gamma)+c_1n^{-2}.
	\end{eqnarray*}
	Conditioning on $\zeta_1,\zeta_4$, we get
	$$|Y_i|\le 9n^\gamma\log^2n.$$
	Then, by the Azuma-Hoeffding inequality  \citep[Theorem 2 in][]{hoeffding1963probability}, for any $\epsilon>0$, we obtain
	$$\mathbb{P}_{\boldsymbol{\theta}}\left((Z_{ij}^n)^2\ge \epsilon^2|\zeta_1,\zeta_4\right)=\mathbb{P}_{\boldsymbol{\theta}}\left(|Z_{ij}^n|\ge \epsilon|\zeta_1,\zeta_4\right)\le 2\exp\left(-\frac{\epsilon^2n}{18n^{2\gamma}\log^4n}\right).$$
	Let {$\epsilon=\delta/p$}, then 
	{
	\begin{eqnarray}\label{Fisherdistance}
	\mathbb{P}_{\boldsymbol{\theta}}\left(|||I_s(\boldsymbol{\theta}_s)-Q_s(\boldsymbol{\theta}_s)|||_2\ge \delta\right)&\le&\mathbb{P}_{\boldsymbol{\theta}}\bigg(\big(\sum_{j,k\ne s}(Z_{jk}^n)^2\big)^{1/2}\ge \delta \bigg)\\
	&\le&\mathbb{P}_{\boldsymbol{\theta}}\bigg(\big(\sum_{j,k\ne s}(Z_{jk}^n)^2\big)^{1/2}\ge \delta|\zeta_1,\zeta_4 \bigg)+\mathbb{P}_{\boldsymbol{\theta}}(\zeta_4^c)+\mathbb{P}_{\boldsymbol{\theta}}(\zeta_1^c)\nonumber\\
	&\le& 2p^2\exp\left\{-\frac{\delta^2n}{18p^2n^{2\gamma}\log^4 n}\right\}+c_2n\kappa(n,\gamma)+c_1n^{-2}\nonumber\\
	&\le& \exp\left\{-c_1'\frac{\delta^2n}{p^2n^{2\gamma}\log^4 n}+c_2'\log p\right\}+c_2n\kappa(n,\gamma)+c_1n^{-2}.\nonumber
	\end{eqnarray}
	From Equation \eqref{Hessian} and \eqref{Fisherdistance}, we have
	$$\mathbb{P}_{\boldsymbol{\theta}}\left(\Lambda_{\min}(Q_s(\boldsymbol{\theta}_s))\ge \lambda_{\min}-\delta\right)\ge 1-\exp\left\{-c_1'\frac{\delta^2n}{p^2n^{2\gamma}\log^4 n}+c_2'\log p\right\}-c_2n\kappa(n,\gamma)-c_1n^{-2}.$$
	Similarly,
	\begin{eqnarray*}
		\mathbb{P}_{\boldsymbol{\theta}}\left(\Lambda_{\max}\left[\frac{1}{n}\sum_{i=1}^{n} (\mathbf{X}_{V\backslash \{s\}}^{(i)})^T\mathbf{X}_{V\backslash \{s\}}^{(i)}\bigg]\le \lambda_{\max}+\delta\right]\right)&\ge&  1-\exp\left\{-c_1'\frac{\delta^2n}{p^2n^{2\gamma}\log^4 n}+c_2'\log p\right\}\\
		&&~-c_2n\kappa(n,\gamma)-c_1n^{-2}.
	\end{eqnarray*}}
\end{proof}

\begin{md}\label{pro1}
	Assume 1- 5. Then, for any $\delta>0$
	$$\mathbb{P}_{\boldsymbol{\theta}}(\|\nabla l(\boldsymbol{\theta}_s,\mathbf{X}_s;\mathbf{X}_{V\backslash \{s\}})\|_{\infty}\ge \delta)\le \exp(-c_3 n)+c_2n\kappa(n,\gamma)+c_1n^{-2},~\forall~ \boldsymbol{\theta}\in \boldsymbol{\Theta},$$
	when  $n\longrightarrow\infty$.
\end{md}

\begin{proof}
	The $t$-partial derivative of the node conditional log-likelihood $l(\boldsymbol{\theta}_s,\mathbf{X}_s;\mathbf{X}_{V\backslash \{s\}})$ is:
	\begin{eqnarray*}
		W_t=\nabla_t l(\boldsymbol{\theta}_s,\mathbf{X}_s;\mathbf{X}_{V\backslash \{s\}}) &=& \frac{1}{n}\sum_{i=1}^{n}\left[-X_{is}X_{it}+X_{it}D(\langle\boldsymbol{\theta}_s,\mathbf{X}_{V\backslash \{s\}}^{(i)}\rangle )\right]
	\end{eqnarray*}
	Let $V_{is}(t)= X_{is}X_{it}-X_{it}D(\langle\boldsymbol{\theta}_s,\mathbf{X}_{V\backslash \{s\}}^{(i)}\rangle )$. We have,
	\begin{eqnarray}\label{derivative}
	\mathbb{P}_{\boldsymbol{\theta}}(\|W\|_\infty >\delta)&=& \mathbb{P}_{\boldsymbol{\theta}}(\max_{t\in V\backslash \{s\}}|\nabla_t l(\boldsymbol{\theta}_s,X_s,\mathbf{X}_{V\backslash\{s\}})|> \delta)\\
	&\le&\mathbb{P}_{\boldsymbol{\theta}}\bigg( \max_{t\in V\backslash \{s\}}\bigg|\dfrac{1}{n}\sum_{i=1}^{n}V_{is}(t)\bigg|>\delta|\zeta_1,\zeta_2\bigg)+\mathbb{P}_{\boldsymbol{\theta}}(\zeta_1^c)+\mathbb{P}_{\boldsymbol{\theta}}(\zeta_2^c)\nonumber\\
	&\le& p \left[\mathbb{P}_{\boldsymbol{\theta}}\left(\dfrac{1}{n}\sum_{i=1}^{n}V_{is}(t)>\delta|\zeta_1,\zeta_2\right)+\mathbb{P}_{\boldsymbol{\theta}}\left(-\dfrac{1}{n}\sum_{i=1}^{n}V_{is}(t)>\delta|\zeta_1,\zeta_2\right)\right]\nonumber\\
	&&~+\mathbb{P}_{\boldsymbol{\theta}}(\zeta_1^c)+\mathbb{P}_{\boldsymbol{\theta}}(\zeta_2^c)\nonumber\\
	&\le& p\left[ \dfrac{\mathbb{E}_{\boldsymbol{\theta}}\left[\prod_{i=1}^{n}\exp\left(hV_{is}(t)\right)|\zeta_1,\zeta_2\right]}{\exp(nh\delta)}+\dfrac{\mathbb{E}_{\boldsymbol{\theta}}\left[\prod_{i=1}^{n}\exp\left(-hV_{is}(t)\right)|\zeta_1,\zeta_2\right]}{\exp(nh\delta)}\right]\nonumber\\
	&&~+\mathbb{P}_{\boldsymbol{\theta}}(\zeta_1^c)+\mathbb{P}_{\boldsymbol{\theta}}(\zeta_2^c)\nonumber\\
	&=& p \left[\dfrac{\prod_{i=1}^{n}\mathbb{E}_{\boldsymbol{\theta}}\left[\exp\left(hV_{is}(t)\right)|\zeta_1,\zeta_2\right]}{\exp(nh\delta)}+\dfrac{\prod_{i=1}^{n}\mathbb{E}_{\boldsymbol{\theta}}\left[\exp\left(-hV_{is}(t)\right)|\zeta_1,\zeta_2\right]}{\exp(nh\delta)}\right]\nonumber\\
	&&~+\mathbb{P}_{\boldsymbol{\theta}}(\zeta_1^c)+\mathbb{P}_{\boldsymbol{\theta}}(\zeta_2^c)\nonumber\\
	&=& p\bigg[\exp\bigg\{ \sum_{i=1}^{n}\log\mathbb{E}_{\boldsymbol{\theta}}\left[\exp\left(hV_{is}(t)\right)|\zeta_1,\zeta_2\right]-nh\delta\bigg\}\nonumber\\
	&&+\exp\bigg\{\sum_{i=1}^{n}\log\mathbb{E}_{\boldsymbol{\theta}}\left[\exp\left(-hV_{is}(t)\right)|\zeta_1,\zeta_2\right] -nh\delta\bigg\}\bigg]+\mathbb{P}_{\boldsymbol{\theta}}(\zeta_1^c)+\mathbb{P}_{\boldsymbol{\theta}}(\zeta_2^c),\nonumber
	\end{eqnarray}
	for some $h>0$, where $\zeta_1=\left\{\max_iX_{it}<3\log n\right\},$
	$\zeta_2=\big\{\max_i\big(vhX_{it}+\langle\boldsymbol{\theta}_s,\mathbf{X}_{V\backslash \{s\}}^{(i)}\rangle\big) < \gamma\log n \big\}$. 
	From Lemma \ref{Fisher}, we have
	$\mathbb{P}_{\boldsymbol{\theta}}\big(\zeta_1^c\big)\le c_1n^{-2}.$
	Compute $\zeta_2$ similarly to $\zeta_4$ in Lemma \ref{Fisher}, we get
	
	\begin{eqnarray*}
		\mathbb{P}_{\boldsymbol{\theta}}\big(\zeta_2^c\big)
		&=& \mathbb{P}_{\boldsymbol{\theta}}\big(\max_i\big(vhX_{it}+\langle\boldsymbol{\theta}_s,\mathbf{X}_{V\backslash \{s\}}^{(i)}\rangle\big) \ge \gamma\log n\big)\\
		&\le&\mathbb{P}_{\boldsymbol{\theta}}\big(\max_i\big(vhX_{it}+\langle\boldsymbol{\theta}_s,\mathbf{X}_{V\backslash \{s\}}^{(i)}\rangle\big) \ge \gamma\log n,\zeta_1\big)+\mathbb{P}_{\boldsymbol{\theta}}(\zeta_1^c)\\
		&\le&\mathbb{P}_{\boldsymbol{\theta}}\big(\max_i\big(\nu+\langle\boldsymbol{\theta}_s,\mathbf{X}_{V\backslash \{s\}}^{(i)}\rangle\big) \ge \gamma\log n,\zeta_1\big)+\mathbb{P}_{\boldsymbol{\theta}}(\zeta_1^c)\\
		&\le&c_2n\kappa(n,\gamma)+c_1n^{-2},
	\end{eqnarray*}
	provided that  $|h|<\dfrac{\nu}{3\log n}$.
	We therefore need to compute $\sum_{i=1}^{n}\log\mathbb{E}_{\boldsymbol{\theta}}\left[\exp\left(hV_{is}(t)\right)|\zeta_1,\zeta_2\right]$, and $\sum_{i=1}^{n}\log\mathbb{E}_{\boldsymbol{\theta}}\left[\exp\left(-hV_{is}(t)\right)|\zeta_1,\zeta_2\right]$. 
	First, we have
	\begin{eqnarray}\label{moment}
	\mathbb{E}_{\boldsymbol{\theta}_s}\left[\exp(hV_{is}(t))|\mathbf{x}_{V\backslash \{s\}}^{(i)}\right] &=& \sum_{x_{is}=0}^{\infty}\exp\bigg\{h[x_{is}x_{it}-x_{it}D(\langle\boldsymbol{\theta}_s,\mathbf{x}_{V\backslash \{s\}}^{(i)}\rangle )]\nonumber\\
	&&~+x_{is} \langle\boldsymbol{\theta}_s,\mathbf{x}_{V\backslash \{s\}}^{(i)}\rangle
	- \log x_{is}!	- D(\langle\boldsymbol{\theta}_s,\mathbf{x}_{V\backslash \{s\}}^{(i)}\rangle)\bigg\}\nonumber\\
	&=&\sum_{x_{is}=0}^{\infty}\exp\bigg\{x_{is}[hx_{it}+\langle\boldsymbol{\theta}_s,\mathbf{x}_{V\backslash \{s\}}^{(i)}\rangle ]- \log x_{is}!\nonumber\\
	&& ~- hx_{it}D(\langle\boldsymbol{\theta}_s,\mathbf{x}_{V\backslash \{s\}}^{(i)}\rangle )- D(\langle\boldsymbol{\theta}_s,\mathbf{x}_{V\backslash \{s\}}^{(i)}\rangle)\bigg\}\nonumber\\
	&=& \exp\bigg\{D(hx_{it}+\langle\boldsymbol{\theta}_s,\mathbf{x}_{V\backslash \{s\}}^{(i)}\rangle)-D(\langle\boldsymbol{\theta}_s,\mathbf{x}_{V\backslash \{s\}}^{(i)}\rangle)\nonumber\\
	&&~ -hx_{it}D(\langle\boldsymbol{\theta}_s,\mathbf{x}_{V\backslash \{s\}}^{(i)}\rangle) \bigg\}\nonumber\\
	&=& \exp\left\{\frac{h^2}{2}(x_{it})^2D(vhx_{it}+\langle\boldsymbol{\theta}_s,\mathbf{x}_{V\backslash \{s\}}^{(i)}\rangle)\right\},\nonumber
	\end{eqnarray}
	$\text{for some } v\in [0,1],$ where we move from line 2 to line 3 by applying $e^\lambda=\sum_{x=0}^\infty \dfrac{e^x}{x!}$, and from line 3 to line 4 by using a Taylor expansion for function $D(.)$ at $\langle\boldsymbol{\theta}_s,\mathbf{x}_{V\backslash \{s\}}^{(i)}\rangle$. 
	Therefore,
	\begin{eqnarray}\label{pospart}
	&&\sum_{i=1}^{n}\log\mathbb{E}_{\boldsymbol{\theta}}\left[\exp\left(hV_{is}(t)\right)|\zeta_1,\zeta_2\right]\\
	&&~=\sum_{i=1}^{n}\log\mathbb{E}_{\boldsymbol{\theta}_{V\backslash \{s\}}}\bigg[
	\mathbb{E}_{\boldsymbol{\theta}_s}\left[\exp(hV_{is}(t))|\mathbf{x}_{V\backslash \{s\}}^{(i)}\right]|\zeta_1,\zeta_2\bigg]\nonumber\\
	&&~=\sum_{i=1}^{n}\log\mathbb{E}_{\boldsymbol{\theta}_{V\backslash \{s\}}}\bigg[\exp\left\{\frac{h^2}{2}(X_{it})^2D(vhX_{it}+\langle\boldsymbol{\theta}_s,\mathbf{X}_{V\backslash \{s\}}^{(i)}\rangle)\right\}|\zeta_1,\zeta_2\bigg]\nonumber\\
	&&~\le\sum_{i=1}^{n}\log\mathbb{E}_{\boldsymbol{\theta}_{V\backslash \{s\}}}\bigg[\exp\left\{\frac{h^2}{2}(3 \log n)^2D(\gamma\log n)\right\}|\zeta_1,\zeta_2\bigg]\nonumber\\
	&&~=n\left\{\frac{h^2}{2}(3 \log n)^2n^{\gamma}\right\}\nonumber\\
	\end{eqnarray}
	Similarly,
	\begin{eqnarray}\label{moment1}
	\mathbb{E}_{\boldsymbol{\theta}_s}\left[\exp(-hV_{is}(t))|\mathbf{x}_{V\backslash \{s\}}^{(i)}\right] 
	&=& \exp\left\{\frac{h^2}{2}(x_{it})^2D(-vhx_{it}+\langle\boldsymbol{\theta}_s,\mathbf{x}_{V\backslash \{s\}}^{(i)}\rangle)\right\}.\nonumber
	\end{eqnarray}
	Hence,
	\begin{eqnarray}\label{negpart}
	\sum_{i=1}^{n}\log\mathbb{E}_{\boldsymbol{\theta}}\left[\exp\left(-hV_{is}(t)\right)|\zeta_1,\zeta_2\right]&\le&
	n\left\{\frac{h^2}{2}(3 \log n)^2n^{\gamma}\right\}\nonumber\\
	\end{eqnarray}
	Let $h=\dfrac{\delta}{9n^{\gamma}\log^2n}$, from \eqref{derivative}--\eqref{negpart}, we get
	\begin{eqnarray*}
		\mathbb{P}_{\boldsymbol{\theta}}(\|W\|_\infty >\delta)&\le& p\bigg[\exp\bigg\{ n\left\{\frac{h^2}{2}(3 \log n)^2n^{\gamma}\right\}-nh\delta\bigg\}
		\\&&+\exp\bigg\{n\left\{\frac{h^2}{2}(3 \log n)^2n^{\gamma}\right\} -nh\delta\bigg\}\bigg]+c_1n^{-2}+c_2n\kappa(n,\gamma)\\
		&\le&  2p\bigg[\exp\bigg\{ \dfrac{-n\delta^2}{18n^\gamma\log^2n}\bigg\}\bigg]+c_1n^{-2}+c_2n\kappa(n,\gamma)\\
		&\le&\exp\{ -c_3n\}+c_1n^{-2}+c_2n\kappa(n,\gamma),
	\end{eqnarray*}
	provided that $p<\dfrac{1}{2}\exp\left(\dfrac{n^{1-\gamma}\delta^2}{36\log^2 n}\right)$.
	
\end{proof}

\begin{dl}\label{dl1}
	Assume 1- 5. Then, there exists a non-negative decreasing sequence $\delta_n\rightarrow 0$, such that
	$$\mathbb{P}_{\boldsymbol{\theta}}(\|\hat{\boldsymbol{\theta}}_s-\boldsymbol{\theta}_s\|_2\le \delta_n)\ge 1-\exp\{ -c_4n^{1-2\gamma}\}-c_1n^{-2}-c_2n\kappa(n,\gamma),~\forall~ \boldsymbol{\theta}\in \boldsymbol{\Theta},$$
	when $n\rightarrow \infty$.
\end{dl}

\begin{proof}
	For a fixed design $\mathbb{X}$,  define $G: \mathbb{R}^{p-1}\longrightarrow \mathbb{R}$ as
	$$G({\mathbf{u}}, \mathbb{X}_s;\mathbb{X}_{V\backslash\{s\}})=l(\boldsymbol{\theta}_s+{\mathbf{u}},\mathbb{X}_s;\mathbb{X}_{V\backslash \{s\}})- l(\boldsymbol{\theta}_s,\mathbb{X}_s;\mathbb{X}_{V\backslash \{s\}}).$$
	Then, $G(0, \mathbb{X}_s;\mathbb{X}_{V\backslash\{s\}})=0$. Moreover, let $\hat{\mathbf{u}}= \hat{\boldsymbol{\theta}}_s-\boldsymbol{\theta}_s$, we have $G(\hat{\mathbf{u}}, \mathbb{X}_s;\mathbb{X}_{V\backslash\{s\}})\le 0$.
	
	Given a value $\epsilon>0$, if $G({\mathbf{u}}, \mathbb{X}_s;\mathbb{X}_{V\backslash\{s\}})>0,~\forall \mathbf{u}\in \mathbb{R}^{p-1}$ such that $\|\mathbf{u}\|_2=\epsilon$, then $\|\hat{\mathbf{u}}\|_2\le \epsilon$, since $G(., \mathbb{X}_s;\mathbb{X}_{V\backslash\{s\}})$ is a convex function. Therefore,
	$$\mathbb{P}_{\boldsymbol{\theta}}\left(\|\hat{\boldsymbol{\theta}}_s-\boldsymbol{\theta}_s\|_2\le \epsilon\right)\ge\mathbb{P}_{\boldsymbol{\theta}}\left(G({\mathbf{u}}, {X}_s;\mathbf{X}_{V\backslash\{s\}})>0),~\forall \mathbf{u}\in \mathbb{R}^{p-1} \text{ such that } \|\mathbf{u}\|_2=\epsilon\right).$$
	Using Taylor expansion of the node conditional log-likelihood at $\boldsymbol{\theta}_s$, we have
	\begin{eqnarray}\label{funcg}
	G(\mathbf{u}) &=& l(\boldsymbol{\theta}_s+\mathbf{u},{X}_s;\mathbf{X}_{V\backslash \{s\}})- l(\boldsymbol{\theta}_s,{X}_s;\mathbf{X}_{V\backslash \{s\}})\\
	&=& \nabla l(\boldsymbol{\theta}_s,{X}_s;\mathbf{X}_{V\backslash \{s\}})) \mathbf{u}^T+\mathbf{u}[\nabla^2(l(\boldsymbol{\theta}_s+v\mathbf{u},{X}_s;\mathbf{X}_{V\backslash \{s\}})]\mathbf{u}^T\nonumber,
	\end{eqnarray}
	for some $v\in [0,1]$. Let 
	\begin{eqnarray*}
		q&=& \Lambda_{\min}(\nabla^2(l(\boldsymbol{\theta}_s+v\mathbf{u},{X}_s;\mathbf{X}_{V\backslash \{s\}})))\\
		&\ge& \min_{v\in[0,1]}\Lambda_{\min}(\nabla^2(l(\boldsymbol{\theta}_s+v\mathbf{u},{X}_s;\mathbf{X}_{V\backslash \{s\}})))\\
		&=& \min_{v\in[0,1]}\Lambda_{\min} \left[\frac{1}{n}\sum_{i=1}^{n} D(\langle\boldsymbol{\theta}_s+v\mathbf{u},\mathbf{X}^{(i)}_{V\backslash \{s\}}\rangle )(\mathbf{X}_{V\backslash \{s\}}^{(i)})^T\mathbf{X}_{V\backslash \{s\}}^{(i)}\right].
	\end{eqnarray*}
	Using Taylor expansion for $D(\langle\boldsymbol{\theta}_s+v\mathbf{u},\mathbf{X}^{(i)}_{V\backslash \{s\}}\rangle )$ at $\langle\boldsymbol{\theta}_s,\mathbf{X}^{(i)}_{V\backslash \{s\}}\rangle$, we have
	\begin{eqnarray*}
		&&\frac{1}{n}\sum_{i=1}^{n} D(\langle\boldsymbol{\theta}_s+v\mathbf{u},\mathbf{X}^{(i)}_{V\backslash \{s\}}\rangle )(\mathbf{X}_{V\backslash \{s\}}^{(i)})^T\mathbf{X}_{V\backslash \{s\}}^{(i)})\\ 
		&&=\quad\frac{1}{n}\sum_{i=1}^{n} D(\langle\boldsymbol{\theta}_s,\mathbf{X}^{(i)}_{V\backslash \{s\}}\rangle )(\mathbf{X}_{V\backslash \{s\}}^{(i)})^T\mathbf{X}_{V\backslash \{s\}}^{(i)}+\\
		&&\qquad
		\frac{1}{n}\sum_{i=1}^{n} D(\langle\boldsymbol{\theta}_s+v'\mathbf{u},\mathbf{X}^{(i)}_{V\backslash \{s\}}\rangle )[v\mathbf{u}(\mathbf{X}_{V\backslash \{s\}}^{(i)})^T][(\mathbf{X}_{V\backslash \{s\}}^{(i)})^T\mathbf{X}_{V\backslash \{s\}}^{(i)}],
	\end{eqnarray*}
	for some $v'\in [0,1]$. Hence, 
	\begin{eqnarray*}
		q&\ge& \Lambda_{\min}\left[\frac{1}{n}\sum_{i=1}^{n} D(\langle\boldsymbol{\theta}_s,\mathbf{X}^{(i)}_{V\backslash \{s\}}\rangle )(\mathbf{X}_{V\backslash \{s\}}^{(i)})^T\mathbf{X}_{V\backslash \{s\}}^{(i)}\right]\\
		&&-\max_{v'\in[0,1]}\Lambda_{\max}\bigg[\frac{1}{n}\sum_{i=1}^{n} D(\langle\boldsymbol{\theta}_s+v'\mathbf{u},\mathbf{X}^{(i)}_{V\backslash \{s\}}\rangle )[\mathbf{u}(\mathbf{X}_{V\backslash \{s\}}^{(i)})^T]\big[(\mathbf{X}_{V\backslash \{s\}}^{(i)})^T\mathbf{X}_{V\backslash \{s\}}^{(i)}\big]\bigg].
	\end{eqnarray*}
	Define a new event
	\begin{eqnarray*}
		\zeta_3 &=&\{ (\langle\boldsymbol{\theta}_s+v'\mathbf{u},\mathbf{X}^{(i)}_{V\backslash \{s\}}\rangle)\le \gamma\log n \}.
	\end{eqnarray*}
	Similarly to the event $\zeta_2$, we have $\mathbb{P}_{\boldsymbol{\theta}}(\zeta_3^c)\le c_2n\kappa(n,\gamma)+c_1n^{-2}$, provided that $\epsilon\le \dfrac{1}{3p\log n}$. As consequence, $|D(\langle\boldsymbol{\theta}_s+v'\mathbf{u},\mathbf{X}^{(i)}_{V\backslash \{s\}}\rangle)|\le n^{\gamma}$, with probability at least $1-c_2n\kappa(n,\gamma)-c_1n^{-2}$. Fixed $\delta=\dfrac{\lambda_{\min}}{8}$  in  Lemma \ref{Fisher}, and conditioned on $\zeta_1,\zeta_3$ we have

	\begin{eqnarray*}
		q&\ge& \lambda_{\min}-\delta-\max_{v'\in[0,1]}\Lambda_{\max}\bigg[\frac{1}{n}\sum_{i=1}^{n} D(\langle\boldsymbol{\theta}_s+v'\mathbf{u},\mathbf{X}^{(i)}_{V\backslash \{s\}}\rangle )[\mathbf{u}(\mathbf{X}_{V\backslash \{s\}}^{(i)})^T]\big[(\mathbf{X}_{V\backslash \{s\}}^{(i)})^T\mathbf{X}_{V\backslash \{s\}}^{(i)}\big]\bigg]\\
		&\ge& \lambda_{\min}-\delta-\max_{v'\in[0,1]}\big|D(\langle\boldsymbol{\theta}_{V\backslash \{s\}}+v'\mathbf{u},\mathbf{X}^{(i)}_{V\backslash \{s\}}\rangle )\big|\big|\mathbf{u}(\mathbf{X}_{V\backslash \{s\}}^{(i)})^T\big|\Lambda_{\max}\bigg[\frac{1}{n}\sum_{i=1}^{n} (\mathbf{X}_{V\backslash \{s\}}^{(i)})^T\mathbf{X}_{V\backslash \{s\}}^{(i)}\bigg]\\
		&\ge& \lambda_{\min}-2\delta- n^{\gamma}3\log n\sqrt{p}\|\mathbf{u}\|_2\lambda_{\max}\\
		&=& \lambda_{\min}-2\delta- 3\sqrt{p}\epsilon\lambda_{\max}n^{\gamma}\log n \\
		&\ge& \dfrac{\lambda_{\min}}{2},\quad \text{provided that } \epsilon< \dfrac{\lambda_{\min}}{12\sqrt{p}\lambda_{\max}n^{\gamma}\log n},
	\end{eqnarray*}
	with probability at least	$1- \exp\left\{-c_4n\right\}-c_2n\kappa(n,\gamma)-c_1n^{-2}$. Let $\delta=\dfrac{\lambda_{\min}}{2}\epsilon$ in Proposition \ref{pro1}, we have
	$$\nabla_t l(\boldsymbol{\theta}_s,{X}_s;\mathbf{X}_{V\backslash \{s\}}))>-\dfrac{\lambda_{\min}}{2}\epsilon,$$
	with probability at least $1-\exp\{ -c_3n\}-c_2n\kappa(n,\gamma)-c_1n^{-2}$, 
	provided that
	$p<\dfrac{1}{2}\exp\left(\dfrac{n^{1-\gamma}\lambda_{\min}^2\epsilon^2}{144\log^2 n}\right).$
	Combining with the inequality of $q$, we have
	\begin{eqnarray}
	G(\mathbf{u})&=& \nabla l(\boldsymbol{\theta}_s,{X}_s;\mathbf{X}_{V\backslash \{s\}})) \mathbf{u}^T+\mathbf{u}[\nabla^2(l(\boldsymbol{\theta}_s+v\mathbf{u},_s,{X}_s;\mathbf{X}_{V\backslash \{s\}}))]\mathbf{u}^T\nonumber\\
	&>& -\dfrac{\lambda_{\min}}{2}\epsilon^2+\dfrac{\lambda_{\min}}{2}\epsilon^2=0
	\end{eqnarray}
	provided that $ \epsilon< \dfrac{\lambda_{\min}}{12\sqrt{p}\lambda_{\max}n^{\gamma}\log n}$, and $p<\dfrac{1}{2}\exp\left(\dfrac{n^{1-\gamma}\lambda_{\min}^2\epsilon^2}{144\log^2 n}\right)$. It means that $\|\hat{\mathbf{u}}\|_2<\epsilon$.
	When $n\rightarrow\infty$ we can choose a non- negative decreasing sequence $\delta_n$ such that $\delta_n<\dfrac{\lambda_{\min}}{12\sqrt{p}\lambda_{\max}n^{\gamma}\log n}$, then
	\begin{eqnarray*}
		\mathbb{P}_{\boldsymbol{\theta}}(\|\hat{\boldsymbol{\theta}}_{V\backslash \{s\}}-\boldsymbol{\theta}_{V\backslash \{s\}}\|_2\le \delta_n)&\ge& 1-\exp\{ -c_4n\}-\exp\{ -c_3n\}-c_2n\kappa(n,\gamma)-c_1n^{-2}\\
		&=&1-\exp\{ -cn\}-c_2n\kappa(n,\gamma)-c_1n^{-2}
	\end{eqnarray*}
	when $n\rightarrow \infty$.
\end{proof}

Results for $\mathbf{K}\subset V$ are derived as following.\\
\begin{md}\label{pro11}
	Assume 1- 5. Then, $ \forall~\delta>0$, $\forall~s\in V$ and $\forall~\mathbf{K}\subseteq V\backslash\{s\}$,
	$$\mathbb{P}_{\boldsymbol{\theta}}(\|\nabla l(\boldsymbol{\theta}_{s|\mathbf{K}},{X}_s;{\mathbf{X}_{\mathbf{K}}})\|_{\infty}\ge \delta)\le \exp\{-c_3n\}+c_2\kappa(n,\gamma)+c_1n^{-2},$$
	$~\forall~ \boldsymbol{\theta_{s|\mathbf{K}}}\in \boldsymbol{\Theta},$ when  $n\rightarrow\infty$.
\end{md}

\begin{proof}
	The proof of Proposition \ref{pro11} follows the lines of Proposition \ref{pro1}. We note that the set of explanatory variables $\mathbf{X}_{\mathbf{K}}$ in the generalized linear model $X_s$ given $\mathbf{X}_{\mathbf{K}}$ does not include variables $X_t$, with $t\in V\backslash\{\mathbf{K}\cup\{s\}\}$. Suppose we zero-pad the true parameter $\boldsymbol{\theta}_{s|\mathbf{K}}\in\mathbb{R}^{|\mathbf{K}|}$ to include zero weights over $V\backslash\{\mathbf{K}\cup\{s\}\}$, then the resulting parameter would lie in $\mathbb{R}^{|p-1|}$. 
\end{proof}

\begin{cy}\label{cy1}
	When the maximum number of neighbours that one node is allowed to have is fixed, a control is operated on the cardinality of the set $\mathbf{K}$,  $|\mathbf{K}|\le m+1$. In this case, parameters $\theta_{st|\mathbf{K}}$ are estimated from models that are restricted on subsets of variables with their cardinalities less than or equal to $m+1$. Therefore, $p$ in Proposition \ref{pro1} is replaced by $m+1$. In detail, for all $s\in V$ and any $\delta>0$
	$$\mathbb{P}_{\boldsymbol{\theta}}(\|\nabla l(\boldsymbol{\theta}_{s|\mathbf{K}},{X}_s;{\mathbf{X}_{\mathbf{K}}})\|_{\infty}\ge \delta)\le \exp\{-c_3n\}+c_2\kappa(n,\gamma)+c_1n^{-2},$$
	$~\forall~ \boldsymbol{\theta_{s|\mathbf{K}}}\in \boldsymbol{\Theta},$ provided that $m<\dfrac{1}{2}\exp\left\{\dfrac{n^{1-\gamma}\delta^2}{36\log^2n}\right\}$.
\end{cy}


\noindent
\begin{dl}\label{dl2}
	Assume 1- 5. Then, $\forall~s\in V$ and $\forall~\mathbf{K}\subseteq V\backslash\{s\}$, there exists a non-negative decreasing sequence $\delta_n\rightarrow 0$, such that
	$$\mathbb{P}_{\boldsymbol{\theta}}(\|\hat{\boldsymbol{\theta}}_{s|\mathbf{K}}-\boldsymbol{\theta}_{s|\mathbf{K}}\|_2\le \delta_n)\ge 1-\exp\{ -cn\}-c_2n\kappa(n,\gamma)-c_1n^{-2}, ~\forall~ \boldsymbol{\theta}\in \boldsymbol{\Theta},$$
	when $n\rightarrow \infty$.		
	
\end{dl}
\begin{proof}
	Theorem \ref{dl2} is proved by following the lines of Theorem \ref{dl1}. The difference lies in the set of explanatory variables $\mathbf{X}_{\mathbf{K}}$. As consequence, it is necessary to control the spectral norm of the submatrices,   $Q_{s|\mathbf{K}}(\boldsymbol{\theta}_{s|\mathbf{K}})=\frac{1}{n}\sum_{i=1}^{n} D\left(\langle\boldsymbol{\theta}_{s|\mathbf{K}},{\mathbf{X}^{(i)}_{\mathbf{K}}}\rangle \right)\left({\mathbf{X}^{(i)}_{\mathbf{K}}}\right)^T{\mathbf{X}^{(i)}_{\mathbf{K}}}$ and
	$\left({\mathbf{X}^{(i)}_{\mathbf{K}}}\right)^T{\mathbf{X}^{(i)}_{\mathbf{K}}}$. Here, we employ well-known results on eigenvalue inequalities for a matrix and its submatrix \citep[see, for example, ][]{johnson1981eigenvalue}, that is,
	
	$$\Lambda_{\min}\left[Q_{s|\mathbf{K}}(\boldsymbol{\theta}_{s|\mathbf{K}})\right]\ge \Lambda_{\min}(Q_s(\boldsymbol{\theta}_s))\ge \lambda_{\min}-\delta,$$
	and,
	\begin{eqnarray*}
		\Lambda_{\max}\bigg[\frac{1}{n}\sum_{i=1}^{n}  \left({\mathbf{X}^{(i)}_{\mathbf{K}}}\right)^T{\mathbf{X}^{(i)}_{\mathbf{K}}}\bigg]\le	\Lambda_{\max}\bigg[\frac{1}{n}\sum_{i=1}^{n}  \left({\mathbf{X}^{(i)}_{V\backslash\{s\}}}\right)^T{\mathbf{X}^{(i)}_{V\backslash\{s\}}}\bigg]\le \lambda_{\max}+\delta.
	\end{eqnarray*}
	Then, by performing the same analysis as in the proof of Theorem \ref{dl1} and using the result of Proposition \ref{pro11}, we get the result.
\end{proof}

\begin{cy}\label{cy2}
	In the proof of Theorem \ref{dl2}, we only require  the uniform convergence of a submatrix (restricted on $K$), $Q_{s|\mathbf{K}}(\boldsymbol{\theta}_{s|\mathbf{K}})$, of the sample Fisher information matrix $Q_{s}(\boldsymbol{\theta}_{s})$. Therefore, when the maximum neighbourhood size  is known,   $|\mathbf{K}|\le m+1$, we have  convergence provided that $n>O_p\left(\log m\right)$. In detail, let $I_{s|\mathbf{K}}(\boldsymbol{\theta}_{s|\mathbf{K}})$ be the submatrix of $I_s(\boldsymbol{\theta}_s)$ indexed in $\mathbf{K}$, Equation \eqref{Fisherdistance} becomes
	\begin{eqnarray*}
		\mathbb{P}_{\boldsymbol{\theta}}\left(|||I_{s|\mathbf{K}}(\boldsymbol{\theta}_{s|\mathbf{K}})-Q_{s|\mathbf{K}}(\boldsymbol{\theta}_{s|\mathbf{K}})|||_2\ge \delta\right)&\le&\mathbb{P}_{\boldsymbol{\theta}}\left(\bigg(\sum_{j,k\in \mathbf{K}\backslash\{s\}}(Z_{jk}^n)^2\bigg)^{1/2}\ge \delta \right)\nonumber\\
		&\le& 2m^2\exp\left\{-\frac{\delta^2n}{36p^2n^{2\gamma}\log^4 n}\right\}+c_2\kappa(n,\gamma)+c_1n^{-2}\nonumber\\
		&\le& \exp\{-c_4n\}+c_2\kappa(n,\gamma)+c_1n^{-2},
	\end{eqnarray*}
	provided that $n>O_p\left(\log m\right)$.
\end{cy}
\newpage

\section{ Additional results in empirical study}
\makeatletter

\renewcommand*{\@Alph}[1]{%
	\ifcase#1\or D.1\or D.2\or D.3\or
	D\or E\or F\or G\or H\or I\or J\or
	K\or L\or M\or N\or O\or P\or R\or S\or\v S\or
	T\or U\or V\or W\or X\or
	Y\or Z\or\v Z\else\@ctrerr\fi
}
\makeatother
\renewcommand\thetable{\Alph{table}}
\renewcommand\thefigure{\Alph{figure}}
\subsection{Tables}
Table \ref{table10-chap2}, and Table \ref{table100-chap2} report $ TP, FP, FN, P, R$ and $F_1$ for each of  methods considered in Section 5 of the main paper and for the three types of networks. Two different graph dimensions, i.e., $p=10, 100$, and three graph structures (see Figure 1 and Figure 2 in the main paper) are considered.

\begin{center}
	\fontsize{6.2}{7}\selectfont
		
		\begin{longtable}{l| l | l r r r r r r r}
			\caption{Simulation results from 50 replicates of the DAGs shown in Figure 1 in of the main paper  for $p= 10$ variables with Poisson node conditional distribution. Monte Carlo means (standard deviations) are shown for $TP, FP, FN, P, R$, and $F_1$. The levels of significance of tests $\alpha=2(1-\Phi(n^{0.15}))$. }  \\
			\label{table10-chap2}\\
			\toprule
			Graph&$n$	& Algorithm & $TP$ & $FP$ & $FN$ & $P$ & $R$ &$F_1$&time \\
			\midrule
			\endfirsthead
			\multicolumn{9}{c}%
			{{\bfseries \tablename\ \thetable{} -- continued from previous page}} \\
			\toprule
			Graph&$n$	& Algorithm & $TP$ & $FP$ & $FN$ & $P$ & $R$ &$F_1$ &time \\
			\midrule	
			
			\endhead
			&100 &	PKBIC & 6.300(1.129) & 1.140(1.069) & 2.700(1.129) & 0.857(0.123) & 0.700(0.125) & 0.764(0.106) &  0.147 \\ 
  & & Or-PPGM & 5.520(1.147) & 1.060(1.219) & 3.480(1.147) & 0.859(0.151) & 0.613(0.127) & 0.706(0.114) &  0.880 \\ 
 & & Or-LPGM & 6.260(1.084) & 1.600(1.107) & 2.740(1.084) & 0.807(0.103) & 0.696(0.120) & 0.740(0.093) &  0.063 \\ 
& &  PDN & 7.680(0.741) & 27.740(4.716) & 1.320(0.741) & 0.221(0.037) & 0.853(0.082) & 0.350(0.051) &  0.103 \\ 
& &  ODS & 2.500(1.052) & 1.292(0.967) & 6.500(1.052) & 0.670(0.209) & 0.278(0.117) & 0.384(0.145) &  3.492  \\ 
& &  PC & 3.440(1.128) & 2.780(1.389) & 5.560(1.128) & 0.564(0.169) & 0.382(0.125) & 0.450(0.135) &  0.007 \\ 
 & & MMHC & 2.061(0.944) & 3.551(1.138) & 6.939(0.944) & 0.367(0.144) & 0.229(0.105) & 0.279(0.116) &  0.005 \\  
    & &GBiDAG & 3.380(1.469) & 4.900(1.644) & 5.620(1.469) & 0.409(0.161) & 0.376(0.163) & 0.388(0.157) & 2.728 \\

 &&& & & & &  \\
  Scale-free &200 &  PKBIC & 7.840(0.912) & 0.660(1.062) & 1.160(0.912) & 0.932(0.100) & 0.871(0.101) & 0.896(0.083) &  0.160 \\ 
 & & Or-PPGM & 7.300(1.313) & 0.500(0.909) & 1.700(1.313) & 0.947(0.088) & 0.811(0.146) & 0.864(0.097) &  0.501 \\ 
 & & Or-LPGM & 7.580(0.906) & 0.660(0.982) & 1.420(0.906) & 0.930(0.096) & 0.842(0.101) & 0.879(0.077) &  0.064 \\ 
 & & PDN & 8.380(0.697) & 21.160(4.157) & 0.620(0.697) & 0.289(0.051) & 0.931(0.077) & 0.440(0.063) &  0.106 \\ 
 & & ODS & 3.740(0.922) & 1.880(1.206) & 5.260(0.922) & 0.683(0.181) & 0.416(0.102) & 0.511(0.118) &  3.841 \\ 
 & & PC & 5.020(0.892) & 2.400(1.010) & 3.980(0.892) & 0.681(0.103) & 0.558(0.099) & 0.610(0.093) &  0.007 \\ 
 & &  MMHC& 3.160(1.113) & 4.460(1.388) & 5.840(1.113) & 0.417(0.137) & 0.351(0.124) & 0.379(0.126) &  0.007 \\
    & &GBiDAG & 4.396(1.540) & 4.354(1.768) & 4.604(1.540) & 0.507(0.178) & 0.488(0.171) & 0.495(0.170) & 2.629 \\ 
 &&&& & & & & & \\
 &500 & PKBIC & 8.780(0.418) & 0.540(0.706) & 0.220(0.418) & 0.947(0.067) & 0.976(0.046) & 0.959(0.044) &  0.204 \\ 
 & & Or-PPGM & 8.380(0.805) & 0.360(0.598) & 0.620(0.805) & 0.962(0.064) & 0.931(0.089) & 0.944(0.062) &  0.549 \\ 
 & & Or-LPGM & 8.480(0.544) & 0.380(0.567) & 0.520(0.544) & 0.961(0.058) & 0.942(0.060) & 0.950(0.043) &  0.033 \\ 
 & & PDN & 8.540(0.579) & 14.280(1.938) & 0.460(0.579) & 0.376(0.032) & 0.949(0.064) & 0.538(0.036) &  0.116 \\ 
 & & ODS & 5.100(1.035) & 2.360(1.258) & 3.900(1.035) & 0.692(0.142) & 0.567(0.115) & 0.620(0.118) &  4.717 \\ 
 & & PC & 5.820(0.850) & 2.500(0.863) & 3.180(0.850) & 0.701(0.096) & 0.647(0.094) & 0.672(0.092) &  0.008 \\ 
 & & MMHC & 4.040(1.340) & 4.820(1.224) & 4.960(1.340) & 0.453(0.137) & 0.449(0.149) & 0.450(0.142) &  0.008 \\ 
    & &GBiDAG & 5.560(1.431) & 4.080(1.805) & 3.440(1.431) & 0.584(0.158) & 0.618(0.159) & 0.599(0.155) & 2.139 \\ 
 &&&& & & & & & \\
 &1000 & PKBIC & 8.940(0.240) & 0.360(0.598) & 0.060(0.240) & 0.965(0.057) & 0.993(0.027) & 0.978(0.035) &  0.271 \\ 
 & & Or-PPGM & 8.820(0.482) & 0.120(0.328) & 0.180(0.482) & 0.988(0.033) & 0.980(0.054) & 0.983(0.032) &  0.956 \\ 
 & & Or-LPGM & 8.840(0.370) & 0.260(0.527) & 0.160(0.370) & 0.974(0.052) & 0.982(0.041) & 0.977(0.038) &  0.065 \\ 
 & & PDN & 8.660(0.479) & 10.480(1.359) & 0.340(0.479) & 0.454(0.034) & 0.962(0.053) & 0.616(0.036) &  0.137 \\ 
 & & ODS & 5.820(1.024) & 2.840(1.283) & 3.180(1.024) & 0.679(0.126) & 0.647(0.114) & 0.660(0.112) &  5.690 \\ 
 & & PC & 6.420(0.906) & 2.360(0.776) & 2.580(0.906) & 0.730(0.088) & 0.713(0.101) & 0.721(0.093) &  0.008 \\ 
 & & MMHC & 4.960(1.087) & 4.420(0.971) & 4.040(1.087) & 0.527(0.104) & 0.551(0.121) & 0.538(0.110) &  0.010 \\ 
  & &GBiDAG & 5.740(1.367) & 3.620(1.469) & 3.260(1.367) & 0.615(0.149) & 0.638(0.152) & 0.626(0.148) & 1.577 \\
 \hline
 &&&& & & & & & \\
 &100 &PKBIC & 5.480(1.035) & 1.500(1.111) & 2.520(1.035) & 0.800(0.134) & 0.685(0.129) & 0.730(0.104) &  0.149 \\ 
 & &  Or-PPGM & 4.600(1.678) & 1.580(1.230) & 3.400(1.678) & 0.755(0.178) & 0.575(0.210) & 0.633(0.183) &  0.914 \\ 
 & &  Or-LPGM & 5.640(1.025) & 1.920(1.469) & 2.360(1.025) & 0.768(0.144) & 0.705(0.128) & 0.725(0.104) &  0.064 \\ 
 & &  PDN & 5.180(1.101) & 33.300(4.022) & 2.820(1.101) & 0.135(0.029) & 0.647(0.138) & 0.224(0.047) &  0.104 \\ 
  & & ODS & 1.449(0.614) & 0.347(0.597) & 6.551(0.614) & 0.881(0.201) & 0.181(0.077) & 0.289(0.097) &  3.551  \\ 
  & & PC & 2.920(1.104) & 2.200(1.294) & 5.080(1.104) & 0.579(0.201) & 0.365(0.138) & 0.443(0.158) &  0.007 \\ 
  & &  MMHC& 1.122(0.400) & 3.732(1.225) & 6.878(0.400) & 0.243(0.090) & 0.140(0.050) & 0.175(0.058) &  0.005\\ 
  & &GBiDAG & 3.640(1.467) & 3.540(1.798) & 4.360(1.467) & 0.516(0.211) & 0.455(0.183) & 0.479(0.188) & 2.982 \\
 
   &&&& & & & & & \\
	Hub &200 & PKBIC & 6.300(0.814) & 0.680(0.819) & 1.700(0.814) & 0.911(0.103) & 0.787(0.102) & 0.840(0.083) &  0.166 \\ 
  & & Or-PPGM & 5.580(1.444) & 0.680(0.741) & 2.420(1.444) & 0.901(0.108) & 0.698(0.181) & 0.771(0.142) &  0.510 \\ 
  & & Or-LPGM & 6.280(0.858) & 0.980(0.937) & 1.720(0.858) & 0.877(0.111) & 0.785(0.107) & 0.822(0.082) &  0.065 \\ 
  & & PDN & 5.660(1.171) & 31.240(3.679) & 2.340(1.171) & 0.154(0.030) & 0.708(0.146) & 0.252(0.049) &  0.107 \\ 
  & & ODS & 2.560(0.951) & 0.700(0.789) & 5.440(0.951) & 0.810(0.193) & 0.320(0.119) & 0.446(0.136) &  4.002 \\ 
  & & PC & 4.640(1.102) & 1.400(0.808) & 3.360(1.102) & 0.766(0.138) & 0.580(0.138) & 0.656(0.133) &  0.008 \\ 
  & &MMHC & 1.122(0.726) & 5.082(1.272) & 6.878(0.726) & 0.185(0.105) & 0.140(0.091) & 0.158(0.096) &  0.007 \\ 
  & &GBiDAG & 5.082(1.320) & 2.490(1.570) & 2.918(1.320) & 0.683(0.181) & 0.635(0.165) & 0.653(0.161) & 2.404 \\ 
   &&&& & & & & & \\
  &500 & PKBIC & 7.460(0.503) & 0.560(0.705) & 0.540(0.503) & 0.935(0.077) & 0.932(0.063) & 0.932(0.057) &  0.217 \\ 
  & & Or-PPGM & 7.240(0.744) & 0.480(0.646) & 0.760(0.744) & 0.942(0.076) & 0.905(0.093) & 0.920(0.072) &  0.549 \\ 
  & & Or-LPGM & 7.380(0.635) & 0.360(0.663) & 0.620(0.635) & 0.959(0.072) & 0.922(0.079) & 0.937(0.059) &  0.037 \\ 
  & & PDN & 5.560(1.198) & 22.960(4.755) & 2.440(1.198) & 0.200(0.045) & 0.695(0.150) & 0.305(0.053) &  0.118 \\ 
  & & ODS & 4.680(1.096) & 1.280(1.089) & 3.320(1.096) & 0.795(0.164) & 0.585(0.137) & 0.668(0.136) &  5.239 \\ 
  & & PC & 6.300(1.129) & 1.000(1.294) & 1.700(1.129) & 0.869(0.162) & 0.787(0.141) & 0.825(0.147) &  0.009 \\ 
  & & MMHC & 1.188(0.641) & 6.354(1.041) & 6.812(0.641) & 0.157(0.072) & 0.148(0.080) & 0.152(0.075) &  0.008 \\ 
    & &GBiDAG & 6.520(1.015) & 1.520(1.216) & 1.480(1.015) & 0.818(0.131) & 0.815(0.127) & 0.814(0.119) & 1.779 \\
   &&&& & & & & & \\
  &1000 & PKBIC & 7.820(0.388) & 0.340(0.519) & 0.180(0.388) & 0.962(0.058) & 0.978(0.049) & 0.968(0.041) &  0.297 \\ 
  & & Or-PPGM & 7.700(0.463) & 0.140(0.351) & 0.300(0.463) & 0.984(0.040) & 0.963(0.058) & 0.972(0.037) &  0.999 \\ 
  & & Or-LPGM & 7.780(0.418) & 0.160(0.370) & 0.220(0.418) & 0.982(0.043) & 0.973(0.052) & 0.976(0.036) &  0.073 \\ 
  & & PDN & 5.800(1.370) & 18.880(3.910) & 2.200(1.370) & 0.238(0.044) & 0.725(0.171) & 0.353(0.065) &  0.138 \\ 
  & & ODS.3 & 5.660(0.872) & 1.940(1.346) & 2.340(0.872) & 0.762(0.140) & 0.708(0.109) & 0.728(0.104) &  7.062 \\ 
  & & PC & 7.120(1.003) & 0.540(1.110) & 0.880(1.003) & 0.933(0.131) & 0.890(0.125) & 0.910(0.125) &  0.009 \\ 
  & &MMHC & 1.468(1.080) & 6.660(1.069) & 6.532(1.080) & 0.176(0.113) & 0.184(0.135) & 0.179(0.123) &  0.009\\
  & &GBiDAG & 7.060(0.793) & 1.340(0.895) & 0.940(0.793) & 0.844(0.103) & 0.882(0.099) & 0.861(0.094) & 1.548 \\
  \hline
   &&&& & & & & & \\
   &100 &PKBIC & 3.620(1.210) & 1.540(1.014) & 4.380(1.210) & 0.717(0.158) & 0.452(0.151) & 0.539(0.146) &  0.134 \\ 
  & & Or-PPGM & 3.600(1.355) & 1.500(0.974) & 4.400(1.355) & 0.709(0.170) & 0.450(0.169) & 0.536(0.165) &  0.864 \\ 
  & & Or-LPGM & 3.700(1.344) & 1.700(1.182) & 4.300(1.344) & 0.699(0.183) & 0.462(0.168) & 0.540(0.165) &  0.063 \\ 
  & & PDN & 5.700(1.015) & 36.800(4.020) & 2.300(1.015) & 0.135(0.024) & 0.713(0.127) & 0.226(0.040) &  0.104 \\ 
  & & ODS & 1.250(0.568) & 0.438(0.619) & 6.750(0.568) & 0.816(0.242) & 0.156(0.071) & 0.253(0.089) &  3.458 \\ 
  & & PC & 1.590(0.751) & 2.897(0.995) & 6.410(0.751) & 0.355(0.129) & 0.199(0.094) & 0.251(0.106) &  0.006  \\ 
  & & MMHC &  1.894(0.699) & 2.043(1.122) & 6.106(0.699) & 0.516(0.208) & 0.237(0.087) & 0.315(0.105) &  0.005  \\ 
& &GBiDAG & 1.947(0.899) & 3.763(1.807) & 6.053(0.899) & 0.360(0.170) & 0.243(0.112) & 0.284(0.124) & 2.585 \\

   &&&& & & & & & \\
Erdos-Renyi  &200 & PKBIC & 4.740(1.226) & 0.740(0.944) & 3.260(1.226) & 0.879(0.141) & 0.593(0.153) & 0.696(0.132) &  0.140 \\ 
  & & Or-PPGM & 4.680(1.253) & 0.860(0.990) & 3.320(1.253) & 0.864(0.150) & 0.585(0.157) & 0.683(0.133) &  0.457 \\ 
  & & Or-LPGM & 4.660(1.206) & 0.960(1.195) & 3.340(1.206) & 0.853(0.166) & 0.583(0.151) & 0.679(0.134) &  0.065 \\ 
  & & PDN & 6.420(1.012) & 34.060(4.688) & 1.580(1.012) & 0.160(0.029) & 0.802(0.126) & 0.267(0.045) &  0.108 \\ 
  & & ODS & 1.568(0.695) & 1.114(0.993) & 6.432(0.695) & 0.639(0.280) & 0.196(0.087) & 0.293(0.120) &  3.919 \\ 
 & & PC & 2.356(1.282) & 2.222(0.974) & 5.644(1.282) & 0.499(0.193) & 0.294(0.160) & 0.364(0.174) &  0.006\\ 
  & & MMHC & 2.680(1.115) & 2.060(1.185) & 5.320(1.115) & 0.571(0.191) & 0.335(0.139) & 0.416(0.155) &  0.006 \\ 
 & &  GBiDAG & 2.894(1.088) & 3.277(1.514) & 5.106(1.088) & 0.482(0.177) & 0.362(0.136) & 0.408(0.145) & 2.125 \\
   &&&& & & & & & \\
  &500 & PKBIC & 6.540(0.838) & 0.460(0.579) & 1.460(0.838) & 0.938(0.077) & 0.818(0.105) & 0.870(0.077) &  0.163 \\ 
  & & Or-PPGM & 6.400(0.881) & 0.460(0.613) & 1.600(0.881) & 0.938(0.081) & 0.800(0.110) & 0.859(0.080) &  0.543 \\ 
  & & Or-LPGM & 6.300(0.789) & 0.420(0.609) & 1.700(0.789) & 0.944(0.079) & 0.787(0.099) & 0.854(0.069) &  0.032 \\ 
  & & PDN & 7.100(0.735) & 29.300(3.066) & 0.900(0.735) & 0.196(0.025) & 0.887(0.092) & 0.321(0.039) &  0.119 \\ 
 & &  ODS & 3.260(1.275) & 1.800(1.262) & 4.740(1.275) & 0.655(0.209) & 0.408(0.159) & 0.493(0.168) &  4.924 \\ 
  & & PC & 4.204(1.399) & 2.020(0.989) & 3.796(1.399) & 0.666(0.174) & 0.526(0.175) & 0.585(0.173) &  0.007  \\ 
  & & MMHC & 4.700(1.055) & 2.000(1.178) & 3.300(1.055) & 0.710(0.140) & 0.588(0.132) & 0.638(0.126) &  0.007 \\ 
 & & GBiDAG & 3.917(1.381) & 3.208(1.856) & 4.083(1.381) & 0.564(0.212) & 0.490(0.173) & 0.521(0.187) & 1.811 \\ 
   &&&& & & & & & \\
  &1000 & PKBIC & 7.520(0.580) & 0.360(0.598) & 0.480(0.580) & 0.960(0.065) & 0.940(0.072) & 0.947(0.046) &  0.216 \\ 
  & & Or-PPGM & 7.200(0.782) & 0.280(0.454) & 0.800(0.782) & 0.965(0.057) & 0.900(0.098) & 0.928(0.065) &  0.924 \\ 
  & & Or-LPGM & 7.020(0.820) & 0.120(0.385) & 0.980(0.820) & 0.985(0.048) & 0.877(0.103) & 0.925(0.064) &  0.057 \\ 
  & & PDN & 7.460(0.676) & 24.420(2.417) & 0.540(0.676) & 0.235(0.024) & 0.932(0.085) & 0.375(0.035) &  0.139 \\ 
  & & ODS & 4.160(1.315) & 2.800(1.414) & 3.840(1.315) & 0.602(0.186) & 0.520(0.164) & 0.556(0.169) &  6.144 \\ 
  & & PC & 5.160(1.095) & 1.640(0.631) & 2.840(1.095) & 0.753(0.111) & 0.645(0.137) & 0.693(0.124) &  0.008 \\ 
  & & MMHC & 5.420(0.992) & 2.040(1.029) & 2.580(0.992) & 0.729(0.129) & 0.677(0.124) & 0.701(0.121) &  0.009 \\ 
 & & GBiDAG & 4.571(1.720) & 3.367(1.856) & 3.429(1.720) & 0.581(0.224) & 0.571(0.215) & 0.574(0.216) & 1.528 \\ 
			\bottomrule
		\end{longtable}
	
\end{center}

\begin{center}
	\begin{scriptsize}
	\fontsize{6}{8}\selectfont	
		\begin{longtable}{l| l | l r r r r r r r}
			\caption{Simulation results from 50 replicates of the DAGs shown in Figure 2 in of the main paper  for $p= 100$ variables with Poisson node conditional distribution. Monte Carlo means (standard deviations) are shown for $TP, FP, FN, P, R$, and $F_1$. The levels of significance of tests $\alpha=2(1-\Phi(n^{0.2}))$ for $n=500,1000, 2000$, and $\alpha=2(1-\Phi(n^{0.225}))$ for $n=200$. }  \\
			\label{table100-chap2}\\
			\toprule
			Graph&$n$	& Algorithm & $TP$ & $FP$ & $FN$ & $P$ & $R$ &$F_1$&time \\
			\midrule
			\endfirsthead
			\multicolumn{9}{c}%
			{{\bfseries \tablename\ \thetable{} -- continued from previous page}} \\
			\toprule
			Graph&$n$	& Algorithm & $TP$ & $FP$ & $FN$ & $P$ & $R$ &$F_1$& time \\
			\midrule	
			
			\endhead
			&200 &PKBIC & 72.800(3.597) & 74.360(6.608) & 26.200(3.597) & 0.495(0.029) & 0.735(0.036) & 0.592(0.030) &   2.966 \\ 
  &&Or-PPGM & 53.080(2.947) & 5.880(2.336) & 45.920(2.947) & 0.901(0.036) & 0.536(0.030) & 0.672(0.027) &   3.062 \\ 
  &&Or-LPGM & 42.740(3.096) & 7.040(9.689) & 56.260(3.096) & 0.875(0.085) & 0.432(0.031) & 0.575(0.035) &   0.136 \\ 
  &&PDN & 49.220(15.953) & 277.140(163.948) & 49.780(15.953) & 0.241(0.198) & 0.497(0.161) & 0.244(0.079) &   3.876 \\ 
  &&ODS & 33.800(5.051) & 106.700(31.538) & 65.200(5.051) & 0.250(0.046) & 0.341(0.051) & 0.284(0.033) & 108.155 \\ 
  &&PC & 24.280(2.458) & 20.980(3.191) & 74.720(2.458) & 0.538(0.054) & 0.245(0.025) & 0.337(0.032) &   0.097 \\ 
  &&MMHC & 35.400(4.041) & 95.500(6.935) & 63.600(4.041) & 0.271(0.032) & 0.358(0.041) & 0.308(0.035) &   0.367 \\ 
  &&GBiDAG & 47.760(4.693) & 166.440(13.849) & 51.240(4.693) & 0.224(0.026) & 0.482(0.047) & 0.306(0.033) &  155.839 \\ 
  &&&& & & & & & \\
  Scale-free&500&PKBIC & 86.980(1.890) & 45.420(6.970) & 12.020(1.890) & 0.659(0.037) & 0.879(0.019) & 0.752(0.027) &   4.237 \\ 
  &&Or-PPGM & 74.900(3.604) & 3.540(1.705) & 24.100(3.604) & 0.955(0.021) & 0.757(0.036) & 0.844(0.025) &   3.944 \\ 
  &&Or-LPGM & 74.560(2.822) & 2.200(1.604) & 24.440(2.822) & 0.972(0.019) & 0.753(0.029) & 0.848(0.019) &   0.236 \\ 
  &&PDN & 48.520(23.784) & 111.880(87.858) & 50.480(23.784) & 0.417(0.202) & 0.490(0.240) & 0.351(0.134) &   4.831 \\ 
  &&ODS & 45.740(5.938) & 94.840(32.950) & 53.260(5.938) & 0.340(0.062) & 0.462(0.060) & 0.385(0.039) & 195.076 \\ 
  &&PC & 39.940(3.158) & 28.800(3.090) & 59.060(3.158) & 0.581(0.040) & 0.403(0.032) & 0.476(0.034) &   0.123 \\ 
  &&MMHC & 55.220(5.152) & 86.880(7.634) & 43.780(5.152) & 0.389(0.037) & 0.558(0.052) & 0.458(0.042) &   0.530 \\ 
  &&GBiDAG & 67.380(4.989) & 99.400(9.802) & 31.620(4.989) & 0.405(0.035) & 0.681(0.050) & 0.507(0.039) &   85.237 \\
  &&&& & & & & & \\
  &1000&PKBIC & 92.265(0.995) & 32.755(5.329) & 6.735(0.995) & 0.739(0.032) & 0.932(0.010) & 0.824(0.020) &   6.558 \\ 
  &&Or-PPGM & 83.286(3.234) & 1.184(1.034) & 15.714(3.234) & 0.986(0.012) & 0.841(0.033) & 0.908(0.020) &   5.066 \\ 
  &&Or-LPGM & 85.600(2.268) & 0.240(0.476) & 13.400(2.268) & 0.997(0.005) & 0.865(0.023) & 0.926(0.013) &   0.398 \\ 
  &&PDN & 45.980(26.072) & 63.160(52.492) & 53.020(26.072) & 0.538(0.197) & 0.464(0.263) & 0.399(0.162) &   6.473 \\ 
  &&ODS & 51.580(12.993) & 84.300(31.210) & 47.420(12.993) & 0.404(0.090) & 0.521(0.131) & 0.431(0.083) & 307.384 \\ 
  &&PC & 46.880(2.862) & 32.740(2.841) & 52.120(2.862) & 0.589(0.031) & 0.474(0.029) & 0.525(0.028) &   0.152 \\ 
  &&MMHC & 66.760(4.770) & 73.200(7.157) & 32.240(4.770) & 0.478(0.040) & 0.674(0.048) & 0.559(0.043) &   0.825 \\ 
  &&GBiDAG & 77.080(3.757) & 72.540(8.370) & 21.920(3.757) & 0.516(0.036) & 0.779(0.038) & 0.621(0.036) &   76.097 \\
  &&&& & & & & & \\
  &2000&PKBIC & 93.840(0.370) & 22.080(4.575) & 5.160(0.370) & 0.811(0.033) & 0.948(0.004) & 0.874(0.019) &  11.046 \\ 
  &&Or-PPGM & 89.120(2.086) & 1.500(0.735) & 9.880(2.086) & 0.984(0.008) & 0.900(0.021) & 0.940(0.012) &   9.026 \\ 
  &&Or-LPGM & 93.760(1.697) & 0.020(0.141) & 5.240(1.697) & 1.000(0.002) & 0.947(0.017) & 0.973(0.009) &   0.754 \\ 
  &&PDN & 34.540(23.652) & 33.940(28.183) & 64.460(23.652) & 0.594(0.160) & 0.349(0.239) & 0.360(0.189) &  10.075 \\ 
  &&ODS & 61.840(13.864) & 95.280(32.128) & 37.160(13.864) & 0.403(0.074) & 0.625(0.140) & 0.477(0.088) & 516.160 \\ 
  &&PC & 54.200(2.563) & 35.640(2.855) & 44.800(2.563) & 0.603(0.029) & 0.547(0.026) & 0.574(0.027) &   0.195 \\ 
  &&MMHC & 73.900(4.082) & 60.580(7.535) & 25.100(4.082) & 0.551(0.042) & 0.746(0.041) & 0.634(0.041) &   1.396 \\ 
  &&GBiDAG & 81.760(2.945) & 55.420(6.459) & 17.240(2.945) & 0.597(0.032) & 0.826(0.030) & 0.693(0.029) &   60.588 \\ 
\hline
&&&& & & & & & \\
&200&PKBIC & 44.540(3.759) & 88.800(7.635) & 50.460(3.759) & 0.335(0.028) & 0.469(0.040) & 0.390(0.030) &    2.848 \\ 
 && Or-PPGM & 19.000(3.608) & 5.180(2.396) & 76.000(3.608) & 0.787(0.096) & 0.200(0.038) & 0.318(0.052) &    2.658 \\ 
 && Or-LPGM & 11.080(2.448) & 5.920(2.687) & 83.920(2.448) & 0.659(0.116) & 0.117(0.026) & 0.197(0.040) &    0.146 \\ 
 && PDN & 40.540(4.807) & 723.700(69.012) & 54.460(4.807) & 0.054(0.009) & 0.427(0.051) & 0.095(0.013) &    3.896 \\ 
 && ODS & 8.540(3.981) & 60.700(30.999) & 86.460(3.981) & 0.138(0.073) & 0.090(0.042) & 0.104(0.046) &   86.450 \\ 
 && PC & 2.000(1.155) & 13.326(2.504) & 93.000(1.155) & 0.131(0.068) & 0.021(0.012) & 0.036(0.020) &    0.091 \\ 
 && MMHC & 18.720(3.226) & 99.040(7.546) & 76.280(3.226) & 0.159(0.029) & 0.197(0.034) & 0.176(0.031) &    0.389 \\ 
   &&GBiDAG&27.460(3.991) & 168.860(9.928) & 67.540(3.991) & 0.140(0.019) & 0.289(0.042) & 0.188(0.025) &   164.755\\
  &&&& & & & & & \\
 Hub&500& PKBIC & 66.920(2.884) & 55.360(5.713) & 28.080(2.884) & 0.548(0.025) & 0.704(0.030) & 0.616(0.021) &    3.967 \\ 
 && Or-PPGM & 49.120(4.094) & 3.040(1.340) & 45.880(4.094) & 0.942(0.024) & 0.517(0.043) & 0.667(0.037) &    3.103 \\ 
 && Or-LPGM & 37.500(4.735) & 2.700(1.632) & 57.500(4.735) & 0.934(0.039) & 0.395(0.050) & 0.553(0.051) &    0.238 \\ 
 && PDN & 60.980(3.396) & 653.320(42.419) & 34.020(3.396) & 0.086(0.005) & 0.642(0.036) & 0.151(0.008) &    4.852 \\ 
 && ODS & 18.340(8.178) & 72.520(28.367) & 76.660(8.178) & 0.211(0.101) & 0.193(0.086) & 0.194(0.079) &  145.684 \\ 
 && PC & 5.000(1.539) & 24.060(2.189) & 90.000(1.539) & 0.171(0.049) & 0.053(0.016) & 0.080(0.024) &    0.350 \\ 
 && MMHC & 37.260(4.580) & 84.720(10.236) & 57.740(4.580) & 0.307(0.046) & 0.392(0.048) & 0.344(0.046) &    9.120 \\ 
  &&GBiDAG& 48.700(5.179) & 100.980(10.129) & 46.300(5.179) & 0.326(0.038) & 0.513(0.055) & 0.398(0.044) &   565.562 \\
  &&&& & & & & & \\
 &1000& PKBIC & 75.260(1.291) & 38.440(5.588) & 19.740(1.291) & 0.663(0.034) & 0.792(0.014) & 0.722(0.023) &    5.987 \\ 
 && Or-PPGM & 67.760(3.503) & 0.320(0.513) & 27.240(3.503) & 0.995(0.007) & 0.713(0.037) & 0.830(0.026) &    4.147 \\ 
 && Or-LPGM & 57.420(3.844) & 0.240(0.431) & 37.580(3.844) & 0.996(0.008) & 0.604(0.040) & 0.751(0.032) &    0.402 \\ 
 && PDN & 69.200(2.441) & 546.600(20.849) & 25.800(2.441) & 0.112(0.005) & 0.728(0.026) & 0.195(0.008) &    6.503 \\ 
 && ODS & 26.280(9.311) & 92.760(12.697) & 68.720(9.311) & 0.221(0.076) & 0.277(0.098) & 0.245(0.085) &  232.417 \\ 
 && PC & 7.760(1.533) & 34.260(2.221) & 87.240(1.533) & 0.185(0.036) & 0.082(0.016) & 0.113(0.022) &    4.832 \\ 
 && MMHC & 53.740(4.931) & 71.900(9.429) & 41.260(4.931) & 0.429(0.047) & 0.566(0.052) & 0.488(0.047) &  358.392 \\ 
  &&GBiDAG& 60.440(4.820) & 76.940(9.801) & 34.560(4.820) & 0.442(0.044) & 0.636(0.051) & 0.521(0.045) &  4675.011 \\ 
  &&&& & & & & & \\
 &2000& PKBIC & 77.520(0.544) & 26.240(4.897) & 17.480(0.544) & 0.749(0.034) & 0.816(0.006) & 0.780(0.019) &    9.937 \\ 
 && Or-PPGM & 80.800(2.755) & 0.040(0.198) & 14.200(2.755) & 0.999(0.002) & 0.851(0.029) & 0.919(0.017) &    9.954 \\ 
 && Or-LPGM & 76.020(2.317) & 0.000(0.000) & 18.980(2.317) & 1.000(0.000) & 0.800(0.024) & 0.889(0.015) &    0.754 \\ 
 && PDN & 71.100(2.188) & 421.220(30.434) & 23.900(2.188) & 0.145(0.009) & 0.748(0.023) & 0.243(0.013) &   10.005 \\ 
 && ODS & 36.020(11.302) & 82.320(12.173) & 58.980(11.302) & 0.304(0.093) & 0.379(0.119) & 0.337(0.104) &  397.424 \\ 
 && PC & 12.460(2.111) & 46.800(2.304) & 82.540(2.111) & 0.210(0.029) & 0.131(0.022) & 0.161(0.025) &   86.424 \\ 
 && MMHC & 66.000(7.323) & 57.440(10.459) & 29.000(7.323) & 0.537(0.069) & 0.695(0.077) & 0.605(0.072) & 3351.923 \\ 
  &&GBiDAG& 73.380(4.772) & 54.580(8.013) & 21.620(4.772) & 0.575(0.049) & 0.772(0.050) & 0.659(0.049) & 13866.140 \\ 
 \hline
 &&&& & & & & & \\
  &200&PKBIC & 63.320(3.554) & 80.180(6.998) & 45.680(3.554) & 0.442(0.028) & 0.581(0.033) & 0.502(0.028) &   3.024 \\ 
   &&Or-PPGM & 33.840(4.063) & 4.900(2.063) & 75.160(4.063) & 0.875(0.046) & 0.310(0.037) & 0.457(0.042) &   2.760 \\ 
 &&  Or-LPGM & 24.500(2.589) & 6.060(2.385) & 84.500(2.589) & 0.805(0.064) & 0.225(0.024) & 0.351(0.031) &   0.140 \\ 
 &&  PDN & 61.980(4.354) & 641.620(23.913) & 47.020(4.354) & 0.088(0.007) & 0.569(0.040) & 0.153(0.012) &   3.866 \\ 
 &&  ODS & 27.900(6.112) & 108.440(32.255) & 81.100(6.112) & 0.215(0.054) & 0.256(0.056) & 0.227(0.035) &  91.828 \\ 
 &&  PC & 15.180(3.330) & 15.040(3.201) & 93.820(3.330) & 0.502(0.090) & 0.139(0.031) & 0.217(0.045) &   0.089 \\ 
  && MMHC & 28.460(3.877) & 102.080(7.331) & 80.540(3.877) & 0.218(0.029) & 0.261(0.036) & 0.238(0.031) &   0.357 \\
  &&GBiDAG & 34.440(4.630) & 167.920(11.738) & 74.560(4.630) & 0.171(0.023) & 0.316(0.042) & 0.221(0.030) & 140.079 \\ 
  &&&& & & & & & \\
Erdos-Renyi&500&  PKBIC & 88.900(3.144) & 47.780(5.144) & 20.100(3.144) & 0.651(0.025) & 0.816(0.029) & 0.724(0.022) &   4.321 \\ 
 &&  Or-PPGM & 65.520(4.137) & 3.420(1.679) & 43.480(4.137) & 0.951(0.023) & 0.601(0.038) & 0.736(0.029) &   3.224 \\ 
 &&  Or-LPGM & 64.380(3.833) & 2.580(1.527) & 44.620(3.833) & 0.962(0.022) & 0.591(0.035) & 0.731(0.029) &   0.239 \\ 
 &&  PDN & 84.820(3.652) & 494.120(20.495) & 24.180(3.652) & 0.147(0.008) & 0.778(0.034) & 0.247(0.013) &   4.825 \\ 
  && ODS & 44.040(4.458) & 84.660(10.809) & 64.960(4.458) & 0.344(0.038) & 0.404(0.041) & 0.371(0.036) & 158.765 \\ 
 &&  PC & 35.140(3.175) & 25.520(2.816) & 73.860(3.175) & 0.579(0.039) & 0.322(0.029) & 0.414(0.032) &   0.100 \\ 
 &&  MMHC & 51.500(5.632) & 96.420(9.476) & 57.500(5.632) & 0.349(0.040) & 0.472(0.052) & 0.401(0.043) &   0.576 \\
 &&GBiDAG & 57.240(5.605) & 110.220(9.599) & 51.760(5.605) & 0.342(0.035) & 0.525(0.051) & 0.414(0.041) &  79.766 \\
  &&&& & & & & & \\
 &1000&  PKBIC & 99.540(1.460) & 32.820(5.336) & 9.460(1.460) & 0.753(0.031) & 0.913(0.013) & 0.825(0.020) &   6.734 \\ 
 &&  Or-PPGM & 82.040(3.084) & 2.120(1.043) & 26.960(3.084) & 0.975(0.012) & 0.753(0.028) & 0.849(0.019) &   4.152 \\ 
 &&  Or-LPGM & 85.440(3.308) & 0.120(0.328) & 23.560(3.308) & 0.999(0.004) & 0.784(0.030) & 0.878(0.019) &   0.416 \\ 
 &&  PDN & 93.960(2.725) & 361.620(18.856) & 15.040(2.725) & 0.207(0.011) & 0.862(0.025) & 0.333(0.015) &   6.445 \\ 
 &&  ODS & 58.380(5.810) & 74.780(8.115) & 50.620(5.810) & 0.439(0.045) & 0.536(0.053) & 0.482(0.048) & 243.041 \\ 
 &&  PC & 47.620(2.996) & 31.480(3.032) & 61.380(2.996) & 0.602(0.034) & 0.437(0.027) & 0.506(0.029) &   0.117 \\
&&  MMHC & 66.840(4.892) & 79.260(7.537) & 42.160(4.892) & 0.458(0.037) & 0.613(0.045) & 0.524(0.039) &   0.943 \\ 
&&GBiDAG & 70.000(6.214) & 82.060(10.015) & 39.000(6.214) & 0.462(0.047) & 0.642(0.057) & 0.537(0.050) &  65.951 \\
  &&&& & & & & & \\
 &2000&  PKBIC & 104.140(0.783) & 22.820(4.801) & 4.860(0.783) & 0.821(0.031) & 0.955(0.007) & 0.883(0.019) &  11.455 \\ 
  && Or-PPGM & 91.900(2.323) & 2.940(1.361) & 17.100(2.323) & 0.969(0.014) & 0.843(0.021) & 0.902(0.014) &   6.229 \\ 
 &&  Or-LPGM & 98.020(2.369) & 0.040(0.198) & 10.980(2.369) & 1.000(0.002) & 0.899(0.022) & 0.947(0.012) &   0.774 \\ 
  && PDN & 97.560(1.864) & 257.700(11.438) & 11.440(1.864) & 0.275(0.010) & 0.895(0.017) & 0.421(0.012) &   9.936 \\ 
 &&  ODS & 66.440(4.625) & 70.460(7.276) & 42.560(4.625) & 0.486(0.039) & 0.610(0.042) & 0.541(0.040) & 385.807 \\ 
 &&  PC & 57.540(3.303) & 35.100(2.150) & 51.460(3.303) & 0.621(0.025) & 0.528(0.030) & 0.571(0.027) &   0.141 \\ 
 &&  MMHC & 76.400(4.412) & 64.420(7.680) & 32.600(4.412) & 0.544(0.040) & 0.701(0.040) & 0.612(0.039) &   1.797 \\ 
 &&GBiDAG & 80.120(4.893) & 61.940(6.885) & 28.880(4.893) & 0.565(0.039) & 0.735(0.045) & 0.638(0.040) &  64.460 \\
			
			\bottomrule
			
		\end{longtable}
	\end{scriptsize}
	
\end{center}

			\begin{table}[ht]
				\centering
    \fontsize{6}{8}\selectfont	
				{\scriptsize
					\caption{\label{table_spare}\scriptsize{ Monte Carlo  means of TP, PPV, Se and runtime obtained by simulating 50 samples 
							from Erdos-Renyigraphs with $p = 10$ and $p=100$ variables with Poisson node conditional distribution. The probability  of 
							edge inclusion $\pi$ runs from 0.4 to 0.9 for $p=10$, and from 0.04 to 0.09 for $p=100$. The levels of significance of tests $\alpha=0.05$ for $p=100$, and $\alpha=0.01$ for $n=100$. }}
					\begin{tabular}{c|c|rrrr|rrrr|rrrr}
						\hline
				&&\multicolumn{4}{c|} {\bf PKBIC}&\multicolumn{4}{c|}{\bf Or-PPGM }&\multicolumn{4}{c}{\bf Or-LPGM }\\
					$p$ & $\pi$ & TP  & PPV & Se & time & TP  & PPV & Se & time & TP  & PPV & Se & time \\ 
						\hline
10&0.4& 0.98& 0.93& 0.95& 0.30& 0.96& 0.92& 0.94& 0.97& 0.90& 0.96& 0.93& 0.07\\
&0.5& 0.99& 0.96& 0.98& 0.32& 0.96& 0.88& 0.92& 1.48& 0.95& 0.98& 0.97& 0.07\\
&0.6& 0.99& 0.94& 0.96& 0.43& 0.95& 0.72& 0.82& 1.15& 0.95& 0.97& 0.96& 0.08\\
&0.7& 0.99& 0.93& 0.96& 0.35& 0.96& 0.77& 0.85& 1.26& 0.98& 0.98& 0.98& 0.07\\
&0.8& 0.99& 0.83& 0.90& 0.37& 0.97& 0.62& 0.76& 1.09& 0.99& 0.87& 0.92& 0.08\\
&0.9& 1.00& 0.83& 0.90& 0.42& 0.99& 0.61& 0.75& 1.11& 0.99& 0.91& 0.95& 0.07 \\ 
&&&&&&&&&&&\\
100&0.04&0.87&  0.98&  0.92& 26.86&  0.87&  0.93&  0.90&  8.75&  0.79&  0.98&  0.87&  0.75\\
&0.05&0.90&  0.97&  0.93& 34.94&  0.91&  0.90&  0.90& 16.05&  0.83&  0.97&  0.90&  0.77\\
&0.06& 0.92&  0.94&  0.93& 37.96&  0.93&  0.94&  0.93& 90.70&  0.87&  0.96 & 0.91&  0.77\\
&0.07&0.94 & 0.95 & 0.94& 43.16 & 0.95 & 0.90 & 0.92& 41.51 & 0.88 & 0.95 & 0.91 & 0.77\\
&0.08& 0.94 &  0.93 &  0.93 & 42.54 &  0.95 &  0.90 &  0.92& 102.48  & 0.89 &  0.94 & 0.91 &  0.76\\
&0.09& 0.95  & 0.91 &  0.93 & 46.43 &  0.96 &  0.91  & 0.93& 213.32 &  0.90  & 0.95 & 0.93 &  0.76\\

						\hline
					\end{tabular}
				}
			\end{table}
\subsection{Figure}
Figure \ref{PPVSe-10}, and Figure \ref{PPVSe-100} plot Precision $ P,$ and Recall $ R$ for each of  methods considered in Section 5 of the main paper and for the three types of networks. Two different graph dimensions, i.e., $p=10, 100$, and three graph structures (see Figure 1 and Figure 2 in the main paper) are considered.
\begin{figure}[htbp]
	\begin{center}
		\includegraphics[width = 0.9\linewidth, height=0.7\textheight]{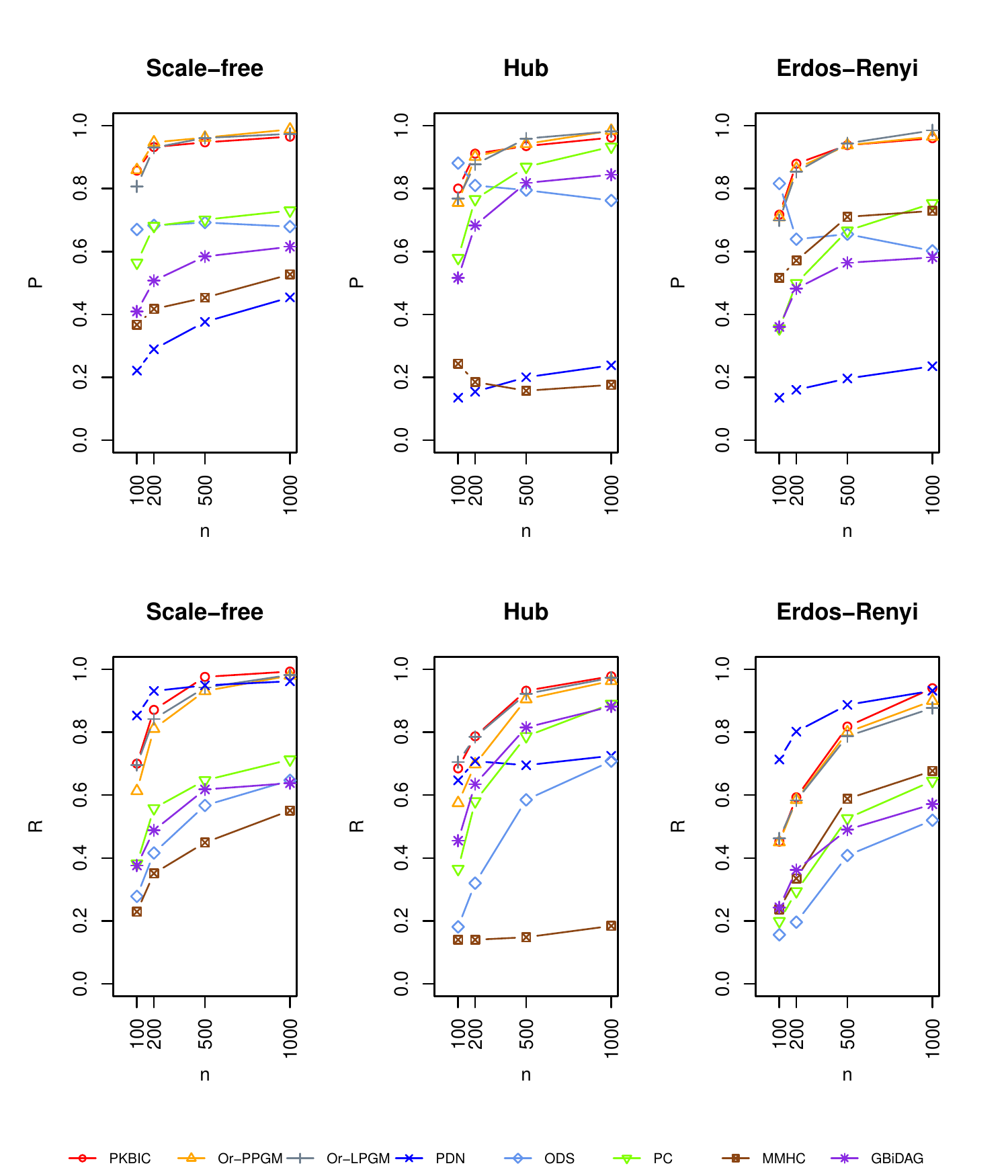}
		\caption{\scriptsize Precision $ P,$ and Recall $ R$ of the considered algorithms: PKBIC; Or-PPGM; Or-LPGM; PDN; ODS; MMHC;  PC; and BiDAG for the three types of graphs in Figure 1 of the main paper with $p=10$ and sample sizes $n=100,200,500,1000.$}
		\label{PPVSe-10}
	\end{center}
\end{figure}

\begin{figure}[htbp]
	\begin{center}
		\includegraphics[width = 0.9\linewidth, height=0.7\textheight]{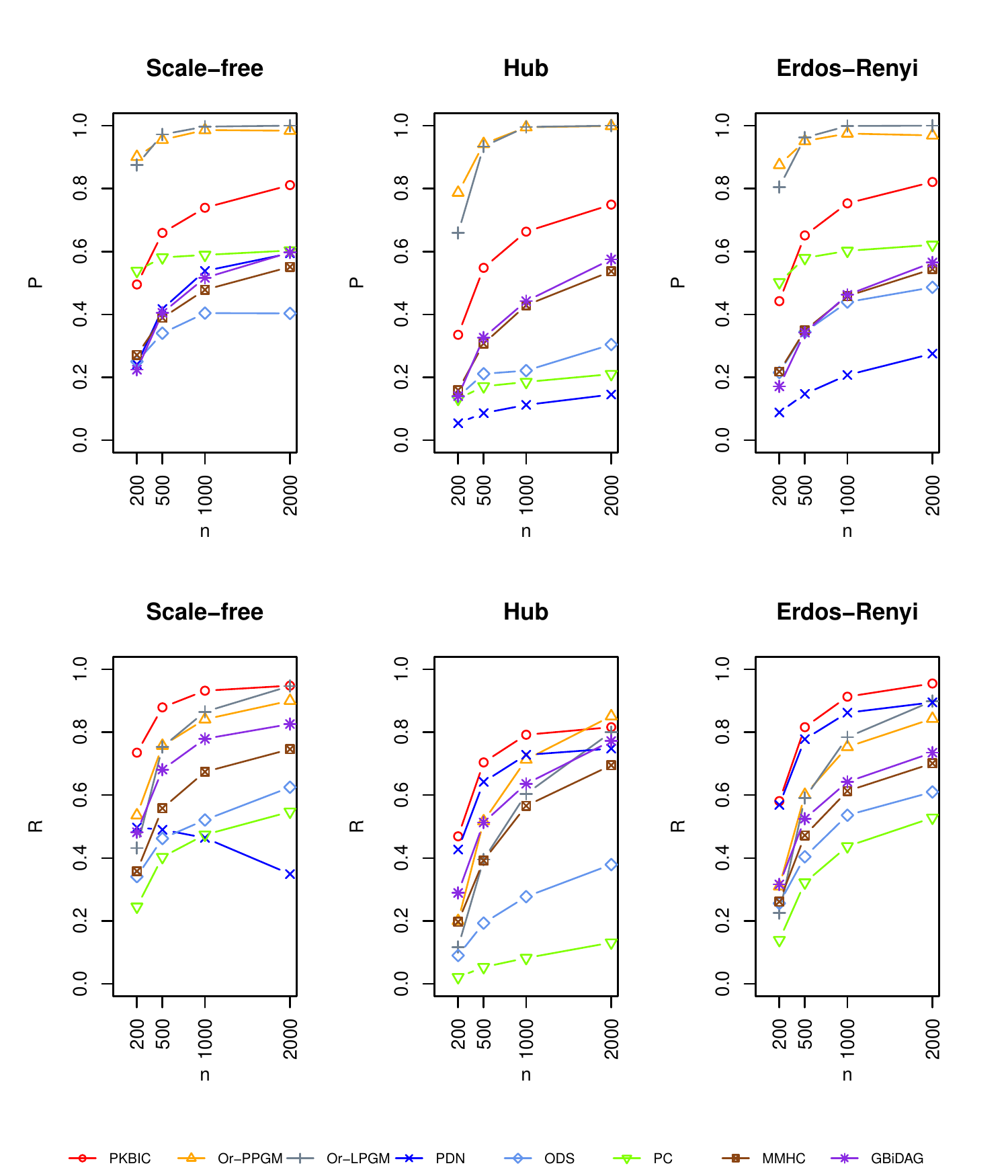}
		\caption{\scriptsize Precision $ P,$ and Recall $ R$ of the considered algorithms: PKBIC; Or-PPGM; Or-LPGM; PDN; ODS; MMHC;  PC; and BiDAG for the three types of graphs in Figure 2 of the main paper with $p=100$ and sample sizes $n=200,500,1000,2000.$}
	\label{PPVSe-100}
	\end{center}
\end{figure}

\newpage
\bibliographystyle{plainnat}

\bibliography{Supplementary}